\documentclass[11pt]{article}
\usepackage[utf8]{inputenc}
\usepackage[english]{babel}
\usepackage[T1]{fontenc}
\usepackage[margin=1.15in]{geometry}
\usepackage{amsmath}
\usepackage{amsthm}
\usepackage{amsfonts}
\usepackage{graphicx}
\usepackage{float}
\usepackage{pgfplots}
\pgfplotsset{compat=1.14, set layers}
\usepackage{bm}
\usepackage{mathtools}
\usepackage{subfig}
\usepackage{enumitem}

\usepackage{lipsum}
\usepackage{setspace}
\usepackage{caption}

\usepackage{mathtools}
\usepackage{todonotes}
\usepackage{xspace}
\usepackage[bottom]{footmisc}
\usepackage{framed}
\usepackage{thm-restate}

\definecolor{refcolor}{rgb}{0.23, 0.27, 0.29}
\usepackage[breaklinks,colorlinks,urlcolor={refcolor},citecolor={refcolor},linkcolor={refcolor}]{hyperref}
\usepackage{cleveref}

\newtheorem{theorem}{Theorem}[section]
\newtheorem{lemma}[theorem]{Lemma}

\newtheorem{corollary}[theorem]{Corollary}
\newtheorem{definition}[theorem]{Definition}
\newtheorem{remark}[theorem]{Remark}

\newtheorem{observation}[theorem]{Observation}

\newcommand{\arr}{\xrightarrow{}}
\newcommand{\larr}{\xleftarrow{}}
\newcommand{\fA}{\mathcal A}
\newcommand{\fC}{\mathcal C}
\newcommand{\fG}{\mathcal G}
\newcommand{\fP}{\mathcal P}
\newcommand{\ps}{\mathcal P}
\newcommand{\undec}{\texttt{undecided}}
\newcommand{\incid}{\operatorname{inc}}
\newcommand{\aps}{\operatorname{act}}
\newcommand{\pairs}{\operatorname{pairs}}
\newcommand{\pred}{\operatorname{pred}}
\newcommand{\first}{\operatorname{first}}
\newcommand{\last}{\operatorname{last}}
\newcommand{\comp}{M}
\newcommand{\merge}{\operatorname{merge}}
\newcommand{\prev}{\operatorname{prev}}
\newcommand{\new}{\operatorname{new}}

\newcommand{\stays}{\operatorname{stay}}
\newcommand{\timey}{\operatorname{time}}
\newcommand{\fin}{\operatorname{fin}}

\newcommand{\betw}{\operatorname{between}}
\newcommand{\done}{\operatorname{done}}
\newcommand{\succc}{\operatorname{succ}}
\newcommand{\aone}{\mathcal A'_{\operatorname{I}}}
\newcommand{\atwo}{\mathcal A'_{\operatorname{II}}}
\newcommand{\lovasz}{Lov\'{a}sz\xspace}

\newcommand{\inn}{{\operatorname{in}}}
\newcommand{\out}{{\operatorname{out}}}
\newcommand{\sinn}{\Sigma_{\inn}}
\newcommand{\sout}{\Sigma_{\out}}
\newcommand{\phinn}{\phi_{\inn}}
\newcommand{\phout}{\phi_{\out}}
\newcommand{\ginn}{g_{\inn}}
\newcommand{\gout}{g_{\out}}
\newcommand{\pinn}{p_{\inn}}
\newcommand{\pout}{p_{\out}}
\newcommand{\noco}{\mathcal N}
\newcommand{\edco}{\mathcal E}
\newcommand{\gee}{g}
\newcommand{\rooot}{r}

\newcommand{\mpc}{\textsf{MPC}\xspace}

\newcommand{\lcl}{\textsf{LCL}\xspace}
\newcommand{\lcls}{\textsf{LCL}s\xspace}
\newcommand{\local}{\textsf{LOCAL}\xspace}

\newcommand{\congest}{\textsf{CONGEST}\xspace}
\newcommand{\llle}{\textsf{LLL}\xspace}

\newcommand{\class}{\textsf{CLASS\xspace}}

\newcommand{\poly}{\operatorname{\text{{\rm poly}}}}
\newcommand{\logstar}[1]{\log^{*} #1}

\newif\ifdraft
\drafttrue

\Crefname{paragraph}{Paragraph}{Paragraphs}

\begin{document}

\begin{center}
    {\huge \bf Exponential Speedup Over Locality in \mpc with Optimal Memory} \\
\vspace{1cm}

\begin{minipage}[H]{17cm} 
{\large \textbf{Alkida Balliu}, Gran Sasso Science Institute -- \href{mailto:alkida.balliu@gssi.it}{\texttt{alkida.balliu@gssi.it}}} \vspace{0.5mm}\\
{\large \textbf{Sebastian Brandt}, CISPA Helmholtz Center for Information Security -- \href{mailto:brandt@cispa.de}{\texttt{brandt@cispa.de}}} \vspace{0.5mm}\\
{\large \textbf{Manuela Fischer}, ETH Zurich -- \href{mailto:manuela.fischer@inf.ethz.ch}{\texttt{manuela.fischer@inf.ethz.ch}}} \vspace{0.5mm}\\
{\large \textbf{Rustam Latypov\footnotemark}, Aalto University -- \href{mailto:rustam.latypov@aalto.fi}{\texttt{rustam.latypov@aalto.fi}}} \vspace{0.5mm}\\
{\large \textbf{Yannic Maus}, TU Graz -- \href{mailto:yannic.maus@ist.tugraz.at}{\texttt{yannic.maus@ist.tugraz.at}}} \vspace{1mm}\\
{\large \textbf{Dennis Olivetti}, Gran Sasso Science Institute -- \href{mailto:dennis.olivetti@gssi.it}{\texttt{dennis.olivetti@gssi.it}}} \vspace{0.5mm}\\
{\large \textbf{Jara Uitto}, Aalto University -- \href{mailto:jara.uitto@aalto.fi}{\texttt{jara.uitto@aalto.fi}}} \vspace{0.5mm}\\
\end{minipage}

\vspace{5mm}
\begin{minipage}[H]{13.3cm}
\begin{center}
    {\bf Abstract} \\ 
\end{center}

Locally Checkable Labeling (\lcl) problems are graph problems in which a solution is correct if it satisfies some given constraints in the local neighborhood of each node. 
Example problems in this class include maximal matching, maximal independent set, and coloring problems.
A successful line of research has been studying the complexities of \lcl problems on paths/cycles, trees, and general graphs, providing many interesting results for the \local model of distributed computing. In this work, we initiate the study of \lcl problems in the low-space Massively Parallel Computation (\mpc) model. In particular, on forests, we provide a method that, given the complexity of an \lcl problem in the \local model, automatically provides an exponentially faster algorithm for the low-space \mpc setting that uses optimal global memory, that is, truly linear. \\

While restricting to forests may seem to weaken the result, we emphasize that all known (conditional) lower bounds for the \mpc setting are obtained by lifting lower bounds obtained in the distributed setting \emph{in tree-like networks} (either forests or high girth graphs), and hence the problems that we study are challenging already on forests. Moreover, the most important technical feature of our algorithms is that they use optimal global memory, that is, memory linear in the number of edges of the graph. In contrast, most of the state-of-the-art algorithms use more than linear global memory. Further,  they typically start with a dense graph, sparsify it, and then solve the problem on the residual graph, exploiting the relative increase in global memory. On forests, this is not possible, because the given graph is already as sparse as it can be, and using optimal memory requires new solutions.

\end{minipage}
\end{center}

\vfill
\thispagestyle{empty}
\footnotetext{Supported in part by the Academy of Finland, Grant 334238}

\newpage
\thispagestyle{empty}
\tableofcontents

\newpage
\pagenumbering{arabic}

\section{Introduction}
The Massively Parallel Computation (\mpc) model, introduced in~\cite{KarloffSV10} and later refined by~\cite{mpcrefine1, mpcrefine2, broadcast}, is a mathematical abstraction of modern data processing platforms such as MapReduce~\cite{dg04}, Hadoop~\cite{White:2012}, Spark~\cite{ZahariaCFSS10}, and Dryad~\cite{Isard:2007}.
Recently, tremendous progress has been made on fundamental \emph{graph problems} in this model, such as maximal independent set (MIS), maximal matching (MM)~\cite{GU19, Czumaj20}, and coloring problems~\cite{ChangMPCColoring, detcol}.
All these problems, and many others, fall under the umbrella of \emph{Locally Checkable} problems, in which the feasibility of a solution can be checked by inspecting local neighborhoods. They also serve as abstractions for fundamental primitives in large-scale graph processing and have recently gained a lot of attention~\cite{BCMOS21, Balliu21, detcol,spanner, balliu2020, Chang2020,GGC20}. \emph{Locally checkable labelings (\lcls)} are locally checkable problems restricted to constant degree graphs. They are defined through a set of feasible configurations from the viewpoint of each individual node. A more formal definition of \lcls is deferred to \Cref{sec:prelim}.

\lcls have been a rich source of research in various models of computation, because they can be seen as a starting point to understand locally checkable problems in general, and this holds independently from the model. For example, in the distributed setting, techniques developed to understand \lcls \cite{sinkless16} have then been used to prove lower bounds in the unbounded degree setting, which the \lcl setting does not include, e.g., for the the maximal independent set problem, or the $\Delta$-coloring problem \cite{Balliu2019,lbrs,hideandseek}.
In the distributed \local model of computing, a lot is known about \lcls: for example, if the graph on which we want to solve the problem is a tree, then there is a discrete set of possible complexities, and in some cases, given an \lcl, we can even automatically decide its distributed time complexity.  Our goal is to bring to the parallel setting, and in particular to the \mpc model, the knowledge that researchers developed about \lcls in the distributed setting, while also developing new techniques that can be used in the parallel setting.
We show that, on forests, the mere knowledge of what is the distributed complexity of a problem is enough to obtain blazingly fast algorithms in the \mpc setting. In particular, we obtain \mpc algorithms that are \emph{exponentially faster} than the best distributed ones. 
We summarize our main result. 

\begin{framed}
	\centering
	The complexity of any \lcl problem on forests in the \mpc model is exponentially lower than its distributed complexity, even when using optimal memory bounds.
\end{framed}

More in detail, in our work, we solve \lcl problems in forests in the most restrictive \emph{low-space} \mpc model with \emph{linear} total memory, which is the most \emph{scalable} variant of the \mpc model.
Our results provide an automatic method that, for all \lcl problems, yields an algorithm that solves the given problem \emph{exponentially} faster than its optimal distributed counterpart.
The resulting algorithms are \emph{component-stable} \cite{focs,componentstable}, which implies that the solutions in individual connected components are independent of the other components. Our results are in some sense optimal: for problems that in the \local model can be solved in $n^{o(1)}$, finding more-than-exponentially faster component-stable algorithms would violate the widely-believed $1$ vs.\ $2$ cycle conjecture in the \mpc setting.

\subparagraph{Why do we care about trees and forests?}
All known conditional lower bounds\footnote{Proving unconditional lower bounds for the \mpc model would imply a major breakthrough in circuit complexity and seems out of reach \cite{Roughgarden18}.} for problems in the \mpc setting are derived by lifting lower bounds that hold in the \local model of distributed computing \cite{focs,componentstable}. Most of the lower bounds known in the \local model are actually proved either on trees or on high-girth graphs (where the neighborhood of each node corresponds to a tree): see, e.g., \cite{KuhnMW16,hideandseek,Balliu2019,lbrs,sinkless16}. It follows that essentially all the conditional lower bounds known in the \mpc setting already hold on forests\footnote{As lifting lower bounds from the \local model to the \mpc model requires hereditary graph classes one cannot immediately lift a lower bound in the \local model that holds on trees. Instead, a lower bound in the \local model on trees implies the same lower bound in the \local model for forests which can then be lifted to a lower bound for \mpc algorithms on forests.}. Despite this fact, with a few exceptions, there is no work on upper bounds on forests in the \mpc model---a gap we aim to fill.

Moreover, understanding the complexity of problems on trees has been already shown to be essential in the \local model: it is typically the case that interesting problems are already challenging on trees, and often even in regular balanced trees of small degree. In fact, most lower bounds known in the \local model hold exactly in this setting. Due to the lifting, the same statement adapted to forests is true for all recent \mpc lower bounds. Hence, to decrease the relevance of trees and forests, we either need completely new lower bound techniques in the \local model coupled with completely new lifting theorems, or completely new lower bound techniques for the \mpc model. 

At first glance it may seem that our results are easy to achieve, because we restrict to forests. Conversely, we would like to emphasize that many state-of-the-art algorithms for problems like MIS and coloring work as follows~\cite{GU19, Czumaj20}: start with a dense graph which requires a lot of memory to store, sparsify it, and then use the freed global memory to solve the problem faster on the sparsified part. On forests, this is not possible, because the given graph is already as sparse as it can be.

\subparagraph{The \mpc Model.}

In the \mpc model, we have $M$ machines who communicate in an all-to-all fashion. We focus on problems where the input is modeled as a graph with $n$ vertices, $m$ edges and maximum degree $\Delta$; we call this graph the \textit{input graph}. 
Each node has a unique ID of size $b=O(\log n)$ bits from a domain $\{1,2,\dots,N\}$, where $N = \text{poly}(n)$. Each node and its incident edges are hosted on a machine(s) with $S = O(n^\delta)$ \textit{local memory}, where $\delta \in (0,1)$ and the units of memory are \emph{words} of $O(\log n)$ bits. When the local memory is bounded by $O(n^\delta)$, the model is called \textit{low-space} (or \textit{sublinear}). 
The number of machines is chosen such that $M = m / S = \Theta(m/n^\delta)$. For trees, where $m=\Theta(n)$, this results in $\Theta(n^{1-\delta})$ machines, that is, a \emph{total memory} (or \emph{global memory}) of $M \cdot S = \Theta(n)$. For simplicity\footnote{In practice, it is assumed that the virtual machines can be shuffled between physical machines, such that the sum of the memory of the virtual machines hosted on any single physical machine is $O(n^\delta)$.}, we assume that each machine $i$ simulates one virtual machine for each node and its incident edges that $i$ hosts, such that the local memory restriction becomes that no virtual machine can use more than $O(n^\delta)$ memory. 

During the execution of an \mpc algorithm, computation is performed in synchronous, fault-tolerant rounds. In each round, every machine performs some (unbounded) computation on the locally stored data, then sends/receives messages to/from any other machine in the network. Each message is sent to exactly one other machine specified by the sending machine. All messages sent and received by each machine in each round, as well as the output, have to fit into local memory. The time complexity is the number of rounds it takes to solve a problem. Upon termination, each node (resp.\ its hosting machine) must know its own part of the solution.
For example in the case of node-coloring, the machine hosting node $u$ must decide on the color of $u$ upon termination of the algorithm.

Unlike in most other works, our algorithms employ $O(m+n)$ words of total memory, which is the strictest possible as it is only enough to store a constant number of copies of the input graph. Note that if we were to allow superlinear $O(m^{1+\delta})$ global memory in our constant-degree setting, many \local algorithms with complexity $O(\log n)$ could be trivially sped up exponentially in the low-space \mpc model by applying the well-known graph exponentiation technique by Lenzen and Wattenhofer~\cite{wattenhofer}. A crucial challenge that comes with the linear global memory restriction is that only a small fraction of $n^{1 - \delta}$ of the (virtual) machines can simultaneously utilize all of their available local memory. Thus, with strictly linear global memory we are forced to develop new techniques which must avoid gathering local neighborhoods, i.e., fundamentally divert from direct simulations of message passing algorithms.

\subsection{The Distributed Complexity Landscapes}\label{ss:landscape}
In the last decade, there has been tremendous progress in understanding the complexities of \lcl{}s in various models of distributed and parallel computing. 
A prime example is the \local model~\cite{linial}, where the input graph corresponds to a message passing system, and the nodes must output their part of the solution according only to local information about the graph. Another example is the \congest model, which is a \local model variant where the message size is restricted to $O(\log n)$ bits~\cite{peleg}. A curious fact about \lcl{}s in the distributed setting is the existence of \emph{complexity gaps}, that is, some complexities are not possible at all. For example, it is known that there are no \lcl{}s with a distributed time complexity in the \local and \congest model that lies between $\omega(\log^* n)$ and $o(\log n)$.
In these two models, the \textit{whole} complexity landscape of \lcl problems is now understood for some important graph families. For instance, a rich line of work~\cite{NaorS95, tree1, CP19, tree2, tree3, balliu2020, Chang2020} recently came to an end when a complexity gap between $\omega(1)$ and $o(\log^* n)$ was proved~\cite{lclcomplete}, completing the randomized/deterministic complexity landscape of \lcl problems in the \local model for trees. 
In the \congest model, the authors of \cite{BCMOS21} showed that, on trees, the complexity of an \lcl problem is asymptotically equal to its complexity in the \local model, whereas the same does not hold in general graphs.
In the randomized/deterministic \local and \congest models, recent work showed that the complexity landscapes of \lcl problems for rooted regular trees are fully understood~\cite{Balliu21}, while the complexity landscapes of \lcl problems in the \local model for rings and tori have already been known for some while~\cite{Brandt17}. Even for general (constant-degree) graphs, the \local complexity landscape of \lcl problems is almost fully understood~\cite{NaorS95, sinkless16, CKP19, CP19, GhaSu17, FischerG17, GHK18, BHKLOS18, tree3, RozhonG20, GhaffariGR21}, only missing a small part of the picture related to the randomized complexity of Lovász Local Lemma (\llle). 

In the case of trees, for deterministic algorithms in the \local model, it is known that there is a discrete set of possible complexities, that we  divide into four categories:

\begin{itemize}
	\setlength{\itemsep}{2mm}
	\item \textbf{Tiny regime}: contains the complexities $O(1)$ and $\Theta(\log^* n)$. \vspace{-1mm}
	\begin{itemize}[noitemsep]
		\item Example problems: maximal independent set, maximal matching, $(\Delta+1)$-vertex coloring\footnote{We denote the maximum degree of the graph by $\Delta$.}, $(2\Delta-1)$-edge coloring, and trivial problems (e.g., all nodes must output $0$). 
	\end{itemize}
	\item \textbf{Mid regime}: contains the complexity $\Theta(\log n)$. \vspace{-1mm}
	\begin{itemize}[noitemsep]
		\item Example problems: sinkless orientation \cite{sinkless16}, $3$-coloring, and $\Delta$-coloring.
	\end{itemize}
	\item \textbf{High regime}: contains the complexities $\Theta(n^{1/k})$, for all $k \in \mathbb{N}$. \vspace{-1mm}
	\begin{itemize}[noitemsep]
		\item Example problems: 2-coloring and $2\tfrac{1}{2}$-coloring \cite{CP19}.
	\end{itemize}
\end{itemize}

Moreover, it is known that randomness can help only in the mid regime, and in particular that \emph{some} problems requiring $\Theta(\log n)$ for deterministic algorithms have randomized complexity $\Theta(\log \log n)$, which constitutes our fourth category---\textbf{Low regime}. Problems residing in the low regime include sinkless orientation and $\Delta$-coloring.

On forests, the complexity landscape in the \local model is the same as on trees.
While this is intuitively evident, it can also be shown formally using an analogous approach to the one used in the proof of~\cite[Lemma 3.3]{lclcomplete}.

\subsection{Our Contributions}

Our main contribution is showing that, given any \lcl problem (see \Cref{def:lcl}) on trees that has deterministic (resp.\ randomized) complexity $T$ in the \local model, we can automatically obtain an \mpc algorithm with deterministic (resp.\ randomized) complexity $O(\log T)$ on forests. In particular, we prove the following.

\begin{restatable}{theorem}{thmSpeedUp} \label{thm:SpeedUp}
	Consider an \lcl problem on trees with deterministic time complexity $f(n)$ and randomized time complexity $g(n)$ in the \local model. This problem has deterministic time complexity $O(\log f(n))$ and randomized time complexity $O(\log g(n))$ in the low-space \mpc model on forests using optimal $O(m+n)$ words of global memory. The provided algorithms are component-stable.
\end{restatable}

Put it differently, a problem in the \local model can only have a deterministic complexity  $f(n) \in \{ \Theta(1), \Theta(\log^* n), \Theta(\log n)\} \cup \{\Theta(n^{1/k}) ~|~ k \in \mathbb{N}\}$, and we show that it is enough to know the asymptotic value of $f(n)$ in order to obtain a deterministic \mpc algorithm with complexity $O(\log(f(n))) \in \{ O(1), O(\log \log^* n), O(\log \log n), O(\log n)\}$.

Moreover, it is known that for all $f(n) \not\in \Theta(\log n)$, the \local randomized complexity of the problem is the same as the deterministic one. Instead, for $f(n) \in \Theta(\log n)$, the \local randomized complexity $g(n)$ can be either $\Theta(\log n)$ or $\Theta(\log \log n)$. If it is $\Theta(\log \log n)$, then we provide an \mpc algorithm with randomized complexity $O(\log \log \log n)$.
If we dismiss the component-stability requirement, we can obtain the same $O(\log \log \log n)$ runtime with a deterministic \mpc algorithm. 

\begin{restatable}{theorem}{corSpeedUpDet} \label{cor:SpeedUpDet}
	Consider an \lcl problem on trees with randomized time complexity $g(n) = \Theta(\log \log n)$ in the \local model. This problem has deterministic time complexity $O(\log \log \log n)$ in the low-space \mpc model on forests using optimal $O(m+n)$ words of global memory. This algorithm is component-unstable.
\end{restatable}
By \cite{focs,componentstable}, we know that \Cref{thm:SpeedUp} is in some sense optimal: if a problem requires $T$ deterministic rounds in the \local model, then it requires $\Omega(\min\{\log T, \log \log n\})$ rounds in the low-space \mpc setting for component-stable algorithms, assuming that the infamous $1$ vs.\ $2$ cycle conjecture holds \cite{mpcrefine2,focs,Roughgarden18}. In contrast, \Cref{cor:SpeedUpDet} shows that one can break the conditional  lower bound of $\Omega(\log\log n)$ for deterministic \mpc algorithms for all \lcl problems in the aforementioned class by diverting to component-unstable algorithms. Achieving the same result even for a single problem without dismissing the component-stability requirement would be a major breakthrough, as it would falsify the conjecture.

As a subroutine for solving all problems that belong to the high regime in $O(\log n)$ \mpc rounds, we also develop an $O(\log n)$ round \mpc algorithm for rooting a forest. This rooting algorithm is component-stable, and may be of independent interest, since it is also compatible with arbitrary degrees (see \Cref{lemma:rooting}).

\subparagraph{Additional observations.}
There is a long line of research that provided algorithms for \mpc that are exponentially faster than the best algorithms for the \local model. Most existing results achieved these speedup results by using additional global memory, that is, $\omega(m)$ words~\cite{Behnezhad19, GGC20, detcol, componentstable}. We emphasize that, deviating from the usual approach,  all of our results use \emph{optimal} \mpc parameters, in the sense that we work in the low-space setting with $O(n^\delta)$ words of local memory and $O(m+n)$ words of global memory. 

Hence, our contribution is twofold, on the one hand we prove that we can indeed achieve this exponential speedup for all \lcl{}s, while on the other hand we show that this exponential speedup can be achieved without requiring any additional memory. 
Furthermore, graph problems in trees and forests are widely unexplored, despite their central role that we have already elaborated on.
It is known that a $4$-coloring, MIS, and maximal matching can be found in $O(\log \log n)$ rounds~\cite{GGC20}.
However, the coloring result heavily relies on randomness and the MIS and matching results require a (small) overhead in the total memory.
To compare, our results deterministically yield a $3$-coloring in $O(\log \log n)$ rounds with linear total memory.
It is not clear whether randomness can even help in the case of $3$-coloring, which is a significant difference to the case of $4$-coloring.
Furthermore, it is not clear whether the previous approaches to MIS and matching can be extended to work deterministically with the same runtime and with linear total memory.
While the previous work is designed for arbitrary degree graphs, it is not clear whether the algorithms could be tuned to work faster with constant degrees.

\subparagraph{Open Questions.}
In the tiny regime, our results extend to general graphs (see \Cref{thm:logstar}). In the low regime, our results extend to general graphs if we allow slightly more global memory (see \Cref{thm:generalGraphMoreSpace}). Once we reach the mid regime, i.e., logarithmic distributed complexities, we do not know the behaviour in general graphs. This leads to an interesting open question. As mentioned, the asymptotic complexity of any problem on trees is identical in the \local and \congest model, and the same is true (modulo the exact complexity of the \llle in both models) on general graphs as long as the complexity is sublogarithmic~\cite{BCMOS21}. However, there is a an exponential separation between the models for complexities that are at least logarithmic \cite{BCMOS21}. Does such a separation between the complexity of an \lcl in the \local model and the \mpc model also hold for large complexities? Here, of course, we would want to have a doubly exponential separation. 

Interestingly, current conditional lower bounds for the \mpc model cannot prove \mpc lower bounds that are  $\omega(\log\log n)$. So, while our results in the high regime show that any problem on forests can be solved in $O(\log n)$ rounds in the \mpc model, it remains unclear whether we cannot improve on this bound, even without falsifying the $1$ vs. $2$ cycle conjecture.

\subparagraph{Component-stability.}
The term of a \emph{component-stable} \mpc algorithm has been introduced in \cite{focs} in the context of lifting distributed lower bounds to the \mpc setting. By their definition, informally, an algorithm is component-stable if the output of a node does not change if other connected components in the graph are altered (see \Cref{def:componentStability}). 

While initially believed that it might be an artifact of their lifting techniques, Czumaj, Davies and Parter \cite{componentstable} showed the contrary, i.e., they showed that \emph{component-unstable} algorithms can beat the conditional lower bounds of \cite{focs}. 
Their results hold assuming their revised definition of component-stability, which is argued to be more robust (see \Cref{def:componentStabilityRev}). Under their definition, it is not strictly easier nor harder to design algorithms to be component-stable, as compared to the definition of \cite{focs}.
The main difference is that they allow the output of component-stable algorithms to depend on the total number of nodes in the graph and the maximum degree. In our work, we adopt the revised definition of component-stability \cite{componentstable}. See \Cref{ssec:components-stability} and the discussion therein for further details. 
	  
\subsection{Challenges \& Key Techniques}
We now provide an overview of the challenges that we had to tackle in order to prove our results, and a very high level explanation of the key techniques that we used to solve them.

The tiny regime serves as a good warm-up to see why using an optimal amount of global memory is difficult. The most technically involved part is the high regime, where we obtain an $O(\log n)$-time \mpc algorithm for any \lcl problem.

\subparagraph{Graph Exponentiation.} A reoccurring challenge for all regimes lies in respecting the linear global memory, which roughly means that on average, every node can use only a constant amount of memory. This is particularly unfortunate because almost all recent \mpc results---and in particular all that achieve exponential speedups---rely on the memory-intense \emph{graph exponentiation} technique \cite{wattenhofer}. Informally, this technique enables a node to gather its $2^k$-hop neighborhood in $k$ communication rounds. Doing this in parallel for every node in the graph results in a $\Delta^{2^k}$ overhead in global memory. For this technique to be useful, $k$ has to be $\omega(1)$, yielding a non-constant multiplicative increase in the global memory requirement. In order to use this technique but not violate linear global memory, we develop new solutions that are discussed in the following paragraphs.

\subparagraph{Tiny regime $f(n)=\Theta(1)$ and $f(n)=\Theta(\logstar n)$:}
Handling the $\Theta(1)$ complexity is trivial, since any \local algorithm for \lcl{}s can be simulated in the \mpc setting. For the $\Theta(\logstar n)$ class, it is known from prior work that all problems can be solved in the \local model in a very specific way: reduce to the problem of computing a distance-$k$ coloring with a small enough number of colors, where $k$ is a constant that depends on the problem. In a distance-$k$ $c$-coloring, each node is assigned a color in $\{1,\ldots, c\}$ such that nodes at distance at most $k$ have different colors.  Such a coloring can be computed in $O(\log^* n)$ rounds in the \local model, and it could be computed easily in the \mpc setting in $O(\log \log^* n)$ rounds, by exploiting the graph exponentiation technique, if we allow an additional $O(\log^* n)$ factor overhead in the amount of global memory.

We show that this overhead is not required, by developing a novel \mpc algorithm for coloring. The algorithm that we provide reduces the problem of coloring a general graph to coloring directed pseudoforests, that is, graphs where all edges are oriented and every node has at most one outgoing edge. Then, we show that in directed pseudoforests, it is possible to solve the coloring problem through a variant of graph exponentiation that only requires keeping track of a constant number of IDs. This way, the memory use of each node is constant, and the global memory is linear.

\subparagraph{High regime $f(n)=\Theta(n^{1/k})$, for all $k \in \mathbb{N}$:}

We explicitly provide, for any solvable \lcl, a novel algorithm that has a runtime of $O(\log n)$. Essentially, we solve each tree in the forest separately, hence we will consider trees in the following argumentation.  On a high level, our algorithm first roots the tree using our $O(\log n)$-time tree rooting algorithm (see \Cref{sec:rooting}), and then proceeds in two phases. In the first phase, roughly speaking, the goal is to compute, for a substantial number of nodes $v$, the set of possible output labels that can be output at $v$ such that the label choice can be extended to a (locally) correct solution in the subtree hanging from $v$. This is done in an iterative manner, proceeding from the leaves towards the root. The second phase consists of using the computed information to solve the given \lcl from the root downwards.

While this outline sounds simple, there are a number of intricate challenges that require the development of novel techniques, both in the design of the algorithm and its analysis. For instance, the depth of the input tree may be $\omega(\log n)$ (which prevents us from performing the above ideas in a sequential manner),
and the storage of the required completability information grows exponentially when using graph exponentiation, exceeding the available global memory. Our key technical contributions are the following.
\begin{itemize}
	\item The design of a process that allows for interleaving graph exponentiation steps and compressing the graph (and compatibility information) such that the process is also reversible (second phase of the algorithm). The main challenge here is that multiple graph exponentiation processes executed on individual parts of the tree have to be merged, simultaneously or at different times, into one process during the execution.
	
	\item The design of a fine-tuned potential function for the analysis of the complex algorithm resulting from addressing the aforementioned issues and the highly non-sequential behavior arising from interleaving graph exponentiation steps.
\end{itemize}

\subparagraph{Mid regime $f(n)=\Theta(\log n)$: } We would wish to use the algorithm of Chang and Pettie \cite{CP19} as a black box. On a very high level idea, their \local algorithm uses $O(\log n)$ rounds to compute a rake-and-compress decomposition of size $O(\log n)$, which is essentially the classic $H$-partition by Miller and Reif \cite{MillerReif89}. Then, compatibility information of the given \lcl problem (see \Cref{sec:mid} for more details) is propagated layer by layer to the top, and then labels are fixed at the top and propagated down.

Applying known \mpc techniques like graph exponentiation to speed up this process does not work out of the box for several reasons. First, the compatibility information they propagate grows exponentially, which creates congestion in the \mpc model. Secondly, since the input graph is as sparse as it could possibly be, the direct application of graph exponentiation would violate the optimal global memory bounds we are striving for. We resolve the first issue by first observing that the compatibility information can be reduced to constant size in every iteration. The second issue is remedied by interleaving exponentiation steps with memory freeing steps in a balanced way.

\subparagraph{Low regime $g(n)=\Theta(\log \log n)$:} With an additional $O(\log n)$ factor of global memory, this result is easy to obtain. Previous work \cite{BCMOS21} has a constant time reduction to instances of size $N=\log n$, resulting in a \local algorithm with runtime $\poly (\log N) = \poly(\log\log n)$. 
A straightforward application of graph exponentiation would yield an \mpc algorithm with runtime $O(\log\log \log n)$. Exploiting additional global memory in this manner has been used in a similar setting in \cite{componentstable}. 
However, without the additional memory it is harder to solve the small instances in triple logarithmic time.
The work around for this memory issue is to use our mid regime algorithm on the small instances, yielding a memory efficient algorithm with runtime $O(\log\log\log n)$. 
To the best of our knowledge there is no other paper that can efficiently deal with such occurring small instances---small instances occur also in many other problems like MIS and graph coloring---with optimal global memory.

\subsection{Further Related Work}

For many of the classic graph problems, simple $O(\log n)$-time \mpc algorithms follow from classic literature in the \local model and \textsf{PRAM}~\cite{Alon86, linial, Luby85}. In particular in the case of bounded degree graphs, it is often straightforward to simulate algorithms from other models. However, it is usually desirable to get algorithms that run \emph{much faster} than their \local counterparts.
If the \mpc algorithms are given \emph{linear} $\Theta(n)$ or even \emph{superlinear} $\Theta(n^{1+\delta})$ local memory, fast algorithms are known for many classic graph problems.

In the sublinear (or low-space) model, \cite{Chang2019} provided a randomized algorithm for the $(\Delta + 1)$-coloring problem that, combined with the new network decomposition results~\cite{RozhonG20,GhaffariGR21}, yields an $O(\log \log \log n)$ \mpc algorithm, that is exponentially faster than its \local counterpart. A recent result by Czumaj, Davies, and Parter~\cite{detcol} provides a deterministic $O(\log \log \log n)$-time algorithm for the same problem using derandomization techniques. 
For many other problems, the current state of the art in the sublinear model is still far from the aforementioned exponential improvements over the \local counterparts, at least in the case of general graphs. For example, the best known MIS, maximal matching, $(1+\epsilon)$-approximation of maximum matching, and 2-approximation of minimum vertex cover algorithms run in $\widetilde{O}(\sqrt{\log \Delta} + \sqrt{\log \log n})$ time~\cite{GU19}, whereas the best known \local algorithm has a logarithmic dependency on $\Delta$~\cite{Ghaffari16}. For restricted graph classes, such as trees and graphs with small arboricity\footnote{The arboricity of a graph is the minimum number of disjoint forests into which the edges of the graph can be partitioned.} $\alpha$, better algorithms are known~\cite{sirocco, Behnezhad19}. Through a recent work by Ghaffari, Grunau and Jin, the current state of the art for MIS and maximal matching are $O(\sqrt{\log \alpha} \cdot \log \log \alpha + \log \log n)$-time algorithms using $\widetilde{O}(n + m)$ words of global memory~\cite{GGC20}.

As for lower bounds, \cite{focs} gave conditional lower bounds of $\Omega(\log \log n)$ for component-stable sublinear \mpc algorithms for constant approximation of maximum matching and minimum vertex cover, and MIS. In addition, the authors provided a lower bound of $\Omega(\log \log \log n)$ for \llle. Their hardness results are conditioned on a widely believed conjecture in \mpc about the complexity of the connectivity problem, which asks to detect the connected components of a graph. It is argued that disproving this conjecture would imply rather strong and surprising implications in circuit complexity~\cite{Roughgarden18}. When assuming component-stability, they also argue that all known algorithms in the literature are component-stable or can easily be made component-stable with no asymptotic increase in the round complexity. However, recent work~\cite{componentstable} gave a separation between stable and unstable algorithms, and that some particular problems (e.g., computing an independent set of size $\Omega(n/\Delta)$) can be solved faster with unstable algorithms than with stable ones.

It is also worth discussing the complexity of rooting a tree, as it is an important subroutine in our high regime. On the randomized side, \cite{sirocco} gave an $O(\log d \cdot \log \log n)$ time algorithm, where $d$ is the diameter of the graph. On the deterministic side, Coy and Czumaj~\cite{coy2021deterministic} gave an $O(\log n)$ time algorithm using (component-unstable) derandomization methods, which is the current state of the art. In \Cref{sec:rooting} we provide a totally different rooting algorithm that is also deterministic and takes $O(\log n)$ time, but is component-stable. We note that \cite{pathexp} uses similar techniques in a more general setting, but in $\omega(\log n)$ time.

\subsection{Outline}

After the formal introduction of \lcl problems and other notations in \Cref{sec:prelim}, we start proving the exponential speedup for the different regimes \Cref{thm:SpeedUp} in separate sections. In \Cref{sec:tiny}, we warm-up with the tiny regime.  In \Cref{sec:high},  we present the algorithm for our most involved result, the high regime.
In \Cref{sec:rooting} we present the rooting algorithm that is used as a subroutine in the high regime.  Due to its complexity and length, the formal analysis for the high regime is deferred to  \Cref{sec:apphigh}. 
In \Cref{sec:low,sec:mid}, we present the speedup for the low and mid regime, respectively. As the proof of \Cref{cor:SpeedUpDet} requires the same techniques as the speedup for the mid regime, its proof is also presented in \Cref{sec:mid}.
Some of our speedup results use a description of a distributed algorithm with the claimed runtime to obtain the speedup. In \Cref{sec:automatic} we show that such a description can be inferred merely by knowing the distributed complexity class in which the problem resides. The section also contains additional reasons why our results apply to forests, for all the cases not reasoned elsewhere.
Lastly, in \Cref{sec:broadcasttree}, we describe the \mpc broadcast tree for completeness, which is an important primitive of the model, and is used implicitly throughout the paper.

\section{Definitions and Notation} \label{sec:prelim}

We work with undirected, finite, simple graphs $G = (V,E)$ with $n=|V|$ nodes and $m=|E|$ edges such that $E \subseteq [V]^2$ and $V \cap E = \emptyset$. Let $\deg_G(v)$ denote the degree of a node $v$ in $G$ and let $\Delta$ denote the maximum degree of $G$. The distance $d_G(v,u)$ between two vertices $v,u$ in $G$ is the length of a shortest $v - u$ path in $G$; if no such path exists, we set $d_G(v, u) \coloneqq \infty$. The greatest distance between any two vertices in $G$ is the diameter of $G$, denoted by $\text{diam}(G)$. For a subset $S \subseteq V$, we use $G[S]$ to denote the subgraph of $G$ induced by nodes in $S$. Let $G^k$, where $k \in \mathbb{N}$, denote the $k$:th power of a graph $G$, which is another graph on the same vertex set, but in which two vertices are adjacent if their distance in $G$ is at most $k$. In the context of \mpc, $G^k$ is the resulting virtual graph after performing $\log k$ steps of graph exponentiation~\cite{wattenhofer}.

For each node $v$ and for every radius $k \in \mathbb{N}$, we denote the $k$-hop (or $k$-radius) neighborhood of $v$ as $N^k(v) = \{ u \in V : d(v,u) \leq k\}$. The topology of a neighborhood $N^k(v)$ of $v$ is simply $G[N^k(v)]$. However, with slight abuse of notation, we sometimes refer to $N^k(v)$ both as the node set and the subgraph induced by node set $N^k(v)$. Neighborhood topology knowledge is often referred to as vision, e.g., node $v$ sees $N^k(v)$.
In trees and forests, the number $n$ of nodes and the number $m$ of edges are asymptotically equal, and we may use them interchangeably throughout the paper when reasoning about global memory.

\subsection{\lcl Definitions}

In their seminal work~\cite{NaorS95}, Naor and Stockmeyer introduced the notion of a locally checkable labeling problem (\lcl problem or just \lcl for short). The definition they provide restricts attention to problems where nodes are labeled (such as vertex coloring problems), but they remark that a similar definition can be given for problems where edges are labeled (such as edge coloring problems). A modern way to define \lcl problems that captures both of the above types of problems (and combinations thereof) labels \emph{half-edges} instead, i.e., pairs $(v,e)$ where $e$ is an edge incident to vertex $v$. Let us first define a half-edge labeling formally, and then provide this modern \lcl problem definition.

\begin{definition}[Half-edge labeling]\label{def:halfedge}
	A \emph{half-edge} in a graph $G = (V,E)$ is a pair $(v,e)$, where $v \in V$ is a vertex, and $e \in E$ is an edge incident to $v$.
	A half-edge $(v,e)$ is incident to some vertex $w$ if $v = w$.
	We denote the set of half-edges of $G$ by $H = H(G)$.
	A \emph{half-edge labeling} of $G$ with labels from a set $\Sigma$ is a function $g \colon H(G) \to \Sigma$.
\end{definition}

We distinguish between two kinds of half-edge labelings: \emph{input labelings} that are part of the input and \emph{output labelings} that are provided by an algorithm executed on input-labeled instances. Throughout the paper, we will assume that any considered input graph $G$ comes with an input labeling $\ginn \colon H(G) \to \sinn$ and will refer to $\sinn$ as the \emph{set of input labels}; if the considered \lcl problem does not have input labels, we can simply assume that $\sinn = \{\bot\}$ and that each node is labeled with $\bot$.
Then, \Cref{def:solve} details how a correct solution for an \lcl problem is formally specified.

\begin{definition}[\lcl] \label{def:lcl}
	An \lcl problem, \lcl for short, is a quadruple $\Pi = (\sinn, \sout, r, \fP)$ where $\sinn$ and $\sout$ are finite sets (of input and output labels, respectively), $r \geq 1$ is an integer, and $\fP$ is a finite set of labeled graphs $(P, \pinn, \pout)$. The input and output labeling of $P$ are specified by $\pinn \colon H(P) \to \sinn$ and $\pout \colon H(P) \to \sout$, respectively.\footnote{Note that the original definition given in~\cite{NaorS95} considers \emph{centered graphs}; however, since we only consider trees, considering uncentered graphs instead suffices.}
\end{definition}

Recall that $N^r(v)$ denotes the subgraph of $G$ induced by all nodes at distance at most $r$ from $v$. This naturally extends to labeled graphs.
\begin{definition}[Solving an \lcl] \label{def:solve}
	 A \emph{correct solution} for an \lcl problem $\Pi = (\sinn, \sout, r, \fP)$ on a graph $(G, \ginn)$ labeled with elements from $\sinn$ is a half-edge labeling $\gout \colon H(G) \to \sout$ s.t.~for each node $v \in V(G)$, the neighborhood $N^r(v)$ in $(G, \ginn, \gout)$ is isomorphic to some member of $\fP$. We require that the isomorphism respects\footnote{In other words, any two half-edges (in $N^r(v)$ and the member of $\fP$, respectively) that are (implicitly) mapped to each other via the isomorphism are required to have identical input and output labels.} the input and output labelings of $N^r(v)$ and the member of $\fP$. We say that an algorithm $\fA$ \emph{solves} an \lcl problem $\Pi$ on a graph class $\fG$ if it provides a correct solution for $\Pi$ for every $G \in \fG$.
\end{definition}

Note that the \lcl definitions above implicitly require that graph class $\fG$ has constant degree.
It is often useful to rephrase a given \lcl in a way that minimizes the integer $r$ in the \lcl definition. In fact, since we only consider trees, any \lcl can be rephrased in a special form, called \emph{node-edge-checkable \lcl}, where $r$ is essentially set to $1$.\footnote{Arguably, this can be seen as $r = 1/2$, which might provide a better intuition.} While the formal definition of a node-edge-checkable \lcl appears complicated, the intuition behind it is simple: essentially, we have a list of allowed output label combinations around nodes, a list of allowed output label combinations on edges, and a list of allowed input-output label combinations, all of which a correct solution for the \lcl has to satisfy.

\begin{definition}[Node-edge-checkable \lcl]\label{def:nodeedge}
	Let $\Delta$ be some non-negative integer constant. A \emph{node-edge-checkable \lcl} is a quintuple $\Pi = (\sinn, \sout, \noco, \edco, \gee)$ where $\sinn$ and $\sout$ are finite sets, $\noco = \{\noco_1, \dots, \noco_{\Delta} \}$ consists of sets $\noco_i$ of cardinality-$i$ multisets with elements from $\sout$, $\edco$ is a set of cardinality-$2$ multisets with elements from $\sout$, and $\gee \colon \sinn \to 2^{\sout}$ is a function mapping input labels to sets of output labels.
	We call $\noco_1 \cup \dots \cup \noco_{\Delta}$ and $\edco$ the \emph{node constraint} and \emph{edge constraint} of $\Pi$, respectively.
	Furthermore, we call each element of $\noco$ a \emph{node configuration}, and each element of $\edco$ an \emph{edge configuration}.	
	For a node $v$, denote the half-edges of the form $(v,e)$ for some edge $e$ by $h_1^v, \dots, h_{\deg(v)}^v$ (in arbitrary order).
	For an edge $e$, denote the half-edges of the form $(v,e)$ for some node $v$ by $h_1^e, h_2^e$ (in arbitrary order).
	A correct solution for $\Pi$ is a half-edge labeling $\gout \colon H(G) \to \sout$ such that
	\begin{enumerate}
		\item for each node $v$, the multiset of outputs assigned by $\gout$ to $h_1^v$, $\dots$, $h_{\deg(v)}^v$ is an element of $\noco_{\deg(v)}$,
		\item for each edge $e$, the cardinality-$2$ multiset of outputs assigned by $\gout$ to $h_1^e, h_2^e$ is an element of $\edco$, and
		\item for each half-edge $h \in H(G)$, we have $\gout(h) \in \gee(\iota)$, where $\iota = \ginn(h)$ is the input label assigned to $h$.
	\end{enumerate}
\end{definition}

On trees, each \lcl $\Pi$ (with parameter $r$ in its definition) can be transformed into a node-edge-checkable \lcl $\Pi'$ by the standard technique of requiring each node $v$ to output, on each incident half-edge $h$, an encoding of its entire $r$-hop neighborhood (including input labels, output labels, and a marker indicating which of the half-edges in the encoded tree corresponds to half-edge $h$).
From the definition of $\Pi'$, it follows immediately that $\Pi'$ is equivalent to $\Pi$ in the sense that any solution for $\Pi$ can be transformed (by a deterministic distributed algorithm) in constant time into a solution for $\Pi'$, and vice versa.
Hence, for the purposes of this work, we can safely restrict our attention to node-edge-checkable \lcls.

\subsection{Component-stability} \label{ssec:components-stability}

The term of a \emph{component-stable} \mpc algorithm has been introduced in \cite{focs} in the context of lifting distributed lower bounds to the \mpc setting. It was later revised by Czumaj, Davies and Parter \cite{componentstable} and argued to be made more robust.

\begin{definition}[Component-stability, \cite{focs}] \label{def:componentStability}
	An \mpc algorithm is component-stable if the outputs of nodes in different connected components are independent. Formally, assume that for a graph $G$, $\mathcal{D}_G$ denotes the initial distribution of the edges of $G$ among the $M$ machines and the assignment of unique IDs to the nodes of $G$. For a subgraph $H$ of $G$ let $\mathcal{D}_H$ be defined as $\mathcal{D}_G$ restricted to the nodes and edges of $H$. Let $H_v$ be the connected component of node $v$. An \mpc algorithm $\mathcal{A}$ is called component-stable if for each node $v \in V$, the output of $v$ depends (deterministically) on the node $v$ itself, the initial distribution and ID assignment $\mathcal{D}_{H_v}$ of the connected component $H_v$ of $v$, and on the shared randomness $\mathcal{S}_M$.
\end{definition}

In their revised definition, \cite{componentstable} assume the setting where all input graphs are legal.

\begin{definition}[Legal graph] \label{def:legalGraph}
    A graph $G$ is called legal if it is equipped with functions ID, name: $V(G) \arr [\poly(n)]$ providing nodes with IDs and names, such that all names are fully unique and all IDs are unique in every connected component.
\end{definition}

\begin{definition}[Component-stability (revised), \cite{componentstable}] \label{def:componentStabilityRev}
	A randomized \mpc algorithm $A_{\mpc}$ is component-stable if its output at any node $v$ is entirely, deterministically, dependent on the topology and IDs (but independent of names) of $v$’s connected component (which we will denote $CC(v)$), $v$ itself, the exact number of nodes $n$ and maximum degree $\Delta$ in the \emph{entire} input graph, and the input random seed $\mathcal{S}$. That is, the output of $A_{\mpc}$ at $v$ can be expressed as a deterministic function $A_{\mpc}(CC(v),v,n,\Delta,\mathcal{S})$. A deterministic \mpc algorithm $A_{\mpc}$ is component-stable under the same definition, but omitting dependency on the random seed $\mathcal{S}$.
\end{definition}

As opposed to \cite{focs}, \cite{componentstable} allow the output of component-stable algorithms to depend on the total number of nodes in the graph and the maximum degree of the graph. Additionally, they assume the following setting: all input graphs are \emph{legal} (see \Cref{def:legalGraph}), i.e., all nodes have an ID that is unique in every connected component, and a \emph{name} that is unique across the whole input graph. Assuming the above setting, the output of a component-stable algorithm is allowed to depend on the IDs of all nodes in the same components, but not the names. 

In our work, we adopt the revised definition of component-stability~\cite{componentstable}.
In all of our algorithms, nodes from different components only communicate in order to maintain a certain global synchrony.
This synchrony influences \emph{when} certain steps are executed and hence the \emph{execution} of our algorithms.
However, the output at each node is not influenced by the global communication.

\Cref{cor:SpeedUpDet} shows that the lower bounds for component-stable algorithms can be beaten for a large class of problems on trees and forests even with optimal memory. The long term effect of the term component-stable in this setting is unclear, but it provides room for many interesting open questions. One interesting aspect would be to see under which circumstances one can obtain algorithms with stronger component dependent guarantees, e.g., one may want to develop algorithms for which not just the output of a node, but also the time until it has computed its output  can only depend on the size of its component. Our algorithms do not meet this stronger definition. Besides an ID space dependence our algorithms have the following runtime behaviour. In the low and mid regime the time until we know the output of a node depends on the number of nodes in the largest connected component. In the high regime this time depends on the number of nodes in the whole graph. Going from trees to forests in the high regime relies on the recent beautiful (deterministic) connected components algorithm by Czumaj and Coy~\cite{coy2021deterministic,googleconnectivity}.

\section{The Tiny Regime} \label{sec:tiny}

In this section, we show that any \lcl problem on general graphs that can be solved in the \local model in $O(\log^* n)$ rounds, can be solved in the \mpc model in $O(\log \log^* n)$ rounds. By combining this result with known gaps in the landscape of possible complexities in the \local model \cite{CKP19}, we obtain the following result.

\begin{theorem}\label{thm:logstar}
	Let $\Pi$ be an \lcl problem on general graphs. Assume that there is a deterministic algorithm for the \local model that solves $\Pi$ in $o(\log n)$ rounds, or a randomized algorithm that solves it in $o(\log \log n)$ rounds. Then, the problem $\Pi$ can be solved deterministically in $O(\log \log^* N)$ rounds in the low-space \mpc model using $O(m+n)$ words of global memory, where $N=\poly(n)$ is the size of the ID space. The algorithm works even if the graph consists of disconnected components, and it is components-stable.
\end{theorem}

The rest of this section is devoted to proving \Cref{thm:logstar}.
\paragraph{A Universal Algorithm.}
In the \local model, it is known that, if an \lcl can be solved with an algorithm $A$ in $o(\log n)$ deterministic rounds, or in $o(\log \log n)$ randomized rounds, then it can also be solved with a deterministic algorithm $A'$ that requires just $O(\log^* n)$ rounds \cite{CKP19}. 
In order to prove this result, \cite{CKP19} shows how to convert any such algorithm $A$ into an algorithm $A'$ that works as follows (for some constant $k$ that depends on the problem $\Pi$ and the algorithm $A$):
\begin{enumerate}
	\item Compute a distance-$k$ $O(\Delta^{2k})$-coloring of the graph;\label{ckpstep1}
	\item Run a $k$-round algorithm $B$ that uses the computed coloring to produce the final output.\label{ckpstep2}
\end{enumerate}
In \cite{CKP19} is shown that the constant $k$, and the $k$-round algorithm $B$, can be mechanically determined from the original algorithm $A$. 
The runtime of algorithm $A'$ is $O(\log^* n)$ rounds since this is the runtime for the first step, while the second step only requires constant time.

\paragraph{Why it Works.}
The high-level purpose of computing the coloring in \Cref{ckpstep1} is to provide new identifiers at the nodes that are unique up to distance $k$ and come from a much smaller space than the original identifiers (that are part of the setting in the \local model).
Roughly speaking, this ensures that the $k$-hop view of any node that interprets the computed colors as identifiers is consistent with the node living in a constant-sized graph (with a constant-sized identifier space).

In \cite{CKP19}, it is argued why this approach works, and on a high level, the reason can be summarized as follows. For some sufficiently large constant $k$, algorithm $A$ can be executed on all graphs of a suitable constant size with a runtime of just $k$ rounds. Since each node of the original graph executing this $k$-round algorithm cannot distinguish between living in the original graph with the generated new identifiers and living in (a suitable) one of these constant-sized graphs (on all of which the algorithm is correct), the $k$-round algorithm must also be correct on the (much larger) original graph.
This is just a high-level sketch of the proof presented in \cite{CKP19}; there are a number of intricate details that have to be taken care of and are explained in \cite{CKP19}.

\paragraph{How We Proceed.}
For our purpose, we do not actually need to know the details of \cite{CKP19} on how $A'$ is constructed as a function of $A$, and we just use the following statement that comes from \cite{CKP19}: if the problem $\Pi$ can be solved in $o(\log n)$ deterministic rounds or $o(\log \log n)$ randomized rounds, then it can also be solved in $O(\log^* n)$ deterministic rounds using an algorithm that first applies \Cref{ckpstep1} and then applies \Cref{ckpstep2}. 
In fact, in our case, we are not even given the algorithm $A$ as input: we just know that the problem can be solved in $o(\log n)$ deterministic or $o(\log \log n)$ randomized rounds, but we are not given an algorithm $A$ with such a complexity. Hence, we cannot apply the construction of \cite{CKP19} directly.

In \Cref{sec:automatic}, we show that this is not an issue, in the sense that, if an algorithm exists, then it can be found by brute force. To show that, we use the following two important ingredients presented in \cite{NaorS95}:
\begin{itemize}
    \item  Any constant time algorithm that solves an \lcl{} in the \local model can be transformed into an algorithm that does not require nodes to have IDs.
    \item For every $k$, it is decidable whether there exists a $k$-round algorithm that solves a given problem in a setting where we do not have IDs and we are given a (suitable) distance-$k$ coloring. The reason is that, in this setting, there are only a finite number of possible algorithm candidates (and they can be enumerated), and given a candidate, it is possible to check if it constitutes a correct algorithm by using a centralized offline procedure. 
\end{itemize}
We use the above ingredients as follows. If we just know that $\Pi$ can be solved in $o(\log n)$ deterministic rounds or $o(\log \log n)$ randomized rounds, even if no algorithm is given, we can use \cite{CKP19} to claim that there exists a $k$ for which there is a $k$-round algorithm $B$ that solves $\Pi$ given a distance-$k$ coloring, and then use the first ingredient to claim that this algorithm does not need the presence of IDs. Finally, we use the second ingredient to say that if we try increasing values of $k$, we are going to find the algorithm $B$ that we need.

From the above discussion, in order to prove \Cref{thm:logstar}, we only need to show how to compute a distance-$k$ $O(\Delta^{2k})$-coloring in $O(\log \log^* n)$ deterministic \mpc rounds.

\subsection{\local Algorithm}
We start by presenting an algorithm for computing such a coloring in the \local model. While computing such a coloring in the \local model is easy, we present an algorithm amenable to be converted into a faster \mpc algorithm. This algorithm is not new: it has been already presented in \cite{GoldbergPS88,PanconesiR01}, and we report it here, with minor modifications, for completeness.
\begin{lemma}\label{lemma: local-k-dist}
	For any constant $k$, the distance-$k$ $O(\Delta^{2k})$-coloring problem on general graphs can be solved in the \local model with a deterministic algorithm running in $O(\log^* n)$ rounds.
\end{lemma}
\begin{proof}
	We present an algorithm that is able to compute an $O(\Delta^2)$ coloring of a given graph $G$, where $\Delta$ is the maximum degree of $G$, in $O(\log^* n)$ rounds. By simulating such an algorithm on $G^k$, the $k$-th power of $G$, which has maximum degree $\Delta^{k}$, we obtain the claimed result. Note that the running time is also asymptotically the same, since $k$ is a constant.
	
	The algorithm works as follows. At the beginning, each edge is oriented arbitrarily. Then, each node marks its incident outgoing edges with different numbers from $\{1,\ldots,\Delta\}$. In this way, we decomposed our graph $G$ into $\Delta$ edge-disjoint directed subgraphs $G_1,\ldots,G_\Delta$, where each $G_i$ is the graph induced by edges marked $i$. Also, notice that by construction, for each $i$, each node in $G_i$ has at most a single outgoing edge, and hence each $G_i$ is a directed pseudoforest.
	
	Assume we can color each directed pseudoforest with $3$ colors in $O(\log^* n)$ rounds. Then, we can obtain a proper coloring for the nodes of $G$ with $3^\Delta$ colors, by letting each node construct the tuple $c(v) = (c_1(v),\ldots,c_\Delta(v))$, where $c_i(v)$ is the color of $v$ in $G_i$. In fact, consider two neighboring nodes $u$ and $v$ connected through an edge $e$. Assume that $e$ is oriented from $u$ to $v$, and that $u$ marked $e$ with value $i$. Then, in $G_i$, $u$ and $v$ are neighbors, and hence they obtained different colors $c_i(u)$ and $c_i(v)$, implying that $c(u) \neq c(v)$. Once a $3^\Delta$-coloring is obtained, we can then spend $O(3^\Delta)$ rounds to reduce the number of colors to $O(\Delta^2)$, by using a simple greedy algorithm.
	
	We now show that each pseudoforest can be $3$-colored efficiently. Let $P$ be an arbitrary pseudotree. At first, we can use the IDs of the nodes to produce a $\poly(n)$-coloring of $P$. Then we apply $1$ round of Linial's coloring algorithm \cite{Linial92} in order to obtain an $O(\log n)$-coloring of $P$. While this step of coloring is not necessary for the \local algorithm, it allows us to reduce the amount of information that we will later need to transmit in the \mpc algorithm. Nodes can then spend $T = O(\log^* n)$ rounds to gather the color of their successors in $P$ at distance at most $T$, and it is known that, with this information, nodes can compute a proper coloring of $P$, by simulating $O(\log^* n)$ steps of a color reduction algorithm for directed paths \cite{GoldbergPS88,ColeV86}.
\end{proof}

\subsection{\mpc Implementation}

We now show how to convert the \local algorithm into an exponentially faster low-space \mpc algorithm. The \local algorithm consists of two main steps: The distance-$k$ $O(\Delta^{2k})$-coloring and the $k$-round algorithm. Since $k$ and $\Delta$ are constant, the latter step is trivial, and the former step can be computed efficiently using graph exponentiation, where nodes keep track of the IDs of the two outermost nodes, and the colors of all nodes in between. \Cref{lem:distanceColoringMPC} of the following paragraph proves the former step, completing the proof for \Cref{thm:logstar}. Component-stability and compatibility with disconnected components follows directly from the fact that all arguments are local, i.e., nodes in separate components never communicate, and that the runtime depends only on $N$.

\paragraph*{Distance-\texorpdfstring{$k$}{Lg} Coloring}

We show that the initial distance-$k$ coloring can be computed in $O(\log \log^* n)$ low-space \mpc rounds, while respecting linear global memory. First, we observe that using the standard graph exponentiation technique, we can compute the $k$th power of a graph; for constant $k$, the memory overhead is only a constant. Then, we will apply techniques similar to the ones used in the \local model in \Cref{lemma: local-k-dist}.
\begin{observation}\label{obs: mpcpower}
	For an input graph $G$ with $n$ nodes, $m$ edges, and maximum degree $\Delta$, the power graph $G^k$ can be computed deterministically in $O(\log k)$ low-space \mpc rounds with $O(\Delta^k)$ words of local and $O(m + n \cdot \Delta^k)$ words of global memory, as long as $\Delta^k<n^\delta$.
\end{observation}

\begin{observation}
	Every $k$-round \local algorithm can be simulated in $O(\log k)$ low-space \mpc rounds with $O(\Delta^k)$ words of local and $O(m + n \cdot \Delta^k)$ words of global memory, as long as $\Delta^k<n^\delta$.
	If the \local algorithm is deterministic, then the \mpc algorithm is deterministic as well.
\end{observation}
\begin{proof}
	Using \Cref{obs: mpcpower}, we can collect the $k$-hop neighborhood of each node and hence, simulate a $k$-round \local algorithm in an additional $O(1)$ low-space \mpc rounds. Observe that this also holds for general graphs.
\end{proof}

\begin{lemma}
	\label{lem:distanceColoringMPC}
	The distance-$k$ $O(\Delta^{2k})$-coloring problem on general graphs can be solved in the low-space \mpc model with a $O(\log \log^* n + \log k)$-time deterministic algorithm, as long as $\Delta^k<n^\delta$. 
	The algorithm requires $O(\Delta^k)$ words of local and $O(m + n \cdot \Delta^k)$ words of global memory.
	If $k$ and $\Delta$ are constants, the runtime reduces to $O(\log \logstar n)$ and we require $O(1)$ words of local and $O(m+n)$ words of global memory.
\end{lemma}

\begin{proof}
	Using \Cref{obs: mpcpower}, we can first compute $G^{k}$ in $O(\log k)$ rounds, and operate on $G^{k}$ instead of the input graph $G$ henceforth. The application of \Cref{obs: mpcpower} requires $O(\Delta^k)$ words of local memory and $O(m + n \cdot \Delta^k)$ words of global memory.
	Then, similarly to \Cref{lemma: local-k-dist}, we can reduce the coloring problem to $O(1)$-coloring of directed pseudoforests that are initially colored with $O(\log n)$ colors.
	
	Next, our goal is to use the graph exponentiation technique such that each node can collect the topology and the colors of its $O(\log^* n)$ successors in its pseudoforest in $O(\log \log^* n)$ time.
	Here, we have to take care of the subtle detail that the \emph{color} of a successor is not enough to determine the machine on which this successor lies.
	Suppose that each node is initially labeled with its $O(\log \log n)$-bit color and its $O(\log n)$-bit identifier that encodes both the identity (color) of the node and the machine containing the node.
	Then, in round 1, each node knows the identifier and the color of its successor.
	For an inductive argument, suppose that each node $u$ knows the identifier the successor $v_i$ in distance $i$ and the vector of colors of all nodes in between $u$ and $v_i$, on the directed path from $u$ to $v_i$.
	Then, in $O(1)$ \mpc rounds, $u$ can learn the identifier of the $2i$:th successor $v_{2i}$ and the colors of all nodes between $u$ and $v_{2i}$.
	After learning the identifier of $v_{2i}$, node $u$ can forget about the identifier of $v_i$ and hence, $u$ only keeps track of one identifier.
	By induction, node $u$ learns the colors of its $O(\log^* n)$ successors in $O(\log \log^* n)$ \mpc rounds.
	
	Using the vector of colors of the successors, in $O(1)$ \mpc rounds, each node can simulate the $O(\log^* n)$-time \local algorithm to obtain an $O(\Delta^{2k})$-coloring.
	This requires $O(\log^* n \cdot \log \log n + \log n) = O(\log n)$ \emph{bits} of memory per node per pseudoforest that the node belongs to, counting the colors of the successors and the identifier of the furthest successor.
	Altogether, this results in a global memory requirement of $O(n \log n \cdot \Delta^k)$ bits which fits $O(n \cdot \Delta^{k})$ words.
\end{proof}

\section{The High Regime}\label{sec:high}

In this section, we will prove that all solvable \lcl problems on forests, i.e., \emph{all \lcl problems that have a correct solution on every forest}, can be solved deterministically in $O(\log n)$ time in the low-space \mpc model using $O(m+n)$ words of global memory. Our proof is constructive: we explicitly provide, for any solvable \lcl, an algorithm that has a runtime of $O(\log n)$. In fact, our construction can be used to find an $O(\log n)$-time algorithm $\fA$ \emph{even for unsolvable \lcl{}s}, with the guarantee that on any instance that admits a correct solution the given output will be correct (while the algorithm detects it if no solution exists). We show the following theorem.

\begin{theorem} \label{thm:high}
	For any solvable \lcl problem $\Pi$ on a forest, there is an $O(\log n)$-time deterministic low-space \mpc algorithm that is component-stable and uses $O(m+n)$ words of global memory.
\end{theorem}
The runtime bound of \Cref{thm:high} follows from  \Cref{cor:fewernodes},  \Cref{lem:gi,lem:runtimephaseone,lem:runtimephasetwo},  and the implementation details that we provide in \Cref{sec:highimple}. Its correctness is proven in \Cref{lem:gi,lem:wellandcorrect}.
We elaborate on this in \Cref{sec:proofhigh}.
In particular, in \Cref{sec:proofhigh} we provide a method to solve any LCL on forests if we can solve it on trees.
Hence, w.l.o.g., we can restrict attention to trees and will do so for the remainder of the discussion of the high regime.

\subsection{High-level Overview of the Algorithm and Its Analysis}\label{sec:overviewhcr}

Consider an arbitrary solvable \lcl problem $\Pi$ on trees.
Throughout this section, we will assume that the \lcl is given as a node-edge-checkable \lcl (\Cref{def:nodeedge}), which we can do w.l.o.g., as observed in \Cref{sec:prelim}. In the following, we will give a slightly simplified view of the algorithm $\fA$ we will use to solve $\Pi$ in $O(\log n)$ time. On a high level, algorithm $\fA$ proceeds in $2$ phases. Assume that already before the first phase we root the input tree by using the algorithm described in \Cref{sec:rooting}.

In the first phase, which we will refer to as the \emph{leaves-to-root phase}, roughly speaking, the goal is to compute, for a substantial number of edges $e = (u, v)$, the set of output labels that can be output at half-edge $(v, e)$ such that the label choice can be extended to a (locally) correct solution in the subtree hanging from $v$ via $e$.
This is done in an iterative manner, proceeding from the leaves towards the root. When, at last, the root has computed this set of output labels for each incident half-edge, it can, on each such half-edge, select an output label from the computed set such that the obtained node configuration is contained in the node constraint of $\Pi$ and the input-output constraints of $\Pi$ (given by the function $\gee$ in the definition of $\Pi$) are satisfied. Such a selection must exist due to the fact that $\Pi$ has a correct solution on the considered instance. We refer to these sets as the \emph{completability information}. 

The second phase, which we will refer to as the \emph{root-to-leaves phase}, consists of completing the solution from the root downwards, by iteratively propagating the selected solution further towards the leaves.
With the same argumentation as at the root, certain nodes $v$ can select an output label at the half-edge leading to its parent and output labels from the sets computed on its incident half-edges leading to its children such that the obtained node configuration is contained in the node constraint of $\Pi$, the obtained edge configuration on the edge from $v$ to its parent is contained in the edge constraint of $\Pi$, and the input-output constraints of $\Pi$ are satisfied.
The fact that the selected labels come from the sets computed in the first phase ensures that after each choice the current partial solution is part of a correct global solution.
While this outline sounds simple, there are a number of intricate challenges to make the mentioned ideas work in $O(\log n)$ rounds while staying within the memory bound of $O(m+n)$. 

Unfortunately, if the depth of the input tree is $\omega(\log n)$ the outlined approach has  $\omega(\log n)$ steps and running them sequentially is insufficient for an  $O(\log n)$-time algorithm. In order to mitigate this issue, we will not only process the leaves of the remaining unprocessed tree in each iteration, but also the nodes of degree $2$, inspired by the rake-and-compress decomposition by Miller and Reif \cite{MillerReif89} which guarantees that after $O(\log n)$ iterations of removing all degree-$1$ and degree-$2$ nodes all nodes have been removed. The advantage of degree-$2$ nodes over higher-degree nodes w.r.t.\ storing completability information (as in the above outline) is that they form paths, which by definition only have two endpoints; the idea, when processing such a path, is to simply store in the two endpoints the information for which pairs of labels at the two half-edges at the ends of the path there exists a correct completion of the solution inside the path. This allows to naturally add processing degree-$2$ nodes to the leaves-to-root phase, while for the root-to-leaves phase, the information stored at the endpoints $s, t$ of a path essentially allows us to start extending the current partial solution on the path itself (and thereafter on the subtrees \emph{hanging} from nodes on the path) one step after the output labels at $s$ and $t$ are selected. Note that the degrees of nodes change throughout the process due to the removal of nodes  and hence new nodes might become degree-$2$ nodes after every step of the algorithm.

Unfortunately, there are further challenges in obtaining an $O(\log n)$ runtime. In the leaves-to-root phase, even when using graph exponentiation, processing a path of degree-$2$ nodes of length $L$ involves coordination between its endpoints and takes $\Omega(\log L)$ time, whereas the $O(\log n)$ time guarantee of the rake-and-compress technique crucially relies on the fact that each iteration (optimally, an iteration would remove all leaves and all degree-$2$ nodes)  can be performed in constant time.
Hence, essentially, we will only perform one step of graph exponentiation on paths in each iteration.
Here, a new obstacle arises: before the graph exponentiation is finished, new nodes (that just became degree-2 nodes due to all except one of their remaining children being conclusively processed in the most recent iteration) might join the path.
Nevertheless, we will show that this process still terminates in logarithmic time by designing a fine-tuned potential function that is inspired by the idea of counting how many nodes from certain groups of degree-$2$ nodes are contained in any fixed ``pointer chain'' from some leaf to the root.

Another issue is that we have to be able to store the completability information (recall, the sets) that we compute in the leaves-to-root phase until we use it (again) in the root-to-leaves phase.
Recall that the graph exponentiation technique adds new edges/pointers. Even on paths their number can be up to logarithmic in $n$ per node (even on average), yielding  a logarithmic overhead in global memory.

In order to remedy this problem, we perform preprocessing before the leaves-to-root phase, and, as a result thereof, postprocessing after the root-to-leaves phase. The preprocessing can be thought of as a more memory-efficient (hence relatively slower) version of (a few iterations in) the leaves-to-root phase. It differs by processing the degree-$2$ nodes, i.e., paths, in a way that guarantees that the number of new edges introduced by the graph exponentiation (which we should rather call pointer forwarding at this point) on each path in each iteration is only a constant fraction of the length of the respective path.
This is achieved by finding, in each iteration, a maximal independent set (MIS) on each path, letting only MIS nodes forward pointers, and removing the MIS nodes afterwards.
The preprocessing runs for $\Theta(\log \log n)$ iterations, and computing an MIS on paths in each of them takes $O(\logstar N)$ time, where $N$ is the size of the ID space. Note that due to the removal of vertices and the way we treat paths, new paths can appear in each iteration and we need to pay the $O(\logstar N)$ runtime in each iteration, yielding a runtime of $O(\log\log n \cdot \logstar N)$ for the preprocessing, which is much less than the target runtime of $O(\log n)$ rounds.

We will show that the number of remaining nodes is $O(n/\log n)$ after the preprocessing. 
This property ensures that the memory overhead of $O(\log n)$ edges per node introduced in the leaves-to-root phase does not exceed the desired global memory of $O(m+n)$ words.
The postprocessing runs for $\Theta(\log \log n)$ iterations and is conceptually very similar to the preprocessing. We simply iteratively extend the partial solution (computed so far) on the edges that were processed during preprocessing, analogous to the approach in the root-to-leaves phase.
Lastly, we also have to ensure that the local memory restrictions of low-space \mpc are not exceeded; we take care of this in \Cref{sec:highimple}.

\subsection{The Algorithm}\label{sec:highalgo}

In this section, we provide the desired algorithm that can be implemented in $O(\log n)$ time in the low-space \mpc model and prove its correctness. 
The details about the exact implementation in the \mpc model are deferred to \Cref{sec:highimple}.
Let $\Pi = (\sinn, \sout, \noco, \edco, \gee)$ be the considered \lcl, and let $G$ denote the input tree.
Before describing the algorithm, we need to introduce the new notion of a \emph{compatibility tree}.
In a sense, a compatibility tree is a structure that stores the constraints that a given \lcl imposes for a given input tree, i.e., the constraints that two labels on an edge or $\deg(v)$ labels around a node $v$ have to satisfy (as well as which output labels can be used at which half-edge, which the \lcl encodes via input labels) are explicitly encoded on the edge and around the node.

\begin{definition}\label{def:comptree}
	A \emph{compatibility tree} is a rooted tree $T$ (without input labels)
	where each edge $(u,v)$ is labeled with a subset $S_{uv}$ of $\sout \times \sout$, and each node $w$ is labeled with a tuple $S_w$ consisting of tuples of the form $(s_w^e)_{e \in \incid(w)}$ where $\incid(w)$ denotes the set of edges incident to $w$, and $s_w^e \in \sout$ for each $e \in \incid(w)$.
	A \emph{correct solution} for a compatibility tree is an assignment $\gout \colon H(T) \to \sout$ s.t.
	\begin{enumerate}
		\item for each edge $e = (u,v)$, we have $(\gout((u,e)), \gout((v,e))) \in S_{uv}$, and
		\item for each node $w$, there exists a tuple $(s_w^e)_{e \in \incid(w)} \in S_w$ such that, for each edge $e \in \incid(w)$, we have $\gout((w,e)) = s_w^e$.
	\end{enumerate}
\end{definition}

Now, we are set to describe the desired algorithm.
The algorithm starts by rooting $G$, using the method described in \Cref{sec:rooting}. We denote the root by $\rooot$. Then, we transform $G$ (which from now on will denote the rooted version of the input tree) into a compatibility tree $G'$ by iteratively removing nodes of degree $1$ and $2$ (while suitably updating the edge set) and assigning a subset $S_{uv}$, resp.\ $S_w$, to each remaining edge $(u,v)$, resp.\ remaining node $w$. We call this step preprocessing (resp. postprocessing when extending the solution back to the removed nodes) and formally define it in \Cref{sec:trafothere}. Next, we design an algorithm $\fA'$ that computes a correct solution for the compatibility tree $G'$. Algorithm $\fA'$ is divided into algorithms $\aone$ and $\atwo$, which largely correspond to Phase I (\Cref{sec:phaseone}) and Phase II (\Cref{sec:phasetwo}) mentioned in the high-level overview. Algorithm $\aone$ finds a correct solution for $G'$, and algorithm $\atwo$ transforms the obtained solution into a correct solution for \lcl $\Pi$ on $G$.

\subsubsection{Reducing the \lcl to a Small Compatibility Tree (Pre- and Postprocessing)}\label{sec:trafothere}
In this section, we show how to transform the rooted tree $G$ into a compatibility tree with $O(n/\log n)$ nodes, where $n$ is the number of nodes of $G$, and how to transform any correct solution for $G'$ into a correct solution for the given \lcl $\Pi$ on $G$.
In other words, we show how to reduce the problem of solving $\Pi$ on $G$ to the problem of finding a correct solution for a compatibility tree with fewer nodes.
The idea behind this approach is that the new, smaller instance can be solved in logarithmic time without exceeding the desired global memory of $O(m+n)$ words.
To obtain a good overall runtime, we will also show how to perform the reduction (and recover the solution) in $O(\log \log n \cdot \log^* N) = O(\log \log n \cdot \log^* n)$ rounds. Recall, that $N$ denotes the size of our ID space.

We start by describing how to obtain $G'$ from $G$.
To this end, we will first transform $G$ into a compatibility tree $G_0$, and then iteratively derive a sequence $G_1, G_2, \dots, G_t$ of compatibility trees from $G_0$, where $t \in O(\log \log n)$ is a parameter we will choose later.

We define $G_0$ in the natural way, by essentially encoding the given \lcl $\Pi$.
The nodes and edges of $G_0$ are precisely the same as in $G$.
For any edge $(u,v)$ with input labels $\iota_u$ and $\iota_v$ at the two half-edges belonging to $(u,v)$, we set $S_{uv}$ to be the set of all pairs $(\ell, \ell') \in \sout \times \sout$ such that the multiset $\{ \ell, \ell' \}$ is contained in the edge constraint $\edco$ of $\Pi$, and we have $\ell \in \gee(\iota_u)$ and $\ell' \in \gee(\iota_v)$.
For any node $w$, we set $S_w$ to be the set of all tuples $(\ell^e)_{e \in \incid(w)}$ such that the multiset $\{ \ell^e \mid e \in \incid(w) \}$ is contained in $\noco_{\deg(w)}$.
Note that the asymmetric nature of this definition comes from the fact that we only need to require compatibility with the function $\gee$ once, in the constraints for nodes or (as we chose) for edges.
From the definition of $G_0$, we obtain directly the following observation.

\begin{observation}\label{obs:equiv}
	A half-edge labeling of $G$ is a correct solution for \lcl $\Pi$ if and only if it is a correct solution for the compatibility tree $G_0$ (under the natural isomorphism between $G$ and the graph underlying $G_0$).
\end{observation}

We now describe how to obtain $G_i$ from $G_{i-1}$, for any $1 \leq i \leq t$.
We transform $G_{i-1}$ into $G_i$ in two steps.
In the first step, we start by finding an MIS $Z$ on the subgraph of $G_{i-1}$ induced by all nodes of degree precisely $2$.
Then, for each node $v \in Z$ with incident edges $e = (u,v)$ and $e' = (v,w)$, we remove $e$ and $e'$ from $G_{i-1}$ and replace them by a new edge $e'' = (u,w)$.
Furthermore, for the new edge, we set $S_{uw}$ to be the set of all label pairs $(\ell, \ell') \in \sout \times \sout$ such that there exist labels $\ell_1, \ell_2$ satisfying $(\ell, \ell_1) \in S_{uv}$, $(\ell_2, \ell') \in S_{vw}$, $s^e_v = \ell_1$, and $s^{e'}_v = \ell_2$.
From the perspective of the nodes $u$ and $w$, the new edge $e''$ replaces the old edges $e$ and $e'$, respectively, in the indexing hidden in the definition of the tuples $S_u$ and $S_w$.
Call the obtained graph $G'_{i-1}$.

In the second step, executed after the first step has finished, each edge $e^* = (x,y)$ such that $x$ is a leaf is removed together with $x$.
Moreover, for such a removed edge, we set $S_y$ to be the set of all tuples $(\ell^e_y)_{e \in \incid(y)}$ such that there exist labels $\ell, \ell'$ such that (in $G'_{i-1}$) we have $(\ell, \ell') \in S_{xy}$ and there exists some tuple $(s^e_y)_{e \in \incid(y)}$ with $s^e_y = \ell^e_y$ (for all $e \neq e^*$) and $s^{e^*}_y = \ell'$.
(Here the first occurrence of $\incid(y)$ denotes the set of edges incident to $y$ after removing $e^*$, while the second occurrence denotes the set before removing $e^*$.)
If a node $y$ of $G'_{i-1}$ has multiple children that are leaves, then we can think of removing the respective edges one by one, each time updating $S_y$.
However, for the actual computation, node $y$ can perform all of these steps at once.

We obtain the following lemma. 

\begin{lemma}\label{lem:gi}
	Let $1 \leq i \leq t$.
	If there exists a correct solution for $G_{i-1}$, then there also exists a correct solution for $G_i$.
	Moreover, given any correct solution for $G_i$, we can transform it into a correct solution for $G_{i-1}$ in a constant number of rounds in the low-space \mpc model using $O(m+n)$ words of global memory.
	Finally, given $G_{i-1}$, we can compute $G_i$ in $O(\log^* N) = O(\log^* n)$ rounds in the described setting.
\end{lemma}

\begin{proof}
	The first statement follows directly from the definition of $G_i$.
	For the second statement, observe that from the definition of $G_i$, it follows that any correct solution for $G_i$ provides a partial solution for $G_{i-1}$ (under the natural transformation that subdivides edges and adds the ``removed'' leaves with their incident edges) that is part of a correct solution for $G_{i-1}$ (and this factors through $G'_{i-1}$ in the obvious way).
	Hence, we can first obtain a correct solution for $G'_{i-1}$ by extending the provided solution on the removed leaves with their incident edges, and then obtain a correct solution for $G_{i-1}$ by doing the same on the subdivided edges.
	Note that the first extension can be performed by the nodes that are incident to the leaves (all of which have only constantly many output labels to determine), and the second extension by the computed nodes in the MIS $Z$ (which we can do in parallel since no two nodes in $Z$ are neighbors).
	The third statement follows from the definition of $G_i$, the fact that an MIS can be computed in $O(\log^* N) = O(\log^* n)$ rounds (already in the \local model), and the above observation about parallelization.
\end{proof}

Next, we bound the number of nodes of $G_i$, which we denote by $n_i$.

\begin{lemma}\label{lem:rcanalysis}
	For any $1 \leq i \leq t$, we have $n_i \leq 2/3 \cdot n_{i-1}$.
\end{lemma}
\begin{proof}
	By the construction of $G_i$, all nodes that are contained in $G_{i-1}$ but not in $G_i$ are either leaves or degree-$2$ nodes in $G_{i-1}$.
	In particular, as the set of leaves is the same in $G_{i-1}$ and $G'_{i-1}$, all leaves of $G_{i-1}$ are not contained in $G_i$.
	Regarding degree-$2$ nodes in $G_{i-1}$, we observe that at least a third of them must be part of the chosen MIS since (i) each degree-$2$ node must be in the MIS or have an MIS node as neighbor, and (ii) each MIS node covers at most three nodes (in the sense that it is equal or adjacent to them).
	
	Also, since the average degree of a node in a tree is below $2$, the number of leaves in a tree is larger than the number of nodes of degree at least $3$. Hence, when going from $G_{i-1}$ to $G_i$, at least half of the nodes of degree $\neq 2$ are removed, and in total we obtain that the number of nodes that are removed is at least $1/3 \cdot n_i$, which proves the lemma.
\end{proof}

Now, by setting $t \coloneqq 2 \log \log n$ and $G' \coloneqq G_t$, we obtain the following straightforward corollary.
\begin{corollary}\label{cor:fewernodes}
	The number of nodes of $G'$ is at most $n/\log n$.
\end{corollary}

Moreover, by Observation~\ref{obs:equiv}, \Cref{lem:gi}, and \Cref{lem:rcanalysis}, we know that we can compute $G'$ in $O(\log \log n \cdot \log^* N) = O(\log \log n \cdot \log^* n)$ rounds, that there is a correct solution for $G'$ (provided the \lcl $\Pi$ admits a correct solution on $G$), and that we can transform any correct solution for $G'$ into a correct solution for $\Pi$ on $G$ in $O(\log \log n)$ rounds.
The stated runtimes are under the premise that we can implement all of the $O(\log \log n)$ steps without running into memory issues, which we will show to be the case in \Cref{sec:highimple}.
(Note that \Cref{lem:gi} only makes statements about single steps.)

\subsubsection{Phase I (leaves-to-root)}\label{sec:phaseone}
In Phase I, we maintain a set of pointers $(u,v)$, which encode the output labels that can be chosen at $u$ and $v$ such that the solution can be correctly completed on the path between $u$ and $v$ and the subtrees hanging from this path.
The goal is to increase the lengths of these pointers, until we obtain leaf-to-root pointers that allow us (in Phase II) to fix output labels at the root that can be completed to a correct solution on the whole tree.
The algorithm in Phase I proceeds in iterations $i = 1, 2, \dots$.
Before explaining the steps taken in each iteration, we need to introduce some definitions.

A pointer is simply a pair $(u, v)$ of nodes such that $v$ is an ancestor of $u$ in $G'$, i.e., $v$ is a node on the path from $u$ to the root $\rooot$, and $v \neq u$.
We say that a pointer $(u, v)$ \emph{starts} in $u$ and \emph{ends} in $v$; we also call $(u, v)$ an \emph{incoming} pointer when considering node $v$, and an \emph{outgoing} pointer when considering node $u$. 
On each pointer $p = (u, v)$, we store several pieces of information which are required for Phase II:
\begin{enumerate}
	\item a set $\pairs_p \subseteq \sout \times \sout$ (encoding completability information as outlined above),
	\item a node $\pred_p$ which might also be empty, i.e., $\pred_p \in V(G') \cup \{\bot\}$ (encoding the information about which node created the pointer), and
	\item pair $(\first_p, \last_p)$ where $\first_p$ is the first and $\last_p$ the last edge on the unique path from $u$ to $v$ (possibly $\first_p = \last_p$).
\end{enumerate}

The initial pointer set $\ps_0$ is set to $E(G)$, where for each pointer $p = (u,v)$, we set $\pairs_p \coloneqq S_{uv}$, $\pred_p = \bot$, and $\first_p = \last_p = (u,v)$.
We also maintain a set of \emph{active pointers} which is initially set to $\aps_0 \coloneqq \ps_0$.
Throughout Phase I, we will guarantee that the set of active pointers contains, for each node $u \neq \rooot$, at most one pointer starting in $u$, and no pointers starting in the root $\rooot$.
We call a node \emph{active} in iteration $i$ if it is the root $\rooot$ or has exactly one outgoing active pointer at the end of iteration $i - 1$, i.e., one outgoing pointer in $\aps_{i-1}$.
For active nodes, we will denote the unique pointer starting in node $u \neq \rooot$ by $p(u)$.
Finally, for each node $u$ that is not a leaf, we also maintain a tuple $\comp(u) = (\comp^e(u))_{e \in \incid(u)}$ such that each element $\comp^e(u)$ is either a subset of $\sout$ or the special label $\undec$.
(The purpose of these tuples is to store completability information about subtrees hanging from $u$ via different edges).
The multiset for node $u$ is initially set to $\comp_0(u) \coloneqq (\undec, \dots, \undec)$. 

In iteration $i$, the pointer set is updated from $\ps_{i-1}$ to $\ps_i$, the active pointer set from $\aps_{i-1}$ to $\aps_i$, and, for each non-leaf node $u$, the tuple $\comp_{i-1}(u)$ is updated to $\comp_i(u)$.
For each $i \geq 0$, we will ensure that $\aps_i \subseteq \ps_i$ and $\ps_i \subseteq \ps_{i+1}$ (if both are defined).
In order to specify the precise update rules, we need two further definitions.

The first definition specifies which nodes (during the iterative process explained above) should be intuitively regarded as degree-$2$ nodes (because for all incident edges (to children) except one, the corresponding subtree has already been completely processed w.r.t.\ completability information) and which as nodes of degree at least $3$.
(Degree-$1$ nodes will only play a passive role in the update rules.)
The second definition provides an operation that combines two pointers into a larger one.

\begin{definition}[$k$-nodes]\label{def:123node}
	We call a node $u \neq r$ a $2$-node if $u$ has at least one incoming active pointer and for any two incoming active pointers $p, p'$, we have $\last_p = \last_{p'}$.
	For a $2$-node $u$, we call the unique incoming edge $e$ satisfying that for any incoming active pointer $p$ we have $\last_p = e$, the \emph{relevant in-edge} of $u$.
	We call a node $u \neq r$ a $3$-node if $u$ has (at least) two incoming active pointers $p, p'$ satisfying $\last_p \neq \last_{p'}$.
	For a node $u$ that is a $3$-node or the root $\rooot$, an incoming edge $e$ is called a \emph{relevant in-edge} of $u$ if there is an incoming active pointer $p$ satisfying $\last_p = e$.
	We call a node $u \neq r$ a $1$-node if $u$ has no incoming active pointer.
\end{definition}

\begin{definition}[Merge]\label{def:merge}
	Let $v$ be a $2$-node, $p = (u, v)$ and $p' = (v, w)$ two active pointers starting and ending in $v$, respectively, and $\comp(v) = \{ \comp^e(v) \}_{e \in \incid(v)}$ the current tuple at $v$.
	(Note that our construction of the update rules will guarantee (as shown in Observation~\ref{obs:undec}) that (for any $2$-node $v$) we have $\comp^e(v) = \undec$ if and only if $e = \last_p$ or $e = \first_{p'}$.)
	Then, we set $\merge(p, p') \coloneqq (u, w)$.
	Furthermore, we set $\pairs_{\merge(p, p')}$ to be the set of all label pairs $(\ell, \ell')$ such that there exists a tuple $(\ell^e)_{e \in \incid(v)}$ of output labels s.t.
	\begin{enumerate}
		\item $(\ell^e)_{e \in \incid(v)} \in S_v$,
		\item $\ell^e \in \comp^e(v)$, for each $e \in \incid(v) \setminus \{\last_p, \first_{p'}\}$, and
		\item $(\ell, \ell^{\last_p}) \in \pairs_p$ and $(\ell^{\first_{p'}}, \ell') \in \pairs_{p'}$.
	\end{enumerate}
	Finally, we set $\pred_{\merge(p, p')} \coloneqq v$, $\first_{\merge(p, p')} \coloneqq \first_p$, and $\last_{\merge(p, p')} \coloneqq \last_{p'}$.
\end{definition}

The above definition ensures that the new pointer $(u, w)$ satisfies the property that $w$ is an ancestor of $u$, i.e., $(u, w)$ is indeed a pointer.
Now, we are set to define precisely how the sets and tuples we maintain throughout the iterations change.
The update rules for iteration $i\geq 1$ are as follows.
We emphasize that update rules~\ref{step2} and~\ref{step3} are executed ``in parallel'' for all $2$-nodes, $3$-nodes, and the root, i.e., there are no dependencies between these steps. Algorithm $\aone$ is defined as follows.

\begin{enumerate}
	\item Start by setting $\aps_i \coloneqq \ps_i \coloneqq \emptyset$, and $\comp_i(u) \coloneqq \comp_{i-1}(u)$ for each node $u$.
	\item\label{step2} For each active $2$-node $u$ (with outgoing pointer $p(u)$), and each incoming active pointer $p \in \aps_{i-1}$, add the pointer $\merge(p, p(u))$ to $\aps_i$.
	\item\label{step3} For each node $u$ that is a $3$-node or the root $\rooot$ do the following.
	Start by asking, for each relevant in-edge $e = (v, u)$ of $u$, whether there is some incoming active pointer $p = (w, u) \in \aps_{i-1}$ such that $\last_p = e$ and $w$ is a leaf.
	If the answer is ``no'' for at least one relevant in-edge, or if $u = r$, then
	\begin{enumerate}
		\item\label{itema} for all relevant in-edges $e$ for which the answer is ``no'', add all incoming active pointers $p' \in \aps_{i-1}$ with $\last_{p'} = e$ to $\aps_i$, and
		\item\label{itemb} for all relevant in-edges $e = (v, u)$ for which the answer is ``yes'', (only) change $\comp^e_i(u)$ (from $\undec$) to the set of all labels $\ell$ satisfying that there exists a pair $(\ell', \ell'') \in \pairs_p$ with $\ell'' = \ell$ and $(\ell') \in S_w$.
	\end{enumerate}
	If the answer is ``yes'' for all relevant in-edges and $u \neq r$, then change the answer to ``no'' on precisely one arbitrarily chosen relevant in-edge, and proceed as in the previous case (i.e., execute steps~\ref{itema} and~\ref{itemb}).
	\item For each node $u \neq \rooot$ that has an outgoing active pointer, set $p(u)$ to be the unique pointer in $\aps_i$ starting in $u$.
	\item Set $\ps_i \coloneqq \ps_{i-1} \cup \aps_i$.
\end{enumerate}

Algorithm $\aone$ terminates after the first iteration $i$ satisfying $\aps_i = \emptyset$.

\begin{figure}
	\centering
	\includegraphics[width=0.85\textwidth]{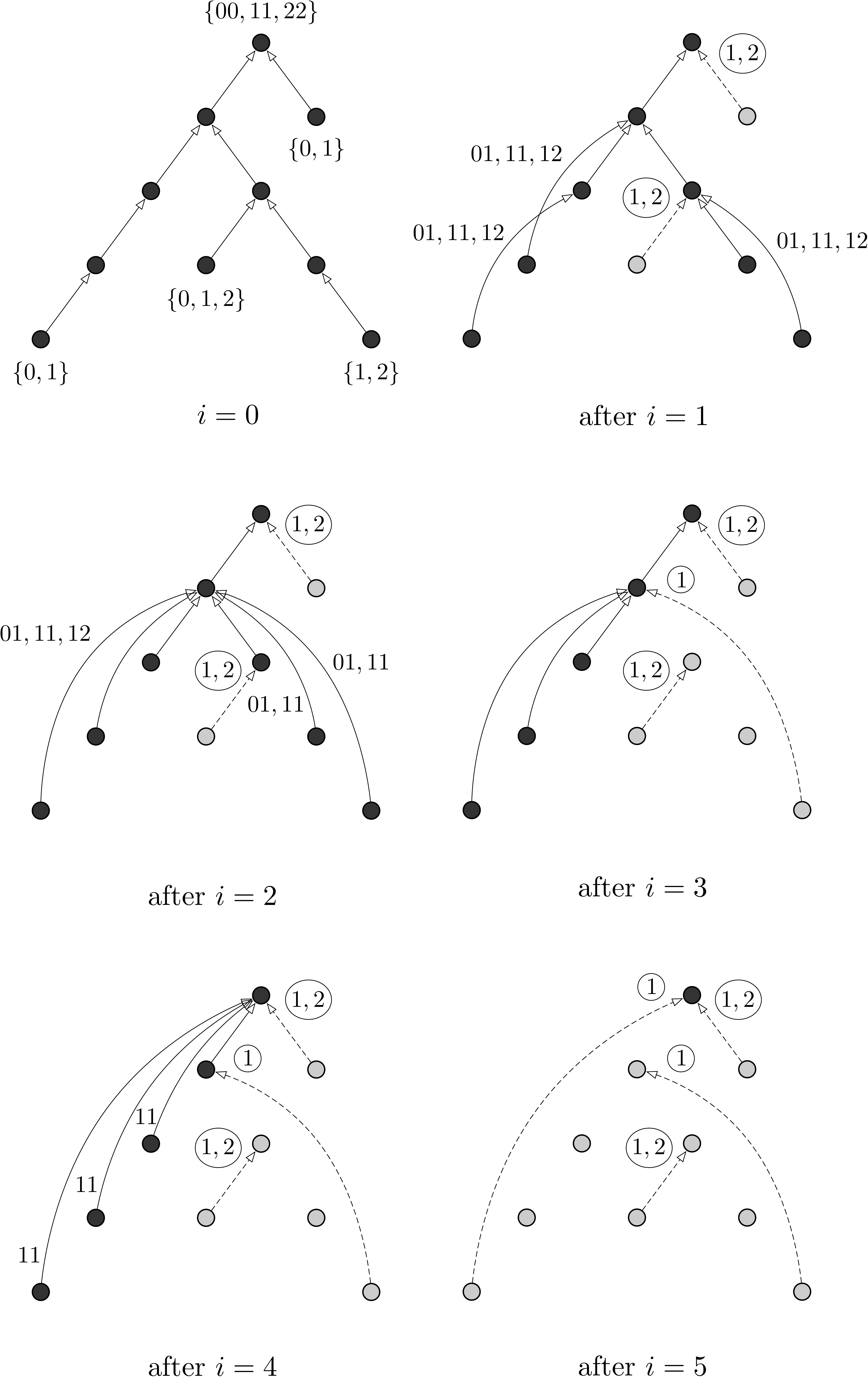}
	\caption{An example execution of the algorithm in Phase I.}
	\label{fig:pointers}
\end{figure}

\paragraph*{An Example for Phase I (leaves-to-root).}
Consider the compatibility tree given in \Cref{fig:pointers}, for $i=0$.
The considered output label set is $\sout = \{ 0, 1, 2 \}$.
For each edge $(u,v)$ in the tree, the associated label pair set $S_{uv}$ is defined as ${(0,1), (1,2)}$, which, for simplicity, we will write in the form $01, 12$.
For each node $w$ that is not a leaf or the root, the associated set $S_w$ is defined as the set of all tuples (of length $\deg(w)$) with pairwise distinct entries (e.g., if $w$ is of degree $3$, each tuple in $S_w$ is a permutation of $(0,1,2)$).
For the leaves and the root, the associated tuple set is given below or above the node.

In iteration $i = 1$, three merge operations are performed by the three non-root nodes of degree $2$.
In each merge operation, two pointers $p = (u,v)$, $p' = (v,w)$, each labeled with $01, 12$, are merged into a larger pointer $(u,w)$ with label $01, 11, 12$.
The label of the new pointer contains, e.g., the pair $12$ since labeling half-edge $(u,(u,v))$ with $1$ and half-edge $(w,(v,w))$ with $2$ can be completed to a labeling of all inbetween half-edges such that the labeling respects the constraints $S_{uv}$, $S_v$, and $S_{vw}$, namely by labeling $(v,(u,v))$ with $2$, and $(v,(v,w))$ with $1$.
Moreover, the two dashed pointers are removed from the set of active pointers in iteration $1$ as they start in a leaf $x$ and end in a $3$-node or the root.
Consequently, the two leaves in which the pointers start become inactive in iteration $1$ (which is illustrated by coloring them gray).
The two nodes $y$ in which the pointers end update their set $\comp(y)$, by setting the entry corresponding to the removed pointer from $\undec$ to the set of all labels that (when written at half-edge $(y,(x,y))$) are completable downwards, i.e., for which a label at the respective leaf $x$ exists that respects $S_x$ and $S_{xy}$.
We illustrate the entries that change from $\undec$ to some set by circling them, e.g., the changed entry of $\comp(\rooot)$ is the circled set consisting of the labels $1$ and $2$.

In iteration $i=2$, three merge operations are performed, by the two non-root $2$-nodes.
In iteration $i=3$, no merge operations are performed as there is no $2$-node.
Furthermore, the left child $z$ of the root (which is a $3$-node) obtains ``yes'' as answer for the question it asks in step~\ref{step3} of the update rules, for both relevant in-edges.
Thus, $z$ changes the answer to ``no'' for one of the relevant in-edges, arbitrarily chosen (in our case the left one).
For the other relevant in-edge $e$, all pointers $p$ with $\last_p = e$ are removed from the set of active pointers, and the corresponding entry in the set $\comp(z)$ is changed from $\undec$ to $\{ 1 \}$.
This also causes three nodes to become inactive.
In iteration $i=4$, the newly born $2$-node performs three merge operations.
In iteration $i=5$, all pointers ending in the root are removed from the set of active pointers due to the existence of a pointer starting in a leaf with the same ``last edge''.
This step causes the set of active pointers to become empty, upon which the algorithm in Phase I terminates.
In \Cref{sec:appphaseone}, we show that Phase I is well-defined and analyze it.
In particular, we prove the following lemma that bounds the number of iterations in Phase I.

\begin{lemma}\label{lem:runtimephaseone}
	Algorithm $\aone$ terminates after $O(\log n)$ iterations.
\end{lemma}

\subsubsection{Phase II (root-to-leaves)}\label{sec:phasetwo}

Let $\ps_{\fin}$ denote the set of pointers at the end of the last iteration of Phase I.
In Phase II, we will go through the pointers of some subset of $\ps_{\fin}$ in some order and ``fix'' them, i.e., for each such pointer $p = (u, v)$ we assign to the two half-edges $(u, \first_p)$ and $(v, \last_p)$ a label from $\sout$ each.
In order to describe the order in which we process the pointers, we group the pointers we want to process into sets $\timey(1), \timey(2), \dots$.
We will process each of the pointers in set $\timey(i)$ in parallel in iteration $i$.

Define $\timey(1)$ to be the set of all pointers $p = (u, \rooot) \in \ps_{\fin}$ for which $u$ is a leaf.
For each $i \geq 2$, define $\timey(i)$ to be the set of all pointers $p'$ s.t.~there is a pointer $p = (u, v) \in \timey(i-1)$ satisfying (1) $p' = (u, \pred_p)$, (2) $p' = (\pred_p, v)$, or (3) $p' = (w, \pred_p)$ where $w$ is a leaf and the edge $\last_{p'}$ does not lie on the path from $u$ to $v$.

Next, we collect some insights about the pointers in $\timey(i)$.
The proofs are deferred to \cref{sec:appphasetwo}.
We start with \Cref{lem:attheroot} which provides information about the leaf-root pointers produced in Phase I.
We continue with \Cref{lem:nicesplit} which highlights which pointers in $\timey(i)$ are ``produced'' by some pointer in $\timey(i-1)$.

\begin{lemma}\label{lem:attheroot}
	For each edge $e$ incoming to the root $\rooot$, there is precisely one pointer $p = (u, \rooot) \in \ps_{\fin}$ such that $u$ is a leaf and $\last_p = e$.
\end{lemma}

\begin{lemma}\label{lem:nicesplit}
	Let $p = (u,v)$ be a pointer in $\ps_{\fin}$ with $\pred_p \neq \bot$.
	Then $\pred_p$ has degree at least $2$ in $G'$.
	Moreover,
	\begin{enumerate}
		\item if $\pred_p$ has degree $2$, then $(u, \pred_p), (\pred_p, v) \in \ps_{\fin}$, and
		\item if $\pred_p$ has degree at least $3$, then $(u, \pred_p), (\pred_p, v) \in \ps_{\fin}$, and for each edge $e$ incoming at $\pred_p$ that does not lie on the path from $u$ to $v$, there is exactly one pointer $p' = (w, \pred_p) \in \ps_{\fin}$ such that $w$ is a leaf and $\last_{p'} = e$.
	\end{enumerate}
\end{lemma}

The next lemma shows that the sets $\timey(i)$ yield a partition of the edge set in a natural way.
For this result we need to introduce a bit of notation.
We call a pointer $(u,v)$ such that $(u,v)$ is an edge of $G'$ a \emph{basic} pointer.
Moreover, we denote by $\done(i)$ the set of all basic pointers contained in $\timey(1) \cup \dots \cup \timey(i)$.
For simplicity, also define $\done(0) \coloneqq \emptyset$.
Finally, for any two nodes $u, v$ such that $v$ is an ancestor of $u$, we denote by $\betw(u,v)$ the set of all edges $(w,x)$ such that 1) $(w,x) = (y,v)$, or 2) $y$ is an ancestor of $w$, but $u$ is not an ancestor of $w$, where $y$ is the child of $v$ that lies on the path from $u$ to $v$.
In other words, $\betw(u,v)$ is the set of all edges that can be reached both from $u$ without crossing $v$, and from $v$ without crossing $u$.
For simplicity, for any pointer $p = (u,v)$, we also define $\betw(p) \coloneqq \betw(u,v)$.

\begin{lemma}\label{lem:between}
	Consider any $i \geq 1$, and any edge $e = (u,v) \in E(G')$.
	If $\done(i-1)$ does not contain the pointer $p = (u,v)$, then there is exactly one pointer $(w,x) \in \timey(i)$ such that $e \in \betw(w,x)$.
	If $\done(i-1)$ contains the pointer $p = (u,v)$, then there is no pointer $(w,x) \in \timey(i)$ such that $e \in \betw(w,x)$.
\end{lemma}

For any $i \geq 1$, and any pointer $p = (u,v) \in \timey(i)$, define $\succc(p)$ to be the set of all pointers $p' \in \timey(i+1)$ satisfying (1) $p' = (u, \pred_p)$, (2) $p' = (\pred_p, v)$, or (3) $p' = (w, \pred_p)$ where $w$ is a leaf and the edge $\last_{p'}$ does not lie on the path from $u$ to $v$.
If $\pred_p = \bot$, set $\succc(p) \coloneqq \emptyset$.
We obtain the following observation.

\begin{observation}\label{obs:summary}
	For any $i \geq 2$, and any pointer $p' \in \timey(i)$, there is exactly one pointer $p \in \timey(i-1)$ such that $p' \in \succc(p)$.
	For any $i \geq 1$, and any pointer $p = (u,v) \in \timey(1)$ with $\pred_p \neq \bot$, we have $\succc(p) = \{ (u, \pred_p), (\pred_p, v), p_1$, $\dots$, $p_{\deg(\pred_p)-2} \}$ where each $p_j$ is a pointer starting in a leaf, ending in $\pred_p$, and satisfying $\last_{p_j} = e_j$, where $e_1, \dots, e_{\deg(\pred_p)-2}$ are the $\deg(\pred_p)-2$ edges incoming to $\pred_p$ that do not lie on the path from $u$ to $v$. 
\end{observation}

\paragraph*{Algorithm $\atwo$.}

Now we describe algorithm $\atwo$ formally. 
The algorithm proceeds in iterations $i = 1, 2, \dots$ where in each iteration $i$, we process all pointers contained in $\timey(i)$.
When processing a pointer $p = (u,v)$, we assign some output label from $\sout$ to each so-far-unlabeled half-edge from $\{ (u, \first_p), (\last_p, v) \}$.
Due to Observation~\ref{obs:summary}, it suffices to explain
\begin{itemize}
	\item[(a)]\label{point1} how we choose those output labels for each pointer in $\timey(1)$, 
	\item[(b)]\label{point2} for each already processed pointer $p$ with $\pred_p \neq \bot$, how we choose those output labels for each pointer in $\succc(p)$.
\end{itemize}
For point (a), let $p_1 = (u_1, \rooot), \dots, p_k = (u_k, \rooot)$ denote the pointers in $\timey(1)$.
By \Cref{lem:attheroot}, we know that $k = \deg(\rooot)$ and, for each edge $e$ incident to $\rooot$, there is precisely one pointer $p_j$ with $\last_{p_j} = e$.
Recall \Cref{def:comptree,def:merge}.
We first assign labels to the half-edges incident to $\rooot$.
More precisely, for each edge $e \in \incid(\rooot)$, assign to half-edge $(\rooot, e)$ some label $\gout((\rooot, e)) \coloneqq \ell^e$ such that, for the obtained tuple $(\ell^e)_{e \in \incid(\rooot)}$, we have $(\ell^e)_{e \in \incid(\rooot)} \in S_{\rooot}$, and $\ell^e \in \comp^e(\rooot)$, for each $e \in \incid(\rooot)$.
For each pointer $p_j$, we assign to half-edge $(u_j, \first_{p_j})$ a label $\gout((u_j, \first_{p_j})) \coloneqq \ell^*$ s.t.\ $(\ell^*, \ell^{\last_{p_j}}) \in \pairs_{p_j}$ and $(\ell^*) \in S_{u_j}$.
For point (b), let $p = (u,v)$ denote an already processed pointer with $\pred_p \neq \bot$.
Note that, for the two pointers $p' \coloneqq (u, \pred_p)$ and $p'' \coloneqq (\pred_p, v)$, the half-edges $(u, \first_{p'})$ and $(\last_{p''}, v)$ have already been assigned output labels since $p$ has already been processed; denote those output labels by $\ell$ and $\ell'$, respectively.
However, by \Cref{lem:between} and Observation~\ref{obs:summary}, these are the only half-edges that are already labeled, out of all the half-edges that ``by definition'' have to be labeled after processing the pointers in $\succc(p)$.
Out of these unlabeled half-edges, we first assign an output to all half-edges incident to $\pred_p$.
Concretely, for each edge $e \in \incid(\pred_p)$, assign to half-edge $(\pred_p, e)$ some label $\gout((\pred_p, e)) \coloneqq \ell^e$ such that, for the obtained tuple $(\ell^e)_{e \in \incid(\pred_p)}$, we have
\begin{enumerate}
	\item $(\ell^e)_{e \in \incid(\pred_p)} \in S_{\pred_p}$ 
	\item $\ell^e \in \comp^e(\pred_p)$, for each $e \in \incid(\pred_p) \setminus \{\last_{p'}, \first_{p''}\}$ 
	\item $(\ell, \ell^{\last_{p'}}) \in \pairs_{p'}$ and $(\ell^{\first_{p''}}, \ell') \in \pairs_{p''}$ \ .
\end{enumerate}
Finally, for each pointer $p''' = (w, \pred_p)$ where $w$ is a leaf and $\last_{p'''}$ does not lie on the path from $u$ to $v$, we assign to half-edge $(w, \first_{p'''})$ a label $\gout((w, \first_{p'''})) \coloneqq \ell^*$ such that $(\ell^*, \ell^{\last_{p'''}}) \in \pairs_{p'''}$ and $(\ell^*) \in S_{w}$.
By Observation~\ref{obs:summary}, this finishes the processing of all the pointers in $\succc(p)$.
The algorithm in Phase II terminates in the first iteration $i$ in which $\timey(i) = \emptyset$.
This concludes the description of $\atwo$.
In \Cref{sec:appphasetwo}, we show that $\atwo$ is well-defined and analyze $\atwo$.
In particular, we will prove the following lemma that bounds the number of iterations in Phase II.

\begin{lemma}\label{lem:runtimephasetwo}
	Algorithm $\atwo$ terminates after $O(\log n)$ iterations.
\end{lemma}

\section{The Mid Regime}\label{sec:mid}

In this section, we will prove that all \lcl problems on trees with deterministic complexity $n^{o(1)}$ in the \local model can be solved deterministically in roughly $O(\log \log n)$ time in the low-space \mpc. In particular, we prove the following.

\begin{restatable}[Mid regime]{theorem}{thmMid}
	\label{thm:mid}
	Consider a forest consisting of (disjoint) connected components $C_1,\ldots,C_k$, each $C_i$ of size $n_i$.
	Furthermore, consider an \lcl problem $\Pi$ that can be solved in $O(\log z)$ rounds by a deterministic \local algorithm on instances with at most $z$ nodes.
	There is a deterministic low-space \mpc algorithm that solves $\Pi$ in $O(\log \log \max_i\{ n_i \})+O(\logstar N)$ time using $O(m+n)$ words of global memory where $N$ is the size of the ID space. The algorithm is component-stable.
\end{restatable}
As the proof of \Cref{cor:SpeedUpDet} requires the same techniques as the proof of \Cref{thm:mid}, we also present the proof of \Cref{cor:SpeedUpDet} at the end of this section.

In previous work, Chang and Pettie showed that in the \local model, there are no \lcl problems on trees whose complexity lies between $\omega(\log n)$ and $n^{o(1)}$~\cite{CP19}. In other words, they showed a complexity gap, giving rise to the \local complexity class $\Theta(\log n)$. They obtain their result by showing that any problem in this range admits a canonical way to solve it using a rake-and-compress decomposition (described in \Cref{sec:rc}) and a careful method that labels the tree, layer by layer (of the decomposition).

In order to prove \Cref{thm:mid}, in \Cref{sec:rc}, we show that both their rake-and-compress decomposition (see \Cref{sec:rc}) and their labeling method (see \Cref{sec:midalgo}) can be sped up to $O(\log \log n)$ in \mpc, while using strict memory parameters. 

\subsection{Rake-and-Compress Decomposition} \label{sec:rc}

In this section, we give a $O(\log \log n)$-time low-space \mpc algorithm for computing a rake-and-compress decomposition. In particular, we prove the following.

\begin{lemma}[Rake-and-Compress] \label{lemma:partitioning}
	Consider a constant-degree forest consisting of (disjoint) connected components $C_1,\ldots,C_k$, each $C_i$ of size $n_i$. There is a deterministic low-space \mpc algorithm that computes a rake-and-compress decomposition in $O(\log \log \max_i\{ n_i \})+O(\logstar N)$ time using $O(m+n)$ words of global memory where $N$ is the size of the ID space. The algorithm is component-stable.
\end{lemma}

Informally, the rake-and-compress decomposition of a graph is a disjoint set of nodes, such that the sets (or in other words \emph{layers}) are enumerated, and every node $v$ has at most two neighbors in the same or higher layers. The precise properties of the decomposition are given in \Cref{obs:decompProperties}. 

In \Cref{subsec:decomplocal}, we summarize the \local rake-and-compress algorithm of \cite{CP19}. Then, in \Cref{subsec:decompmpc}, we describe our low-space \mpc algorithm and prove \Cref{lemma:partitioning}.

\subsubsection{Decomposition in \local} \label{subsec:decomplocal}

The algorithm consists of two steps: a decomposition step, where nodes are partitioned into layers and a postprocessing step, where we compute an $(\alpha, \beta)$-independent set (\Cref{def:alpha-beta}), in $O(\log^*N)$ time~\cite{Linial92} and adjust the layers slightly. Recall that $N$ denotes the size of the ID space.

\begin{definition}[$(\alpha, \beta)$-independent set] \label{def:alpha-beta}
	Let $P$ be a path. A set $I \subset V(P)$ is called an $(\alpha, \beta)$-independent set if the following conditions are met: (i) $I$ is an independent set, and $I$ does not contain either endpoint of $P$, and (ii) each connected component induced by $V(P)-I$ has at least $\alpha$ vertices and at most $\beta$ vertices, unless $|V(P)| < \alpha$, in which case $I=\emptyset$.
\end{definition}

\begin{enumerate}
	\item Suppose $l$ is some constant depending on the \lcl problem. The algorithm begins with $U = V(G)$ and $i=1$, repeats Steps (a)--(c) until $U=\emptyset$, then proceeds to Step 2.
	
	\begin{enumerate}
		\item For each $v \in U$:
		
		\begin{enumerate}
			\item \textsf{Compress}. If $v$ belongs to a path $P$ such that $|V(P)| \geq l$ and $\deg_{U}(u) = 2$ for each $u \in V(P)$, then tag $v$ with $i_C$.
			\item \textsf{Rake}. If $\deg_U(v)=0$, then tag $v$ with $i_R$. If $\deg_U(v)=1$ and the unique neighbor $u$ of $v$ in $U$ satisfies either (i) $\deg_U > 1$ or (ii) $\deg_U = 1$ and $\text{ID}(v) > \text{ID}(u)$, then tag $v$ with $i_R$.
		\end{enumerate}
		\item Remove from $U$ all vertices tagged $i_C$ or $i_R$ and set $i \larr i+1$.
	\end{enumerate}
	
	\item Initialize $V_i$ as the set of all vertices tagged $i_C$ or $i_R$. 
	The graph induced by $V_i$ consists of unbounded length paths, but we prefer constant length paths. For each edge $\{u,v\}$ such that $v$ is tagged $i_R$ and $u$ is tagged $i_C$, promote $v$ from $V_i$ to $V_{i+1}$. For each path $P$ that is a connected component induced by vertices tagged $i_C$, compute an $(l,2l)$-independent set $I_P$ of $P$, and then promote every vertex in $I_P$ from $V_i$ to $V_{i+1}$.
\end{enumerate}

\begin{observation}\label{obs:decompProperties}
	The following properties of the rake-and-compress decomposition are either evident or proven by Chang and Pettie~\cite[Section 3.9]{CP19}.
	
	\begin{itemize}
		\item Define $G_i$ as the graph induced by nodes in layer $i$ or higher: $\bigcup_{j=i}^L V_j$. For each $v \in V_i$, $\deg_{G_i}(v) \leq 2$.
		\item Define $\mathcal{P}_i$ as the set of connected components (paths) induced by the nodes in $V_i$ with more than one node. For each $P \in \mathcal{P}_i$, $l \leq |V(P)| \leq 2l$ and $\deg_{G_i}(v)=2$ for each node $v \in V(P)$.
		\item The graph $G_L$ contains only isolated nodes, i.e., $\mathcal{P}_L=\emptyset$.
		\item At least a constant $\Omega(1/l)$ fraction of vertices in $U$ are eliminated in each iteration, resulting in a runtime of $O(\log n)$ and decomposition size $L=O(\log n)$.
	\end{itemize}
	
	As a consequence, each vertex $v \in V_i$ falls into exactly one of two cases: (i) $v$ has $\deg_{G_i}(v) \leq 1$ and has no neighbor in $V_i$, or (ii) $v$ has $\deg_{G_i}(v) = 2$ and is in some path $P \in \mathcal{P}_i$.
	
\end{observation}

\subsubsection{Decomposition in \mpc} \label{subsec:decompmpc}

Let us first define a helper function.

\begin{itemize}
	\item \textsf{Peel($r$)}: compute the lowest $r$ layers of the decomposition by simulating Step 1 of the \local algorithm $r$ times.
\end{itemize}

Recall that by Observation~\ref{obs:decompProperties}, at least a constant $\Omega(1/l)$ fraction of nodes are eliminated in each simulation. When taking a closer look into~\cite{CP19}, the exact fraction is $1/2(l+1) \geq 1/4l$. Hence, we can state that at most a constant $1-1/4l$ fraction of nodes is left in the graph after each step of the \local algorithm. Set constant $c \larr \text{argmin}_c \{c \mid (1-1/4l)^c < 1/\Delta \}$, and observe that since $\Delta$ and $l$ are constants, $c$ is also constant. Our \mpc algorithm is the following.

\begin{enumerate}
	\item For $i = 0,\dots,\log \log n^\delta$ phases: perform $c$ steps of \textsf{Peel($2^i$)}, and then perform one graph exponentiation step.
	\item Perform \textsf{Peel($\delta \log n$)} until the graph is empty.
	\item Simulate Step 2 of the \local algorithm. 
\end{enumerate}

\begin{proof}[Proof of \Cref{lemma:partitioning}]

	Correctness follows from~\cite[Section 3.9]{CP19}, as we only simulate their algorithm.
	Let us bound the time complexity. In Step 1, during any phase $i$, each node sees its $2^i$-radius neighborhood due to graph exponentiation. This vision enables each node to perform $c$ steps of \textsf{Peel($2^i$)}, which altogether takes constant time. After $\log \log n^\delta$ phases, all nodes see their $\delta \log n$-radius neighborhoods and Step 1 terminates. In Step 2, nodes perform \textsf{Peel($\delta \log n$)} until the graph is empty, which takes $O(1/\delta)=O(1)$ time, since there are $O(\log n)$ layers in the decomposition in total by \Cref{obs:decompProperties}. Since the vision of each node is $\delta \log n$ and Step 2 of the \local algorithm takes $O(\log^*b)$ time, we can simulate it in $O(1)$ time. We conclude that the algorithm runs in $O(\log \log n) + O(\log^*b)$ time.
	
	In Step 1, during any phase $i$, performing $c$ steps of \textsf{Peel($2^i$)} results in simulating Step 1 of the \local algorithm $2^{i}c$ times. Hence, after applying \textsf{Peel($2^i$)}, in any phase $i$, there are at most $n \cdot (1-1/4l)^{2^i c} < n / \Delta^{2^i}$ nodes left in the graph. Since the graph exponentiation step of phase $i$ requires at most $\Delta^{2^{i}}$ memory per node, we conclude that each phase (and hence the whole algorithm), requires at most $O(m+n)$ words of global memory. After $\log \log n^\delta$ phases, all nodes see their $\delta \log n$-radius neighborhoods. Since $\Delta$ is constant, the $\delta \log n$-radius neighborhood of any node contains at most $O(n^\delta)$ nodes and the local memory is always respected.
	
	Observe that all of our arguments are local, i.e., nodes in separate components do not communicate. Hence, the algorithm is component-stable when the input graph is a forest, in which case the runtime becomes $O(\log \log \max_i\{ n_i \})+O(\logstar N)$.
\end{proof}

\subsection{The Labeling Method} \label{sec:midalgo}

Suppose that we are given the rake-and-compress decomposition described in \Cref{sec:rc}. Let us adopt the same node-labeled \lcl problem definition as \cite{CP19}. Note that this definition includes port-numberings and is hence equivalent to our previous definition of half-edge labeled \lcls (\Cref{def:lcl}).

\begin{definition}[\lcl, Chang and Pettie~\cite{CP19}] \label{def:newlcl}
	Fix a class $\fG$ of possible input graphs and let $\Delta$ be the maximum degree in any such graph. An \lcl problem $\Pi$ for $\fG$ has a radius $r$, constant size input and output alphabets $\sinn$, $\sout$, and a set $\fC$ of acceptable configurations. Note that $\sinn$ and $\sout$ can include $\bot$. Each $C \in \fC$ is a graph centered at a specific vertex, in which each vertex has a degree, a port numbering, and two labels from $\sinn$ and $\sout$. Given the input graph $G(V,E,\phinn)$ where $\phinn: V(G)  \xrightarrow{} \sinn$, a feasible labeling output is any function $\phout :V(G)  \xrightarrow{} \sout$ such that for each $v \in V(G)$, the subgraph induced by $N^{r}(v)$ (denoting the $r$-neighborhood of $v$ together with information stored there: vertex degrees, port numberings, input labels, and output labels) is isomorphic to a member of $\fC$. A complete labeling output is such that for each $v$, $\phout(v) \neq \bot$. An \lcl can be described explicitly by enumerating a finite number of acceptable configurations.
\end{definition}

Let us revisit some other definitions and results of \cite{CP19} before introducing our algorithm.

\begin{definition}[Class, Chang and Pettie~\cite{CP19}] \label{def:class}
	Consider a rooted tree $T$ with a root $v$ and an \lcl problem $\Pi$ as in \Cref{def:newlcl}. The (equivalence) class of $T$, denoted $\class(T)$ is the set of all possible node labelings of the $r$-hop neighborhood of $v$ such that the labeling can be extended to a complete feasible labeling of $T$ (with respect to the \lcl problem $\Pi$). Note that for constant-degree trees, the number of equivalence classes is constant.
\end{definition}

\begin{lemma}\label{lemma: replace}
	Consider a graph $G$, an \lcl problem $\Pi$ as in \Cref{def:newlcl}, and the rake-and-compress decomposition described in \Cref{sec:rc}. Recall that the decomposition is parameterized by a constant $\ell$ that depends on the input $\lcl$ and that $G_i$ denotes the graph induced by the nodes in layer $i$ or higher.
	
	Let $P = (v_1, v_2, \ldots, v_x)$, for some $x \in [\ell, 2\ell]$ be a path induced by nodes in $G_i$ with degree $2$ and let $s = v_1$ and $t = v_x$. Moreover, let graph $T_1 \cup T_2 \cup \ldots \cup T_x$, denoted by $H = (T_1, T_2, \ldots, T_x)$, correspond to the sequence of disjoint trees rooted from nodes $(v_1,v_2, \dots, v_x)$. Then, there exists a tree $H^+$ with two dedicated nodes $s^+$ and $t^+$ such that the following holds.
	
	Let $G_s$ be the connected component of $G - H$ that is adjacent to node $s$ in $G$. Notice that $s \not\in G_s$. Then, if graph $G_s$ is not empty, graph $G_{s^+}$ is created by connecting a copy of $G_s$ to a copy of $H^+$ via a single edge $\{u,s^+\}$ such that $u \in G_{s}$ and $s^+ \in H^+$. The graph $G_t$ is constructed identically but using a disjoint copy of $H^+$.
	
	Then, if graphs $G_{s^+}$ and $G_{t^+}$ admit feasible node labelings $\phout^{s^+}$ and $\phout^{t^+}$, then the input graph $G$ admits a feasible node labeling $\phout^*$ such that for any $v \in G - H$, we have $\phi^*(v) = \phout^{s^+} (v)$ if $v \in G_s$ and $\phout^*(v) = \phout^{t^+} (v)$ if $v \in G_t$. Furthermore, the graph $H^+$ can be computed with the knowledge of $\class(T_i)$ for each $i$.
\end{lemma}

\begin{proof}
	The existence of $H^+$ is rigorously proven by Chang and Pettie~\cite[Lemma 9]{CP19} through \textit{pumping}, \textit{tree surgery} and \textit{pre-commitment} techniques. The same previous work describes the computation of graph $H^+$~\cite[Proof of Lemma 13]{CP19}.
\end{proof}

Note that in \Cref{def:class,lemma: replace} we talk about rooted trees and rooted subtrees hanging from nodes in a path. One could rightfully assume that we either root the tree beforehand or assume a rooted tree as input. However, we do none of the previous. Instead, when talking about a tree $T_v$ rooted at node $v$ in layer $i$, we simply refer to the subgraph induced by nodes in layers $< i$ that are connected to $v$.

\paragraph*{The Algorithm}

First, we divide the nodes into \emph{batches} according to which layer they belong to in the decomposition. Let $c$ be such that for all $i$, $|\bigcup_{j \geq i + c} V_j| \leq |\bigcup_{j \geq i} V_j| \cdot \Delta^{-1}$, i.e., if we remove $c$ layers from the decomposition, the number of nodes drops by a factor of at least $\Delta$. By Observation~\ref{obs:decompProperties}, we know that $c$ is a constant. Let us define nodes in layers $V_1,\dots,V_c$ as batch $B_0$. For $i>0$ and as long as $\Delta^{2^i} \leq n^{\delta}$, let us define nodes in layers $V_{(2^i-1) \cdot c + 1}, \ldots, V_{(2^{i+1}-1) \cdot c}$ as batch $i$; see \Cref{fig:mid}. Note that there are $O(\log \log n)$ batches as  defined previously, and assuming that there is enough layers, batch $i$ always consists of $2^i c$ layers. All nodes that do not belong to any batch, as defined previously, are defined as batch $B_L$. The algorithm starts with running $O(\log \log n)$ phases, each of which is executed in a constant number of \mpc rounds and consists of the following steps.

\begin{figure}
	\centering
	\includegraphics[width=0.7\textwidth]{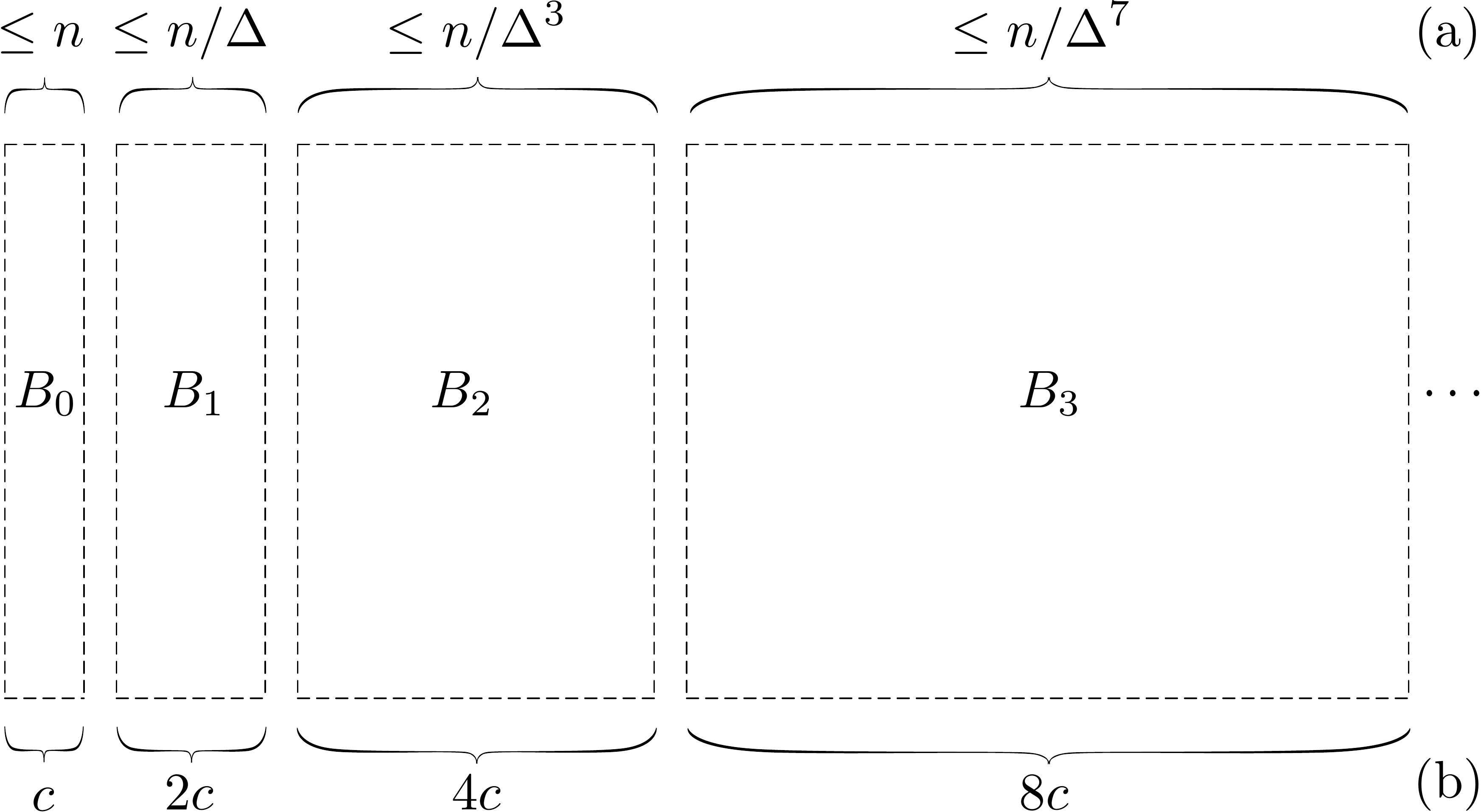}
	\caption{The first four batches, where each batch $i$ contains $2^i \cdot c$ layers and at most $n/\Delta^{2^i-1}$ nodes: (a) number of nodes in a batch, (b) number of layers in a batch.}
	\label{fig:mid}
\end{figure}

\begin{enumerate}
	\item In the start of phase $i$, our communication graph is $G^{2^i}$. We process the $i$:th batch by simulating $2^i c$ iterations of the following local process. If a node $v$ has neighbors only in higher layers, it locally computes $\class(v)$ and informs its (unique) neighbor about the class. If a node $v$ has neighbors in the same layer, then it must be part of a constant length path. Each node that is not an endpoint of such a path, locally computes its class and informs the endpoints of the class. Once an endpoint $s$ learns of all classes on the path, it locally replaces the path with graph $H^+$ as described in \Cref{lemma: replace}. Then $s$ locally computes $\class(s)$ and informs its parent (if any) of the class.
	\item As the second step, for all but the last phase, each node in batch $B_{j > i}$ and in batch $B_L$ performs one step of graph exponentiation. Note that the nodes that have computed their classes, i.e., nodes in  batch $B_{j \leq i}$ do not participate in graph exponentiation. Thus, we obtain $G^{2^{i+1}}$ as our communication graph for batches $>i$ and batch $B_L$. We make the exception for the last phase as we do not want to violate local memory of the nodes in batch $B_L$.
\end{enumerate}

Since we process one batch in each phase, all numbered batches are processed in $O(\log \log n)$ phases. If batch $B_L$ is non-empty, its communication graph is $G^{\Theta(\log n)}$ and the algorithm proceeds by simulating the local process described above until all nodes in $B_L$ has derived their class.

Once all nodes have computed their class, we process the batches in the reversed order. The (local) roots begin by choosing a label that can be extended to a valid labeling on the whole graph. Then, once a node learns the feasible label of its parent (nodes have at most one neighbor in a higher layer), it can choose an extendable label. Similarly, once the parents of both endpoints of a layer-induced path of $H^+$ have decided on their labels, the endpoints choose a valid labeling of the original layer-induced path of $H$.

\begin{proof}[Proof of \Cref{thm:mid}]
	
	We prove the the algorithm of \Cref{sec:midalgo}, since the rake-and-compress algorithm is already proven in \Cref{sec:rc}.  
	
	By the definition of a class and by \Cref{lemma: replace}, the local process in Step 1 is always possible and all nodes can compute their class. The existence of the valid labels for each node is again provided by the definition of a class and by \Cref{lemma: replace}.
	
	Let us bound the time complexity of the algorithm. First, let us analyze the complexity of computing the class of each node. Each phase $i$ indeed takes constant time, since simulating $2^i \cdot c$ iterations of the local process takes $O(1)$ time due to the communication graph in batches $\geq i$ being $G^{2^i}$. As mentioned previously, after $O(\log \log n)$ phases, all batches except $B_L$ have been processed and the communication graph of batch $B_L$ is $G^{\Theta(\log n)}$. Note that by Observation~\ref{obs:decompProperties}, batch $B_L$ contains $O(\log n)$ layers. Hence, simulating the local process of our algorithm on our communication graph $G^{\Theta(\log n)}$ takes constant time, after which all nodes have computed their class. Next, let us analyze the complexity of computing the label of each node. Using the communication graph created throughout the algorithm, we process the batches in the reverse order, which requires the same number of \mpc rounds as computing the class. Since obtaining the rake-and-compress decomposition required for this process takes $O(\log \log n) + O(\logstar N)$ time, the overall runtime of the algorithm is $O(\log \log n) + O(\logstar N)$.
	
	Let us analyze the memory requirement. The local memory bound is respected, since a batch $i$ is defined such that $\Delta^{2^i} \leq n^\delta$ and a node in batch $i$ executes at most $i$ steps of graph exponentiation, resulting in a neighborhood containing at most $\Delta^{2^i}$ nodes. Note that batch $B_L$ performs the same number of graph exponentiation steps as does the last numbered batch and hence, local memory bounds are not violated. Now for the global memory bound. Consider batch $i$. By the definition of $c$, we have that batch $i$ contains at most $|\bigcup_{j \geq (2^i-1) \cdot c} V_j| \leq n \cdot \Delta^{-(2^i-1)}$ nodes. Since batch $i$ executes at most $i$ steps of graph exponentiation, the memory required to store the communication graph per numbered batch is at most $|\bigcup_{j \geq (2^i-1) \cdot c} V_j| \cdot \Delta^{2^i} \leq n \cdot \Delta$. Now recall that nodes in batch $B_L$ performs the same number of graph exponentiation steps as does the last numbered batch. Observe that the size of $B_L$ has the same upper bound as the last numbered batch. Hence, storing the communication graph of $B_L$ also requires at most $n \cdot \Delta$ memory. The number of batches is trivially upper bounded by the number of layers and hence, the global memory use is bounded by $O(n \log n)$.
	
	By separating the nodes in the first $c \log \log n$ layers and computing their classes prior to the execution and their labels after the other nodes have been handled, we can drop the requirement to sharp $O(n)$. This is evident as separating the first $c \log \log n$ layers leaves us with at most $n/\log n$ nodes. This scheme contributes only an additive $O(\log \log n)$ term to the runtime. One source of memory issues during execution could be sending messages of $\omega(1)$ size. Since our messages only contain class information (\Cref{def:class}), and since there are only a constant number of classes, our messages are of constant size.
	
	Observe that all of our arguments are local, i.e., nodes in separate components do not communicate. Hence, the algorithm works equally on a forest, in which case the runtime becomes $O(\log \log \max_i\{ n_i \})+O(\logstar N)$. 
\end{proof}

\section{The Low Regime} \label{sec:low}

In this section, we prove  the following theorem.

\begin{restatable}[Low regime]{theorem}{thmsublogfast}
	\label{thm:sublogfast}
	Any \lcl problem on trees with randomized \local complexity $o(\log n)$ can be solved with a randomized algorithm in $O(\log \log \log n)+O(\logstar b)$ rounds in the low-space \mpc model with $O(m+n)$ words of global memory where IDs and words have $b$ bits. The algorithm is component-stable. 
\end{restatable}

In order to prove the above theorem, we restate the main theorem of our mid regime, and introduce an important result from a previous work, for which we also give a proof sketch.

\thmMid*

\begin{lemma}[\cite{BCMOS21,CP19}]
	\label{lem:sublogarithmicReduction}
	Let $\Pi$ be an \lcl (possibly on general graphs) with a sublogarithmic randomized \local algorithm $\mathcal{A}$. Then there exists a constant $t_0$, and an \lcl $\Pi''$ (whose definition only depends on $\Pi$ and $\mathcal{A}$) and an $O(1)$-time \congest reduction such that:
	\begin{enumerate}
		\item  given a $t_0$-distance coloring with $f(\Delta)$  colors (for some function $f$) reduces solving $\Pi$ to solving 
		$\Pi''$ on several independent connected components $C_1,\ldots,C_k$ each of size at most $|C_i|\leq N=O(\log n)$ while IDs use $b=O(\log n)$ bits and $\sum|C_i|\leq n$ holds. 
		\item If the input graph is a tree (general graph), then the \lcl problem $\Pi''$ can be solved in $O(\log z)$ rounds ($\poly\log z$ rounds) by a deterministic \local algorithm on instances with at most $z$ nodes. 
	\end{enumerate} 
\end{lemma}
\begin{proof}[Proof sketch]
	Any \lcl $\Pi$ with randomized sublogarithmic \local complexity can be solved via the following normal form: 
	\begin{enumerate}
		\item Determine (without communication) a constant time $t_0$ algorithm $\mathcal{A}_0$ for $\Pi$ that errs with a small constant probability and uses a constant number of random bits per node. 
		\item Determine good random bits for each node to execute $\mathcal{A}_0$ such that it errs at no node. 
		\item Execute the constant time algorithm $\mathcal{A}_0$ with the computed random bits.
	\end{enumerate}
	Step~2 is most involved. In fact,  one can show that the problem $\Pi'$ of determining good random bits for the nodes is also an \lcl problem and additionally it is a so-called \llle problem with a polynomial \llle criterion and a constant dependency degree. For more details on \llle{}s see \cite{BCMOS21}. The crucial point is that such \llle problems can be solved via the shattering method: Given a suitable constant distance coloring one can, in a constant number of rounds, set the random bits of some nodes such that remaining nodes, w.h.p, form small components $C_1,\ldots,C_k$, each of size $\leq N=O(\log n)$. One core technical difficulty in \cite{BCMOS21}  is to show that the remaining problem on each component is a proper \lcl problem $\Pi''$ that on trees has an $O(\log z)$ deterministic \local algorithm on instances of size $z$. On general graphs the problem $\Pi''$ can be solved in $\poly\log z$ rounds via the \llle algorithm of \cite{RozhonG20,FischerG17} on instances of size $z$ as the problem is also an \llle problem. 
	
	All parts of this reduction, except for solving $\Pi''$ on the small components, can clearly be executed in a constant number of rounds. 
\end{proof}

We are now ready to prove \Cref{thm:sublogfast}.

\begin{proof}[Proof of \Cref{thm:sublogfast}]
	We apply \Cref{lem:sublogarithmicReduction}.
	All steps of the constant time reduction from \Cref{lem:sublogarithmicReduction} can clearly be executed in $O(1)$ rounds in the \mpc model. A distance-$t_0$ coloring can be computed in $O(\log\logstar n)$ rounds using  \Cref{lem:distanceColoringMPC}. Hence, it remains to solve the \lcl problem $\Pi''$ on several independent instances each of size at most $N$ in parallel. Further, the lemma provides us with the fact that $\Pi''$ can be solved with a deterministic $O(\log z)$ round algorithm on instances of size $z$. Using \Cref{thm:mid}, we obtain that $\Pi''$ can be solved in $O(\log \log z)$ rounds on instances of size $z$. Using this on each component in parallel (setting $z=N$) we obtain an $O(\log\log N)=O(\log\log\log n)$ rounds algorithm. Note that when applying both lemmas IDs use the $b=O(\log n)$ bits.
	
	The global space constraints are met as $\sum|C_i|\leq n$ and the algorithm from \Cref{thm:mid} only requires linear global space. All used subroutines are component-stable. 
\end{proof}

We say an algorithm solves the \emph{connected components problem} if each node in each connected component $C$ outputs $\min_{v\in C} {ID_v}$. 
Note, that by definition, the output of an algorithm solving the connected components algorithm cannot depend on other components, that is, any connected components algorithm is component-stable by definition. \footnote{Note that its crucial that we define the problem slightly different than typically, see e.g., \cite{coy2021deterministic}, where each nodes in a connected component can output any arbitrary number as long as they output the same number and the number is not used by any other component. }
\begin{observation}
	\label{obs:conCompIscomponentStable}
	Any connected components algorithm is component-stable by definition. 
\end{observation}

\subparagraph{The connected component algorithm by \cite{googleconnectivity,coy2021deterministic}:} The crucial ingredient of the deterministic connected components algorithm by \cite{coy2021deterministic} are two deterministic subroutines, one to compute large matchings in paths/cycles, and one to solve a certain set-cover instance. Solving these problems deterministically is sufficient to derandomize the algorithm by \cite{googleconnectivity}. While the problems at hand are irrelevant for this paragraph, the way they are solved is interesting.  Via the method of conditional expectation the random bits of a shared random seed are deterministically chosen in a suitable way to compute the output from it. Here, all parts of the graph, in particular all different components,  use the same seed. Changes in one component of the graph can incur changes in chosen seed, and hence can influence the output of other components of the graph. Thus, the used technique is inherently non component-stable. However, due to \Cref{obs:conCompIscomponentStable} we obtain the following theorem. 

\begin{theorem}[\cite{googleconnectivity,coy2021deterministic}]
	\label{thm:arturConnected}
	There is a deterministic component-stable algorithm to solve the connected components problem with components $C_1,\ldots,C_k$ in $O(\log \max_{i}\textsf{diam}(G[C_i]))+\log\log \max_i |C_i|)$ rounds. 
\end{theorem}
\begin{proof}[Proof Sketch]
	It is immediate that the algorithms of  \cite{googleconnectivity,coy2021deterministic} can solve the aforementioned version of connected components. In the pen-ultimate state the algorithm has several virtual nodes (connected via  a clique) $v_1,\ldots,v_k$ for each connected component $C$, and each of these nodes has an associated set containing some of the IDs of the nodes in the component. The ID of each node appears in exactly one set. Hence, one can easily determine the minimum ID of the component.  
	
	The runtime of the algorithm is stated as $O(\log \max_{i}\textsf{diam}(G[C_i]))+\log\log n)$. However, the $O(\log\log n)$ term stems from the fact that the \emph{level} of a virtual node cannot grow beyond $O(\log \log n)$. In fact, the level is associated with the number of nodes that the virtual node has come in contact with and it cannot grow beyond the $O(\log\log |C_i|)$ where the node stems from component $C_i$. 
\end{proof}

\begin{theorem}
	\label{thm:generalGraphMoreSpace}
	Any \lcl problem on general graphs with randomized \local complexity $o(\log n)$ can be solved with a randomized component-stable algorithm in $O(\log \log n)$ rounds in the low-space \mpc model using $O(n\log n)$ words of global memory.
\end{theorem}
\begin{proof}
	Apply \Cref{lem:sublogarithmicReduction}.
	First, compute a $t_0$-distance coloring via \Cref{lem:distanceColoringMPC} in $O(\log \logstar n)$ rounds.  The constant number of rounds of the reduction can clearly be executed in $O(1)$ rounds in the \mpc model. To solve the problem $\Pi''$ on all components $C_1,\ldots,C_k$, each of size $\leq N$ in parallel, we first run the connected components algorithm from \Cref{thm:arturConnected}. It runs in $O(\log |C_i|)=O(\log N)=O(\log\log n)$ rounds. Afterwards, every node knows the minimum ID in its component. We use a deterministic load balancing algorithm to send all $\leq N$ nodes of $C_i$, including their incident edges, to the same machine, where we can solve $\Pi''$. No additional global space is required. 
	
	We next prove the result for $O(n\log n)$ words of global memory. The problem $\Pi''$ can be solve in $T=\poly\log z$ rounds in the \local model on instances of size $z$ via \cite{RozhonG20} as $\Pi''$ is not just an \lcl but also an \llle problem with a polynomial criterion. Now, we perform graph exponentiation for $\log T$ rounds after which every node knows its $T$-hop neighborhood, which is enough to determine its output. In the worst case, during the exponentiation, each node of $C_i$ learns all of $C_i$, that is the total number of words that we need is $\sum_{i=1}^k |C_i|^2\leq O(n\cdot N)=O(n\log n)$.
\end{proof}

\corSpeedUpDet*
\begin{proof}
	In \cite{componentstable}, Czumaj, Davies, and Parter provided a deterministic component-unstable counterpart of the constant time reduction in \Cref{lem:sublogarithmicReduction}. Essentially, the reduction in \cite{BCMOS21} uses the $\poly \Delta=O(1)$-round shattering framework of \cite{FischerG17} for so called \lovasz Local Lemma instances. The aforementioned authors replace this shattering phase with a deterministic shattering procedure that uses optimal global memory and $\poly \Delta=O(1)$ rounds.   
\end{proof}

\section{An Automatic Procedure}
\label{sec:automatic}
In our results we claim that, if we just know the complexity of a problem in the distributed setting, then we can directly obtain an exponentially faster \mpc algorithm. In some of our proofs we will assume something stronger: that we are given a problem, its distributed complexity, and \local \emph{algorithm} with such a complexity. We now show that assuming that an algorithm for a problem is given is not stronger than assuming that just its asymptotic complexity is provided.
In order to do so, we now consider all the possible complexities that a problem can have in the distributed setting (as discussed in \Cref{ss:landscape}) and show how to obtain a distributed algorithm for free. In this way, given the complexity of a problem in the distributed setting, one can first apply the following procedure, and then apply our speedup results to obtain an exponentially faster \mpc algorithm.

\begin{lemma}
	Consider an \lcl problem on trees, for which we are given its deterministic time complexity $f(n)$ (resp.\ its randomized time complexity $g(n)$) in the \local model. It is possible to automatically find a \local algorithm with deterministic time complexity $f(n)$ (resp.  randomized time complexity $g(n)$).
\end{lemma}
\begin{proof}
	In \cite{NaorS95}, it is shown that for any \lcl $\Pi$ and for any given $k$,  it is possible to decide  whether $\Pi$ can be solved in $k$ rounds. Moreover, if the answer is affirmative, one also obtains an algorithm. Hence, if we already know that $f(n) = O(1)$, we can use this method to get an algorithm for free.
	
	In \cite{CKP19}, it is shown that any $f(n) = O(\log^* n)$ solvable problem can be solved in a very specific way: first compute a distance-$d$ $O(\Delta^{2d})$-coloring, for some specific value of $d$, and then apply a constant time algorithm. Also, observe that if we know what is the right value for $d$, then we can use the method for the $f(n) = O(1)= k$ case to find the algorithm. But we may not know $d$, and we cannot just start testing for $k = 1,2,\ldots$, because when testing for $d=1$ for example, there may not exist any constant time algorithm that solves the problem if given an $O(\Delta)$-coloring, so the procedure may not find ant valid $k$, and just diverge. But we can test differently: we proceed in iterations, and in each iteration $i$ we check all possible values of $k$ and $d$ satisfying $k,d \le i$.
	
	If the problem satisfies $f(n) = \Omega(\log n)$, then, by \cite{CP19,Chang2020}, we know that we can automatically decide what is the right asymptotic value of $f(n)$, and throughout the process, also obtain a \local algorithm for free.
	
	If $g(n) \neq \Theta(\log \log n)$, then we know by prior work that $f(n) = g(n)$, and hence in that case we already showed how to obtain an algorithm. If $g(n) =  \Theta(\log \log n)$, then in particular we know that $g(n) = o(\log n)$. For all problems falling into this category, we know that they can be \emph{sped up} to $O(\log \log n)$ as follows (see \cite{CP19}):
	\begin{itemize}[noitemsep]
		\item Convert the $o(\log n)$-rounds randomized algorithm into a $k = O(1)$-rounds randomized algorithm that has small enough local failure probability $p = \phi(k,\Pi)$ for some function $\phi$ defined in \cite{CP19} (where the local failure probability is the probability that a given specific node fails to produce a correct solution).
		\item Use a distributed \llle algorithm to find good random bits, such that if we run the obtained algorithm with them, it does not fail. This part requires $O(\log \log n)$ on trees.
	\end{itemize}
	We observe that the same techniques used to prove that we can find a constant time algorithm by brute force (see \cite[Theorem 4.3]{NaorS95}) also extend to the randomized case. In particular, we can use the procedure of \cite{NaorS95} to decide whether there exists a $k$-round randomized algorithm that uses at most $b$ random bits on each node and that locally fails with probability at most $p$, for any constant $k$, $b$, and $p$.
	Hence, we can test $b$ and $k$ (and $p = \phi(k,\Pi)$) in phases to find such an algorithm, since by assumption it exists. This algorithm, combined with \llle, gives an $O(\log \log n)$ algorithm for free.
\end{proof}

\begin{remark}
	For the ease of presentation, most of our paper is written from the viewpoint of a single tree. However, we want to point out that all our algorithms work on forests, too. For the tiny, high and the mid regime this is reasoned in detail in the respective proofs. For constant time algorithms this is immediate. 
	In \Cref{sec:tiny}, we reasoned that the speedup in the tiny regime applies also to forests. The main reason is that our asymptotic runtime solely depends on the size of the ID space. Hence, these algorithms are component-stable as long as another component cannot change the ID space. The same is true for the low regime in \Cref{sec:low} for the following reasons. The algorithm is based on a constant time shattering procedure that is component-stable and a post-shattering phase which relies on the component-stable algorithm of the mid regime. However, as the shattering phase requires a sufficiently large constant distance coloring, the same ID space dependency as in the tiny regime applies. 
\end{remark}

\clearpage
\appendix

\section{Rooting}\label{sec:rooting}

In this section, we describe a novel low-space \mpc algorithm that roots a forest deterministically in $O(\log n)$ time. Rooting entails orients the edges of the graph such that in each connected component, they point towards a unique root.

\begin{lemma}[Rooting] \label{lemma:rooting}
	Consider an arbitrary-degree forest consisting of (disjoint) connected components $C_1,\ldots,C_k$, each $C_i$ of size $n_i$. There is a deterministic, component-stable, low-space \mpc algorithm that roots the forest in $O(\log \max_i\{ n_i \})$ time using $O(m+n)$ words of global memory.
\end{lemma}

The rooting algorithm is an essential subroutine in the high regime, but it may also be of independent interest. We start by introducing the technique of \textit{path exponentiation}, which is used to contract long paths in logarithmic time in a memory efficient way. By leveraging the fact that in trees, at least half of the nodes are of degree $<3$, one could apply path exponentiation in a straightforward manner to root a tree in $O(\log^2 n)$ time. Our main contribution is pipelining this process, reducing the runtime to $O(\log n)$, and solving numerous small challenges that arise along the way.

\subparagraph{Path Exponentiation.}

Let us introduce path exponentiation, which is a logarithmic time technique to compress a path such that upon termination, the endpoints share a virtual edge (defined next). The technique is memory efficient in the sense that in addition to the input edges (i.e., edges incident to a node in the input graph), all nodes in a path keep at most two virtual edges in memory. Consider a path $P$ with \textit{endpoints} $s,t$ and \textit{internal nodes} in $P \setminus \{s,t\}$. Leaf node are considered to be endpoints, and degree-2 nodes are consider to be internal nodes. Nodes in $P$ always keep their input edges in memory. Path exponentiation is initialized by duplicating all edges in $P$ and calling this new path the \textit{virtual graph}. Nodes connected by a virtual edge are called virtual neighbors. A new virtual edge $\{v,w\}$ can be created by node $u$ if there previously existed virtual edges $\{u,v\}$ and $\{u,w\}$. In practice, creating a virtual edge $\{v,w\}$ entails node $u$ informing $v$ the ID of $w$ and $w$ the ID of $v$, i.e., node $u$ \textit{connects} nodes $v$ and $w$. Path exponentiation is executed only on this virtual graph. So henceforth, when talking about neighbors and edges, we refer to virtual neighbors and edges, unless specified otherwise.

In each (path) exponentiation step, endpoints and internal nodes are handled separately. During exponentiation, an internal node $u$ has exactly two neighbors and it can be one of three types: (1) neither neighbor is an endpoint (2) one neighbor is an endpoint and one is an internal node (3) both neighbors are endpoints. In each exponentiation step, an internal node $u$ does the following. 

\begin{itemize}
	\item Node $u$ communicates with its neighbors to learn if it is of type 1, 2 or 3
	\item For each node type:
	\begin{enumerate}
		\item Connects its neighbors $v$ and $w$ with an edge and removes edges $\{u,v\}$ and $\{u,w\}$.
		\item Connects its internal neighbor $v$ to its endpoint neighbor $s$ (or $t$) with an edge. It removes the edge $\{u,v\}$, but keeps the edge $\{u,s\}$ (or $\{u,t\}$) in memory.
		\item Connects its endpoint neighbors $s$ and $t$ with an edge and keeps edges $\{u,s\}$ and $\{u,t\}$ in memory.
	\end{enumerate}
\end{itemize}

If we perform the former exponentiation steps as is, both endpoints will aggregate one edge for each node in the path, which may break local and global memory restrictions. To resolve this issue, we implicitly assume the following scheme. If a node is connected to an endpoint, it keeps track if is the furthest away from said endpoint in the input graph, among all nodes that are connected to the endpoint. Immediately after initializing path exponentiation, the furthest away node is the neighbor of the endpoint in the input graph. During exponentiation:

\begin{itemize}
	\item If a node is the furthest away from an endpoint, when creating a new edge between an endpoint and an internal node, it informs both nodes of the new edge. 
	
	\item If a node is not the furthest away from an endpoint, when creating a new edge between an endpoint and an internal node, it only informs the internal node of the new edge. 
	
	\item If a node is an endpoint, upon receiving a new edge, it drops the old one. 
\end{itemize}

This scheme results in endpoints $s$ and $t$ effectively doing nothing during path exponentiation, except keeping track of the latest edge connecting them to an internal node. Eventually, exponentiation terminates when $s$ and $t$ get connected and all internal nodes have two edges, one for each endpoint. As the shortest distance between $s$ and $t$ in the virtual graph decreases by at least a factor of $3/2$ in each step, path exponentiation terminates in $O(\log n)$ time. Due to the aforementioned memory saving scheme, all nodes in $P$ keep at most two edges in memory, resulting in $O(m+n)$ global memory.

\begin{figure}
	\centering
	\includegraphics[width=0.55\textwidth]{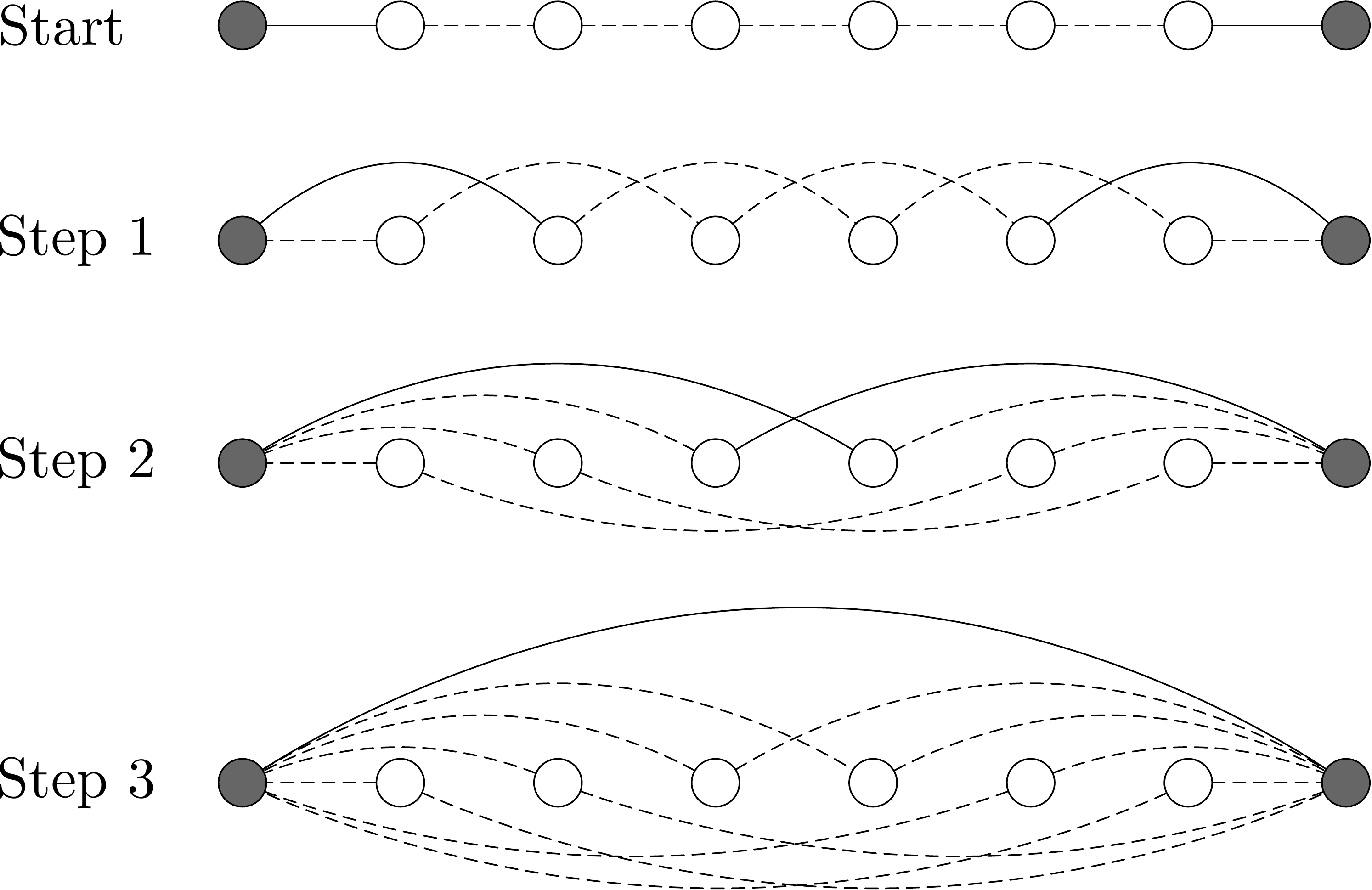}
	\caption{Path exponentiation on a path with 8 nodes. All edges are virtual, and solid edges emphasize the edges that endpoints keep track of. Exponentiation terminates in 3 steps, after which, internal nodes are connected to both endpoints and endpoints are connected by an edge.}
	\label{fig:effexp}
\end{figure}

\subsection{Rooting Algorithm} \label{subsection:treerooting}

The algorithm is split into two parts. First we find the root node (\Cref{subsubsec:findroot}), during which we set a collection of unoriented paths aside. Then, during postprocessing (\Cref{subsubsec:orientingpaths}), we orient said paths in parallel, resulting in a total runtime of $O(\log n)$, local memory $O(n^\delta)$, and global memory $O(m+n)$. We emphasize that we execute the following algorithms on the virtual graph, and not on the input graph. Initially, the virtual graph is an identical copy of the input graph. Also, when talking about neighbors and edges, we refer to virtual neighbors and edges, unless specified otherwise.

\subsubsection{Finding the Root} \label{subsubsec:findroot}

The high level idea is straightforward: perform path exponentiation in all current paths (note that now, endpoints are either leaf nodes or nodes of degree $\geq 3$), and when an endpoint of a path is a leaf that is connected to the other endpoint, we set the path aside (this will become apparent later). For this to work, one major issue must be addressed. An endpoint of degree $\geq 3$ can turn into a degree-$2$ node, extending the current path. The difficulty in this scenario stems from the fact that some nodes are in the middle of path exponentiation and some have not yet started. We resolve this issue by defining nodes that were endpoints of degree $\geq 3$ in the previous phase, but are nodes of degree $2$ in the current phase, as \textit{midpoints}.

\begin{enumerate}
	\item In phase $i$, each node $u$ in a path first identifies if it is an endpoint, a midpoint, or an internal node.
	
	\begin{itemize}
		\item If $u$ is a leaf node that is connected to the other endpoint, we set the path containing $u$ aside$^*$. Note that this also applies to paths of length 1.
		
		\item If $u$ is an endpoint that is not connected to the other endpoint, it does nothing except act as endpoint for the corresponding internal nodes.
		
		\item If $u$ is a midpoint such that both of its neighbors are other midpoints or endpoints, it transforms into an internal node and acts as such henceforth. 
		
		\item If $u$ is a midpoint such that at least one of its neighbors is an internal node, it does nothing except act as endpoint node for the corresponding internal nodes. 
		
		\item If $u$ is an internal node (or a midpoint that has turned into an internal node), it performs path exponentiation.
	\end{itemize}
	
	$^*$Setting path $P$ aside entails leaf node $s$ informing the other endpoint $t$ that the orientation is going to be from $s$ to $t$, so that the algorithm can proceed. The internal nodes of the path do not need to be informed that they are set aside, since they will not do anything for the remainder of the algorithm. Also, instead of $P$, we are actually setting aside $P \setminus t$, since $t$ may be of high degree and has to remain in the graph. Observe that this means that both endpoints of the path we are setting aside are leaves that know the orientation of the path. The edges of the paths remain unoriented until the root is found, after which these paths are oriented in parallel during the postprocessing in \Cref{subsubsec:orientingpaths}. Note that if both endpoints of a path are leaves, the algorithm terminates and the higher ID node is chosen as the root.
\end{enumerate}

\begin{proof}[Proof of \Cref{lemma:rooting}, Finding the root]
	
	Since we only orient paths connected to at least one leaf node, we end up with a valid orientation and a unique root.
	
	Consider endpoints $s$ and $t$ of some path during some phase. The aim of the algorithm is essentially to construct edge $\{s,t\}$. Now consider a current shortest (virtual) path $P_v$ between $s$ and $t$.  Observe that the nodes responsible for eventually creating edge $\{s,t\}$ constitute $P_v$. Hence, all other nodes are redundant and can be thought of as removed.
	
	In order to analyze the number of nodes that are removed in a phase, we want to first count the number of internal nodes in paths such as $P_v$. Since a midpoint is always incident to an internal node (otherwise it would transform into an internal node), at least $1/3$ of all nodes in $P_v$ are internal nodes. This is evident from the ``worst case'' where two consecutive midpoints are followed by one internal node.
	
	Since all internal nodes in $P_v$ perform path exponentiation, the number of internal nodes in $P_v$ drops by a factor of at least $3/2$ in one phase. Observe that in addition to removing at least $1/3$ of the internal nodes in all paths such as $P_v$, we also remove all leaf nodes. Since the average degree of a node in a tree is $<2$, the number of leaves in a tree is larger than the number of nodes of degree $\geq 3$. Hence, in each phase, we remove at least $1/3 \cdot 2/3 = 2/9$ of all of the nodes in the graph, and the algorithm finds the root after $O(\log n)$ phases.
	
	The only memory usage stems from path exponentiation, where in each path, in addition to the input edges (i.e., edges incident to a node in the input graph), all nodes keep at most two virtual edges in memory. Observe that endpoints can partake in multiple path exponentiations. However, since endpoints keep track of only one virtual edge (per path exponentiation), it is easy to see that an endpoint can never have more virtual edges than input edges. Hence, local memory $O(n^\delta)$ and global memory $O(m+n)$ are respected.
\end{proof}

\subsubsection{Postprocessing} \label{subsubsec:orientingpaths}

Before initializing this part, we first we have to ensure that the root finding has terminated, which can be done using the broadcast tree (\Cref{sec:broadcasttree}) in constant time. Then, we can start orienting the paths that were set aside by the root finding algorithm. Recall that they are paths where both endpoints are leaves that know the orientation. We want to orient all edges in these (possibly very long) paths in parallel.

\begin{enumerate}
	\item Consider performing path exponentiation on a path $P$ such that when an edge is created between an endpoint and an internal node, it is oriented according to the orientation information at the endpoint. Upon termination, all nodes orient their input edges according to the orientation of their virtual edges. Note that this requires nodes to keep track which virtual edge corresponds to which edge in the input graph.
\end{enumerate}

\begin{proof}[Proof of \Cref{lemma:rooting}, Postprocessing]
	Observe that an oriented edge is created only by nodes that already have an oriented edge (are of type 2 or 3 in Path Exponentiation) and hence, the orientation will be correct. As we only perform path exponentiation, the runtime is $O(\log n)$. Clearly, this only has a constant overhead compared to path exponentiation. Hence, local memory $O(n^\delta)$ and global memory $O(m+n)$ are respected.
\end{proof}

\begin{proof}[Proof of \Cref{lemma:rooting}, Arbitrary degree and component-stability]
	The extension to arbitrary-degree trees is straightforward, since all meaningful operations are performed on degree-2 nodes. If all edges of a node fit into one machine, nothing changes from the constant-degree case. Otherwise, we can use the broadcast tree structure \Cref{sec:broadcasttree}) for every node $v$ with degree $\omega(n^\delta)$. If some node $u$ wants to communicate with $v$, the communication happens with the machine storing edge $\{u,v\}$.
	
	Component-stability and the compatibility with forests is simple to argue about. The only communication between disconnected components happens in the beginning of postprocessing, when all components wait until the root has been found in all components. Clearly, this does not affect the resulting rooting in each component. It does however affect the runtime, since smaller components may have to wait until larger components have found the root. Hence, the runtime becomes $O(\log \max_i \{n_i\})$.
\end{proof}

\section{The High Regime Algorithm (Analysis and Implementation Details)}\label{sec:apphigh}

\subsection{Analysis of Phase I} \label{sec:appphaseone}
In this section, we will analyze Phase I of algorithm $\fA'$ defined in \Cref{sec:phaseone}. In particular, we will prove that Phase I, i.e., the leaves-to-root phase, terminates after $O(\log n)$ iterations.

\paragraph*{Properties of the Phase I Algorithm.}
We start by collecting some properties of the defined process.
They imply, in particular, that the update rules are well-defined.

\begin{observation}\label{prop:atmostone}
	Each node has, at any point in time, at most one outgoing active pointer.
\end{observation}
\begin{proof}
	This directly follows from the fact that, in the beginning, each node has at most one outgoing active pointer, and in each iteration, each active pointer is either merged into a larger pointer (if it points to a $2$-node), or left unchanged or removed.
\end{proof}

\begin{observation}\label{prop:state}
	If a node is inactive, it can never become active again.
	If a node is an $X$-node, where $X \in \{ 1, 2, 3 \}$, it will never in the further course of the process become a $Y$-node, where $Y\in \{ 1, 2, 3 \}$ and $Y > X$.
\end{observation}
\begin{proof}
	The observation follows from the definitions of the $1$-, $2$-, and $3$-nodes, the definition of $\merge(\cdot,\cdot)$, and the fact that the only new active pointers that are produced during our process are created via $\merge(\cdot,\cdot)$.
\end{proof}

\begin{observation}\label{obs:undec}
	For any $2$-node $v$ with relevant in-edge $e'$ and outgoing edge $e''$, we have $M^e(v) = \undec$ if and only if $e \in \{ e', e'' \}$. 
\end{observation}
\begin{proof}
	Due to the design of the update rules, if $e$ is an outgoing edge (for $v$), then $M^e(v)$ remains $\undec$ indefinitely, and if $e$ is an incoming edge, then $M^e(v)$ is set to some label set $L \neq \undec$ in the first iteration at the end of which there is no active pointer $p$ ending in $v$ and satisfying $\last_p = e$.
	(Note that we use here that there cannot have been a merge of two pointers starting and ending in $v$ so far since otherwise $v$ would not have any incoming pointers and could not be a $2$-node, by Observation~\ref{prop:state}.)
	Since the design of the update rules ensures that once there is no active pointer $p$ ending in $v$ and satisfying $\last_p = e$, this property does not change thereafter, we obtain the lemma statement, by the definition of a relevant in-edge.
\end{proof}

\begin{observation}\label{prop:three}
	At the end of each iteration, it holds that for each active pointer $(u, v)$, the unique path from $u$ to $v$ does not contain a $3$-node, except possibly $u$ and/or $v$.
	Also, if the path from $u$ to $v$ contains a $2$-node $w \neq u$, then the relevant in-edge of $w$ is the edge incoming to $w$ that lies on this path.
\end{observation}
\begin{proof}
	These statements follow since they hold in the beginning of the process and do not change during the process as the merge operation is only ``performed'' by active $2$-nodes (whose relevant in-edge will lie on the path corresponding to the pointer(s) they produce).
	Here, we implicitly use Observation~\ref{prop:state} and the fact that a $2$-node never changes its relevant in-edge (which follows with an analogous argument to the one used in the proof of Observation~\ref{prop:state}).
\end{proof}

\begin{observation}\label{prop:unweird}
	Let $p = (u, v)$ and $p' = (w, x)$ be two active pointers at an end of an iteration, and assume that the unique paths from $u$ to $v$ and from $w$ to $x$ intersect in at least one edge.
	Then there is a directed path that contains all of these four nodes.
\end{observation}
\begin{proof}
	Suppose for a contradiction that this is not the case, which implies that no directed path contains both $u$ and $w$, and let $y$ denote the lowest common ancestor of $u$ and $w$.
	In the beginning, node $y$ is an active $3$-node.
	At the point when $y$ stops being an active $3$-node (which has to happen due to Observation~\ref{prop:three}), the update rules ensure that there is at most one incoming edge $e$ at $y$ such that there exist an active pointer $(z, a)$ such that the unique path from $z$ to $a$ contains $e$.
	Now, the contradiction follows from an analogous argument to the one used in the proof of Observation~\ref{prop:three}.
\end{proof}

\begin{lemma}\label{lem:propcombi}
	At the end of each iteration it holds that (a) for each active node $u$, all nodes on the path from $u$ to $r$ are active, and (b) for any two active nodes $u, v$ such that $v$ is an ancestor of $u$, the node $x$ node $v$ points to is an ancestor of the node $w$ node $u$ points to, or $w = x$.
\end{lemma}
\begin{proof}
	Suppose for a contradiction that the lemma statement is false, and let $i$ be the first iteration such that at the end of iteration $i$ the statement is not satisfied.
	
	Consider first the case that property (b) does not hold, which implies that at the end of iteration $i$, there are active nodes $u, v$ with active pointers $p = (u, w)$, $p' = (v, x)$ such that $v$ is an ancestor of $u$, and $w$ an ancestor of $x$.
	Let $y$ and $z$ denote the nodes $u$ and $v$, respectively, were pointing to at the beginning of iteration $i$.
	Due to the minimality of $i$, we have that $z$ is an ancestor of $y$, or $y = z$, and that $y$ and $z$ are active at the beginning of round $i$.
	The pointer $(u, w)$ must be the result of a merge operation in iteration $i$, as otherwise $w = y$, which would imply that $w$ cannot be an ancestor of $x$.
	Hence, at the beginning of iteration $i$, the active pointer starting at $y$ must be $(y, w)$.
	Also, at the beginning of iteration $i$, we must have $z = x$, or the pointer starting at $z$ must be $(z, x)$ (as otherwise we could not have the active pointer $p' = (v, x)$ at the end of iteration $i$).
	In the latter case, we obtain $y \neq z$ (as otherwise $w = x$), and we see that the nodes $y, z, x, w$ satisfy that, at the beginning of iteration $i$, $z$ is an ancestor of $y$, $w$ is an ancestor of $x$, $z$ points to $x$, and $y$ points to $w$, which yields a contradiction to the minimality of $i$.
	Hence, we can assume that $z = x$.
	This implies that, at the beginning of iteration $i$, $z$ is not an active $2$-node (as otherwise $v$ would not point to $x$ at the end of iteration $i$); since $z$ is an active node (by Observation~\ref{prop:state}) and has an incoming pointer (from $v$), it must be a $3$-node.
	Since $w$ is an ancestor of $x = a$, and we (still) have $y \neq z$, we see that for the active pointer $(y, w)$ at the beginning of iteration $i$, the path from $y$ to $w$ contains a $3$-node that is distinct from both $y$ and $w$, yielding a contradiction to Observation~\ref{prop:three}.
	
	Now consider the second case, namely that property (a) does not hold, which implies that at the end of iteration $i$, there are two vertices $u, v$ such that $v$ is the parent of $u$, $u$ is active, and $v$ is inactive.
	Due to the minimality of $i$ and Observation~\ref{prop:state}, $v$ (as well as $u$) must have been active at the beginning of iteration $i$.
	By the design of the update rules, the only way in which $v$ can have become inactive at the end of iteration $i$ is that the node $w$ node $v$ points to at the beginning of iteration $i$ is a $3$-node that some leaf $x$ satisfying $\last_{(v,w)} = \last_{(x,w)}$ points to as well, at the beginning of iteration $i$.
	Let $y$ denote the end of the active pointer starting in $u$ at the beginning of iteration $i$.
	
	If $u$ lies on the path from $x$ to $w$, then, by property (b), we have that $y$ is an ancestor of $w$, or $y = w$.
	The former cannot be true, as otherwise we would have an active pointer (from $u$ to $y$) at the beginning of iteration $i$ such that the corresponding path contains an internal node that is a $3$-node (namely $w$), which would contradict Observation~\ref{prop:three}.
	However, also the latter cannot be true, as otherwise $u$ would have become inactive at the end of iteration $i$ since $\last_{(u,y)} = \last_{(x,w)}$.
	Hence, $u$ does not lie on the path from $x$ to $w$.
	
	If $y \neq v$, we the two active pointers $(u, y)$ and $(x, w)$ at the beginning of iteration $i$ yield a contradiction to Observation~\ref{prop:unweird}.
	Hence, $y = v$.
	Observe that $i \geq 2$, as at at the beginning of iteration $1$, the only pointers we have are the directed edges of $G'$, and the pointer $(x, w)$ is not such a pointer (as it contains the internal node $v$; we have $x \neq v$ as $x$ is a leaf while $v$ has a child, namely $u$).
	Hence, iteration $i-1$ exists, and at the beginning of iteration $i-1$, node $u$ must have pointed to node $v$ and node $v$ cannot have been a $2$-node, as otherwise we could not have an active pointer $(u,v)$ at the beginning of iteration $i$.
	Since, at the beginning of iteration $i - 1$, node $v$ had an incoming active pointer (from $u$), it cannot have been a $1$-node either, so it must have been a $3$-node.
	Now consider the node $z$ leaf $x$ was pointing to at the beginning of iteration $i-1$.
	As the active pointer starting in $x$ at the beginning of iteration $i$ is $(x,w)$, there are only $3$ possibilities, due to the design of the update rules: 1) $z = w$, or 2) $z \neq v$ lies on the path from $v$ to $w$, or 3) $z \neq v$ lies on the path from $x$ to $v$ and there is an active pointer from $z$ to $w$.
	In either case, we obtain a contradiction to Observation~\ref{prop:three}.
\end{proof}

\begin{lemma}\label{lem:noin23}
	When a node stops being an active $3$-node, it becomes an active $2$-node.
	When a node stops being an active $2$-node, it becomes an active $1$-node.
	When a node stops being an active node, it turns from an active $1$-node into an inactive $1$-node, and remains an inactive $1$-node until the end of Phase I.
	In particular, there are no inactive $2$- or $3$-nodes.
\end{lemma}
\begin{proof}
	Consider an active node $u$ that becomes inactive at the end of some iteration $i$.
	By the update rules, $u$ can only become inactive due to having an active pointer to some node $v \neq u$ at the beginning of iteration $i$, and $v$ having another incoming active pointer from some leaf $w$ (where, potentially, $w = u$) such that $\last_{(u, v)} = \last_{(w, v)}$.
	In particular, $u$ lies on the path from $w$ to $v$.
	At the beginning of iteration $i$, node $u$ cannot be a $3$-node or a $2$-node with the relevant in-edge not lying on the path from $w$ to $v$, since otherwise $u \neq w$, and the active pointer $(w, v)$ together with node $u$ would yield a contradiction to Observation~\ref{prop:three}.
	At the beginning of iteration $i$, node $u$ also cannot be a $2$-node with the relevant in-edge lying on the path from $w$ to $v$ as otherwise $u$ would have an incoming active pointer from some node $x \neq w$ on the path from $w$ to $u$, yielding a contradiction to \Cref{lem:propcombi}.
	Hence, $u$ is a $1$-node at the beginning of iteration $i$.
	By Observation~\ref{prop:state}, the inactive node that $u$ becomes at the end of iteration $i$ must be a $1$-node, and $u$ will remain an inactive $1$-node.
	
	Now consider an active $3$-node that stops being an active $3$-node at the end of some iteration $i$.
	By the above discussion, $u$ is still active at the end of iteration $i$, and, by the design of the update rules, $u$ retains at least one relevant in-edge, which implies that it becomes a $2$-node.
	
	Finally, consider an active $2$-node that stops being an active $2$-node at the end of some iteration $i$.
	Again, we obtain that $u$ is still active at the end of iteration $i$, and, again by the design of the update rules, we see that there can be at most one edge $e$ incoming at $u$ such that there exists an active pointer $p$ satisfying $\last_p = e$, which implies that $u$ is a $1$-node at the end of iteration $i$ (as $u$ stops being an active $2$-node).
\end{proof}

Due to \Cref{lem:noin23}, we will not have to distinguish between active and inactive $2$-nodes (or $3$-nodes) in the remainder of the paper as we know that such nodes cannot be inactive.

\paragraph*{Bounding the Number of Iterations.}

In the following, we will fix some notation that is required to prove that the number of iterations until Phase I terminates is in $O(\log n)$.
We will denote the induced tree consisting of \emph{active} nodes at the end of iteration $i$ by $T_i$; we set $T_0 \coloneqq G'$.
Due to \Cref{lem:propcombi}, we know that $T_i$ is indeed a (rooted) subtree of $G'$, and that its root is the root of $G'$, namely $\rooot$; we also know that, for any $i \geq 1$, tree $T_i$ is an induced subtree of $T_{i-1}$, due to Observation~\ref{prop:state}.
For each $T_i$, we denote the maximal connected components consisting of non-root degree-$2$ nodes by $B_{i,1}, \dots, B_{i,z_i}$, in an arbitrary, but fixed, order.
Here $z_i$ denotes the number of such maximal connected components in $T_i$.
We call the $B_{i,j}$ \emph{blocks} of $T_i$.
For simplicity, we will also use $B_{i,j}$ to denote the set of nodes of $B_{i,j}$.
In the following we will collect some insights about the $T_i$ and $B_{i,j}$.

\begin{lemma}\label{lem:leaves}
	Consider some iteration $i \geq 1$.
	Any leaf $u \neq \rooot$ of $T_i$ is also a leaf of $T_{i-1}$.
	Moreover, any leaf $u \neq \rooot$ of $T_i$ is also a leaf of $G'$.
\end{lemma}
\begin{proof}
	For a contradiction, suppose that, for some $i \geq 1$, tree $T_i$ contains a leaf $u \neq \rooot$ that is not a leaf of $T_{i-1}$.
	Note that $u$ cannot have any incoming active pointer at the end of iteration $i$, and therefore must be a $1$-node at that point in time.
	Let $v$ be a child of $u$ in $T_{i-1}$, and let $w$ denote the node $v$ is pointing to at the end of iteration $i-1$.
	Due to our assumption, $v$ is active at the end of iteration $i-1$, but inactive at the end of iteration $i$.
	Due to the design of our update rules, the only way in which this can happen is that at the end of iteration $i-1$, $w$ is a $3$-node or the root, and there is an active pointer $(x, w)$ from some leaf $x$ with $\last_{(x, w)} = \last_{(v, w)}$.
	Observe that, by \Cref{lem:noin23}, $u$ cannot be a $3$-node at the end of iteration $i - 1$ (as it is a $1$-node at the end of iteration $i$), which implies $w \neq u$.
	Hence, $w$ is an ancestor of $u$, and, by Observation~\ref{prop:three} and \Cref{lem:propcombi}, it follows that at the end of iteration $i-1$, the active pointer starting at $u$ must end in $w$, and $\last_{(u, w)} = \last_{(x, w)}$.
	But this implies, again by the design of the update rules, that if $v$ becomes inactive at the end of iteration $i$, then so does $u$.
	This yields a contradiction to the fact that $u$ is active at the end of iteration $i$.
	
	Since we showed that any leaf $u \neq \rooot$ in the tree of active nodes at the end of some iteration is also a leaf in the tree of active nodes at the end of the previous iteration, we obtain, by applying this argumentation iteratively, that $u$ must also be a leaf in $T_0 = G'$.
\end{proof}

\begin{corollary}\label{cor:niceblocks}
	Consider some iteration $i \geq 1$, and some block $B_{i,j}$.
	Let $u$ be a node in $B_{i,j}$.
	If $u$ is in some block $B_{i-1,j'}$, then $B_{i-1,j'} \subseteq B_{i,j}$ (considered as node sets).
\end{corollary}
\begin{proof}
	Let $u$ be as described in the lemma, and suppose, for a contradiction, that there is some node $v \neq u$ satisfying $v \in B_{i-1,j'}$ and $v \notin B_{i,j}$.
	By the definition of blocks, either $v$ is an ancestor of $u$, or $u$ is an ancestor of $v$.
	In the former case, observe that, due to \Cref{lem:propcombi}, all ancestors of $u$ in $T_{i-1}$ are also contained in $T_i$, which implies that $v$ is a degree-$2$ node in $T_i$ belonging to $B_{i,j}$, yielding a contradiction.
	In the latter case, observe that $v$ and its child in $T_{i-1}$ must be contained in $T_i$ as otherwise some node on the path from $v$ to $u$ must be a leaf in $T_i$ while not being a leaf in $T_{i-1}$, which would contradict \Cref{lem:leaves}.
	Now we obtain a contradiction in an analogous way to the previous case.
\end{proof}

For each block $B_{i,j}$ with $i \geq 1$, we denote the set of blocks $B_{i-1,j'}$ that have non-empty intersection with $B_{i,j}$ by $\prev(B_{i,j})$.
Due to \Cref{cor:niceblocks}, we know that the union of all node sets contained in $\prev(B_{i,j})$ is a subset of $B_{i,j}$.
Moreover, for each block $B_{i,j}$ with $i \geq 1$, we denote the set of vertices in $B_{i,j}$ that are not contained in some $B \in \prev(B_{i,j})$ by $\new(B_{i,j})$.

We will also need the notion of a pointer chain.
\begin{definition}[Pointer chain]
	A \emph{pointer chain} (from a node $c_0$ to a node $c_y$) at the end of some iteration $i$ is a finite sequence $C = (c_0, \dots, c_y)$ of nodes such that for any $1 \leq j \leq y$, there is an active pointer $(c_{j-1}, c_j)$ at the end of iteration $i$.
	We call a pointer chain a \emph{leaf-root pointer chain} if $c_0$ is a leaf in $G'$ and $c_y = \rooot$.
\end{definition}
Note that any pointer chain at the end of some iteration $i$ consists only of active nodes, i.e., of nodes from $T_i$ (this holds for the last node in the pointer chain due to \Cref{lem:propcombi}).

\begin{observation}
	For any iteration $i$, and any leaf $u \neq \rooot$ in $T_i$, there exists a leaf-root pointer chain from $u$ to $\rooot$ at the end of iteration $i$.
\end{observation}
\begin{proof}
	This follows directly from the definition of a pointer, \Cref{lem:propcombi}, and \Cref{lem:leaves}.
\end{proof}

In order to maintain a certain guarantee (given in \Cref{lem:leafroot}) throughout Phase I (that will help us to bound the number of iterations), we will need to assign an integer value $k_{i,j}$ to each block $B_{i,j}$ that, roughly speaking, provides an upper bound for the number of nodes from $B_{i,j}$ contained in any leaf-root pointer chain.
Define $\prev(k_{i,j})$ to be the set of all indices $j'$ such that $B_{i-1,j'} \in \prev(B_{i,j})$.
For each block $B_{0,j}$, we set $k_{0,j} \coloneqq |B_{0,j}|$.
For each block $B_{i,j}$ with $i \geq 1$, we set
\[
k_{i,j} \coloneqq |\new(B_{i,j})| + 1/2 \cdot \sum_{j' \in \prev(k_{i,j})} k_{i-1, j'}\enspace.
\]

\begin{lemma}\label{lem:leafroot}
	Consider a leaf-root pointer chain $C = (c_0, \dots, c_y = \rooot)$ at the end of some iteration $i \geq 0$ (where we set the end of iteration $0$ to be the starting point of our process).
	For each block $B_{i,j}$, the number of nodes contained in $C \cap B_{i,j}$ is at most $k_{i,j}$.
\end{lemma}
\begin{proof}
	We prove the statement by induction in $i$.
	For $i = 0$, the statement trivially holds, by the definition of $k_{0,j}$.
	Now, consider some $i \geq 1$, and assume that the statement holds for $i-1$.
	Let $C = (c_0, \dots, c_y)$ be an arbitrary leaf-root pointer chain at the end of iteration $i$, and let $C' = (c'_0 = c_0, c'_1, \dots, c'_{y'})$ denote the leaf-root pointer chain starting at $c_0$ at the end of iteration $i-1$.
	By the design of the update rules, the definition of the function $\merge(\cdot, \cdot)$, and \Cref{lem:propcombi}, the sequence $C$ is a subsequence of $C'$, i.e., $C$ is obtained from $C'$ by removing elements.
	Furthermore, we observe that any (non-root) degree-$2$ node in $T_{i-1}$ with an incoming active pointer at the end of iteration $i-1$ must be a $2$-node (by the definition of a $2$-node), and any (non-root) node $u$ of degree at least $3$ in $T_{i-1}$ must be a $3$-node at the end of iteration $i-1$ (as, for each child $v$ of $u$ in $T_{i-1}$, there must be an active pointer $(v, u)$, due to Observations~\ref{prop:three} and~\ref{prop:unweird}).
	
	Consider an arbitrary block $B_{i,j}$, and an arbitrary block $B_{i-1,j'} \in \prev(B_{i,j})$.
	Due to the definitions of a pointer and a block, the nodes in $C' \cap B_{i-1,j'}$ form a subsequence of $C'$ consisting of consecutive nodes $c'_p, \dots, c'_q$.
	By the definition of $\merge(\cdot, \cdot)$ and the fact that the nodes in $C'' \coloneqq (c'_p, \dots, c'_q)$ are degree-$2$ nodes in $T_{i-1}$ (and hence $2$-nodes at the end of iteration $T_{i-1}$), we see that for any two consecutive nodes in $C''$, at most one of the nodes is contained in $C$ (by the design of the update rules).
	Moreover, as $\merge$ operations are only ``performed'' by $2$-nodes, any node in $C'$ that is a (non-root) node of degree at least $3$ in $T_{i-1}$ (and hence a $3$-node) will be contained in $C$, which implies that $c'_p$ is not contained in $C$ (as $c'_{p-1}$ is contained in $C$ and points to $c'_{p+1}$ at the end of iteration $i$).
	Hence, we can conclude that for each $B_{i-1,j'}$, we have $|C \cap B_{i-1,j'}| \leq 1/2 \cdot |C' \cap B_{i-1,j'}| \leq 1/2 \cdot k_{i-1,j'}$.
	
	By \Cref{cor:niceblocks}, we have
	\[
	B_{i,j} = \new(B_{i,j}) \cup \bigcup_{B \in \prev(B_{i,j})} B \enspace,
	\]
	which yields
	\begin{align*}
		\left|C \cap B_{i,j}\right| &\leq \left|\new(B_{i,j})\right| + \sum_{B \in \prev(B_{i,j})} \left|C \cap B\right| \\
		&\leq \left|\new(B_{i,j})\right| + \sum_{j' \in \prev(k_{i,j})} \left(1/2 \cdot k_{i-1,j'}\right) = k_{i,j}
	\end{align*}
	as desired.
\end{proof}

In order to bound the number of iterations in Phase I, we will make use of a potential function argument.
Recall that $z_i$ denotes the number of blocks of $T_i$.
For each $i \geq 0$ (such that Phase I has not terminated after $i - 1$ iterations), set $\Phi_i \coloneqq \Phi'_i + \Phi''_i$, where $\Phi'_i$ is the number of leaves in $T_i$, and $\Phi''_i \coloneqq \sum_{1 \leq j \leq z_i} k_{i,j}$.

\begin{lemma}\label{lem:calc}
	Consider any iteration $i \geq 1$ such that Phase I does not terminate at the end of iteration $i$ or $i+1$.
	Then $\Phi_{i+1} \leq 7/8 \cdot \Phi_{i-1}$.
\end{lemma}
\begin{proof}
	For $i' \in \{ i-1, i \}$, let $X_j$ denote the number of leaves that are contained in $T_{i'}$, but not in $T_{i'+1}$.
	By \Cref{lem:leaves}, $X_{i'} = \Phi'_{i'} - \Phi'_{i'+1}$.
	
	By the definition of the $k_{a,j}$ and $\new(B_{a,j})$ (as well as \Cref{cor:niceblocks}), we have $\Phi''_{i'+1} \leq 1/2 \cdot \Phi''_{i'} + |\new(i'+1)|$, where $\new(i'+1)$ denotes the set of (non-root) nodes in $T_{i'+1}$ that have degree $2$ in $T_{i'+1}$ but not in $T_{i'}$.
	Recall (from the proof of \Cref{lem:leafroot}) that any (non-root) node of degree at least $3$ in $T_{i'}$ must be a $3$-node at the end of iteration $i'$, and observe that any leaf in $T_{i'}$ must be a $1$-node at the end of iteration $i'$.
	By the design of the update rules, it follows that to any node $u$ from the set $\new(i'+1)$, we can assign a leaf $f_u$ of $T_{i'}$ such that there is an active pointer $(f_u,u)$ at the end of iteration $i'$, and $f_u$ becomes inactive at the end of iteration $i'+1$.
	As $f_u \neq f_{u'}$ for any two nodes $u \neq u'$ from $\new(i'+1)$ (due to Observation~\ref{prop:atmostone}), we obtain $|\new(i'+1)| \leq X_{i'}$, which implies
	\[
	\Phi''_{i'+1} \leq 1/2 \cdot \Phi''_{i'} + X_{i'} \enspace.
	\]
	As the next step, we bound $\Phi'_{i'+1}$ in terms of $\Phi''_{i'}$.
	Let $\stays(i')$ be the set of all leafs $u$ of $T_{i'}$ such that at the end of iteration $i'$ the active pointer starting in $u$ does not end in a $3$-node or the root.
	Since all (non-root) nodes that have degree at least $3$ in $T_{i'}$ are $3$-nodes at the end of iteration $i'$ (as already observed above), any node $u \in \stays(i')$ must point to some (non-root) degree-$2$ node $f_u$ in $T_{i'}$.
	For any two distinct nodes $u, u' \in \stays{i'}$, the nodes $f_u$ and $f_{u'}$ must lie in different blocks of $T_{i'}$ (due to Observation~\ref{prop:three}), and each block $B_{i',j}$ containing such a node $f_u$ must satisfy $k_{i',j} \geq 1$ (as the leaf-root pointer chain starting in $u$ contains at least one node of $B_{i',j}$, namely $f_u$).
	Hence, $|\stays(i')| \leq \Phi''_{i'}$.
	Moreover, the design of the update rules ensures that out of all the leaves in $T_{i'}$ that point to a $3$-node or the root at the end of iteration $i'$, at least half will become inactive at the end of iteration $i'+1$.
	Thus, we obtain
	\begin{align*}
		\Phi'_{i'+1} &\leq |\stays(i')| + 1/2 \cdot \left(\Phi'_{i'} - |\stays(i')|\right) \\
		&= 1/2 \cdot \left( \Phi'_{i'} + |\stays(i')| \right) \leq 1/2 \cdot \left( \Phi'_{i'} + \Phi''_{i'} \right) \enspace.
	\end{align*}
	
	To finish our calculations, we consider two cases.
	Let us first consider the case that $\Phi''_{i-1} \geq 1/3 \cdot \Phi'_{i-1}$, which implies $\Phi'_{i-1} \leq 3/4 \cdot \Phi_{i-1}$.
	Then, using the equations and inequalities derived above, we obtain
	\begin{align*}
		\Phi_i = \Phi'_i + \Phi''_i &\leq \left( \Phi'_{i-1} - X_{i-1} \right) + \left( 1/2 \cdot \Phi''_{i-1} + X_{i-1} \right)\\
		&= \Phi'_{i-1} + 1/2 \cdot \Phi''_{i-1} \\
		&= 1/2 \cdot \Phi_{i-1} + 1/2 \cdot \Phi'_{i-1} \\
		&\leq 7/8 \cdot \Phi_{i-1} \enspace.
	\end{align*}
	Similarly to above, we see that
	\[
	\Phi_{i+1} \leq \Phi'_i + 1/2 \cdot \Phi''_i \leq \Phi_i \enspace,
	\]
	which implies
	\[
	\Phi_{i+1} \leq 7/8 \cdot \Phi_{i-1} \enspace.
	\]
	
	Now, consider the case that $\Phi''_{i-1} < 1/3 \cdot \Phi'_{i-1}$.
	Using the inequalities derived earlier, we obtain
	\[
	\Phi'_i \leq 1/2 \cdot \left( \Phi'_{i-1} + \Phi''_{i-1} \right) \leq 2/3 \cdot \Phi'_{i-1}
	\]
	and
	\[
	\Phi''_i \leq 1/2 \cdot \Phi''_{i-1} + X_{i-1} \leq \Phi''_{i-1} + \Phi'_{i-1} - \Phi'_i \enspace.
	\]
	Similarly to the previous case, we see that
	\begin{align*}
		\Phi_{i+1} \leq \Phi'_i + 1/2 \cdot \Phi''_i &\leq \Phi'_i + 1/2 \cdot \left( \Phi''_{i-1} + \Phi'_{i-1} - \Phi'_i \right)\\
		&= 1/2 \cdot \Phi''_{i-1} + 1/2 \cdot \left( \Phi'_{i-1} + \Phi'_i \right)\\
		&\leq 1/2 \cdot \Phi''_{i-1} + 5/6 \cdot \Phi'_{i-1} \\
		&\leq 5/6 \cdot \Phi_{i-1} \enspace.
	\end{align*}
	Hence, in both cases, we have $\Phi_{i+1} \leq 7/8 \cdot \Phi_{i-1}$, as desired.
\end{proof}

Using \Cref{lem:calc}, we are finally able to bound the number of iterations in Phase I and prove \Cref{lem:runtimephaseone}.

\begin{lemma}[Restating \Cref{lem:runtimephaseone}]
	Algorithm $\aone$ terminates after $O(\log n)$ iterations.
\end{lemma}
\begin{proof}
	Suppose for a contradiction that there is no constant $c$ such that $\aone$ always terminates after at most $c \cdot \log n$ iterations.
	Observe that $\Phi_0 = \Phi'_0 + \Phi''_0 \leq n + n = 2n$.\footnote{Note that the compatibility tree $G'$ has actually only $O(n/\log n)$ nodes, by \Cref{cor:fewernodes}, but upper bounding this by $n$ suffices.}
	By \Cref{lem:calc}, there exists some constant $c$ such that $\Phi_{c \cdot \log n} < 1$.
	By the definition of $\Phi_i$, it follows that the tree $T_{c \cdot \log n}$ of active nodes obtained after $c \cdot \log n$ iterations does not contain any leaves (apart from, potentially, the root).
	This implies that there is no active pointer after $c \cdot \log n$ iterations, which implies that $\aone$ terminates after at most $c \cdot \log n$ iterations, yielding a contradiction.
\end{proof}

\subsection{Analysis of Phase II}\label{sec:appphasetwo}
In this section, we will analyze Phase II of algorithm $\fA'$ defined in \Cref{sec:phasetwo}. In particular, we will prove that Phase II terminates after $O(\log n)$ iterations.
We will start by providing the missing proofs for
\Cref{lem:attheroot,lem:nicesplit,lem:between}
and Observation~\ref{obs:summary}.

\begin{lemma}[Restating \Cref{lem:attheroot}]
	For each edge $e$ incoming to the root $\rooot$, there is precisely one pointer $p = (u, \rooot) \in \ps_{\fin}$ such that $u$ is a leaf and $\last_p = e$.
\end{lemma}
\begin{proof}
	Fix an arbitrary edge $e$ incoming to the root $\rooot$.
	In the beginning of Phase I, there is an active pointer $p'$ that ends in $\rooot$ and satisfies $\last_{p'} = e$.
	Due to the design of the update rules for Phase I, this can only change once such a pointer that additionally starts in a leaf has been added to the pointer set.
	As at the end of Phase I, there is no active pointer left, it follows that there is at least one pointer $p$ that ends in $\rooot$ and satisfies $\last_{p} = e$.
	
	In order to show that there is at most such pointer, consider the first iteration $i$ in which such a pointer appeared in the set of active pointers (and therefore also in the set of pointers).
	By Observation~\ref{prop:unweird}, there can only be one such pointer at the end of iteration $i$.
	Moreover, due to the design of the update rules, no pointer $p''$ satisfying $\last_{p''} = e$ is added to the set of pointers in any later iteration (as no active pointer $p'''$ satisfying $\last_{p'''} = e$ remains at the end of iteration $i+1$).
	It follows that for edge $e$, there is precisely one pointer as described in the lemma.
\end{proof}

\begin{lemma}[Restating \Cref{lem:nicesplit}]
	Let $p = (u,v)$ be a pointer in $\ps_{\fin}$ with $\pred_p \neq \bot$.
	Then $\pred_p$ has degree at least $2$ in $G'$.
	Moreover,
	\begin{enumerate}
		\item if $\pred_p$ has degree $2$, then $(u, \pred_p), (\pred_p, v) \in \ps_{\fin}$, and
		\item if $\pred_p$ has degree at least $3$, then $(u, \pred_p), (\pred_p, v) \in \ps_{\fin}$, and for each edge $e$ incoming at $\pred_p$ that does not lie on the path from $u$ to $v$, there is exactly one pointer $p' = (w, \pred_p) \in \ps_{\fin}$ such that $w$ is a leaf and $\last_{p'} = e$.
	\end{enumerate}
\end{lemma}
\begin{proof}
	Since $\pred_p \neq \bot$, the pointer $p$ must be the result of a merge operation, which, by the definition of $\merge(\cdot,\cdot)$, implies that $\pred_p$ is an internal vertex of the path from $u$ to $v$.
	Hence, $\deg(\pred_p) \geq 2$.
	
	First, consider the case that $\deg(\pred_p) = 2$.
	From the design of the update rules of Phase I and the definition of $\merge(\cdot,\cdot)$, it follows directly that at some point during Phase I, there must have existed active pointers $(u, \pred_p), (\pred_p, v)$.
	This implies $(u, \pred_p)$, $(\pred_p, v) \in \ps_{\fin}$.
	
	Now, consider the case that $\deg(\pred_p) \geq 3$.
	Analogously to the previous case, we obtain $(u, \pred_p), (\pred_p, v) \in \ps_{\fin}$.
	Now what is left to be shown follows from an analogous argumentation to the one provided in the proof of \Cref{lem:attheroot}, with only one difference: for the considered edge $e$ incoming at $\pred_p$, it could also be the case that, at the point in Phase I where the property that there is an active pointer $p''$ that ends in $\pred_p$ and satisfies $\last_{p''} = e$ becomes false, this happens due to step~\ref{step2}, and not due to step~\ref{step3}, of the update rules.
	However, observe that any vertex $x$ can ``perform'' merge operations in at most one iteration during Phase I (where the node considered to perform a merge operation is the node $\pred_{p'''}$ where $p'''$ is the pointer created during the merge operation) since, by the design of the update rules, the merge operations of that iteration will make sure that no active pointer that ends in $x$ remains (which cannot change thereafter).
	Observe further that for $x = \pred_p$, each of those merge operations must have merged two pointers where the one incoming at $\pred_p$ (let us call it $q$) satisfies $\last_q = e'$ where $e'$ is the edge incoming at $\pred_p$ that lies on the path from $u$ to $v$.
	Hence, the aforementioned difference only applies to pointers $q$ ending at $\pred_p$ satisfying $\last_q = e'$, and since the lemma statement only concerns pointers $p'$ with $\last_{p'} \neq e'$, that difference is irrelevant, and we can simply apply the argumentation from the proof of \Cref{lem:attheroot}.
\end{proof}

\begin{lemma}[Restating \Cref{lem:between}]
	Consider any $i \geq 1$, and any edge $e = (u,v) \in E(G')$.
	If $\done(i-1)$ does not contain the pointer $p = (u,v)$, then there is exactly one pointer $(w,x) \in \timey(i)$ such that $e \in \betw(w,x)$.
	If $\done(i-1)$ contains the pointer $p = (u,v)$, then there is no pointer $(w,x) \in \timey(i)$ such that $e \in \betw(w,x)$.
\end{lemma}
\begin{proof}
	We prove the statement by induction in $i$.
	For $i = 1$, we have $\done(i-1) = \emptyset$, so $p \notin \done(i-1)$.
	By \Cref{lem:attheroot}, for each edge $e'$ incident to $\rooot$, there is exactly one pointer $p_{e'} \in \timey(1)$ with $\last_{p_{e'}} = e'$, and $p_{e'}$ is guaranteed to be a leaf-root pointer.
	By the definition of $\betw(\cdot)$, it follows that $e$ is contained in $\betw(p_{e'})$ where $e'$ is the unique edge incident to $\rooot$ that lies on the path connecting $e$ with $\rooot$, and that $e$ is not contained in $\betw(p'')$ for any pointer $p'' \neq p_{e'}$ from $\timey(1)$.
	This covers the base of the induction.
	
	For the induction step assume that the lemma statement holds for $i-1$ (where $i \geq 2$); we aim to show that it then also holds for $i$.
	Consider first the case that $\done(i-1)$ does not contain the pointer $p = (u,v)$.
	Then also $\done(i-2)$ does not contain the pointer $p = (u,v)$, and the induction hypothesis guarantees that there is exactly one pointer $p' = (y,z) \in \timey(i-1)$ such that $e \in \betw(p')$.
	By the definitions of $\timey(i)$ and $\betw(\cdot)$, the only pointers $p'' \in \timey(i)$ that could possibly satisfy $e \in \betw(p'')$ are $(y, \pred_{p'})$, $(\pred_{p'}, z)$, and some $(a, \pred_{p'})$ where $a$ is a leaf and $\last_{(a, \pred_{p'})}$ does not lie on the path from $y$ to $z$.
	Now, \Cref{lem:nicesplit} (together with the definition of $\timey(i)$) guarantees that $e$ is contained in $\betw(p'')$ for exactly one of those possible choices for $p''$, since the sets $\betw(y, \pred_{p'})$, $\betw(\pred_{p'}, z)$, $\betw(b_1, \pred_{p'})$, $\dots$, $\betw(b_{\deg(\pred_{p'})-2}, \pred_{p'})$ are pairwise disjoint and their union is $\betw(y, z)$.
	(Here, $b_1, \dots, b_{\deg(\pred_{p'})-2}$ are the starting vertices of the precisely $\deg(\pred_{p'})-2$ pointers $p''' \in \timey(i)$ ending in $\pred_{p'}$ and satisfying that $\last_{p'''}$ does not lie on the path from $y$ to $z$, whose existence is guaranteed by \Cref{lem:nicesplit}.)
	Hence, there is exactly one pointer $(w,x) \in \timey(i)$ such that $e \in \betw(w,x)$, as desired.
	
	Now consider the second case, i.e., that $\done(i-1)$ contains the pointer $p = (u,v)$.
	By the induction hypothesis, there is no pointer $p' \neq p$ contained in $\timey(i-1)$ such that $e \in \betw(p')$.
	By the definitions of $\timey(i)$ and $\betw(\cdot)$ (and the fact that $\pred_p = \bot$), it follows that there is no pointer $(w,x) \in \timey(i)$ such that $e \in \betw(w,x)$, as desired.
\end{proof}

\begin{observation}[Restating Observation~\ref{obs:summary}]
	For any $i \geq 2$, and any pointer $p' \in \timey(i)$, there is exactly one pointer $p \in \timey(i-1)$ such that $p' \in \succc(p)$.
	Moreover, for any $i \geq 1$, and any pointer $p = (u,v) \in \timey(1)$ with $\pred_p \neq \bot$, we have $\succc(p) = \{ (u, \pred_p), (\pred_p, v), p_1$, $\dots$, $p_{\deg(\pred_p)-2} \}$ where each $p_j$ is a pointer starting in a leaf, ending in $\pred_p$, and satisfying $\last_{p_j} = e_j$, where $e_1, \dots, e_{\deg(\pred_p)-2}$ are the $\deg(\pred_p)-2$ edges incoming to $\pred_p$ that do not lie on the path from $u$ to $v$. 
\end{observation}
\begin{proof}
	The observation follows from the definition of $\timey(\cdot)$, \Cref{lem:nicesplit,lem:between}.
\end{proof}

Next, we show that $\atwo$ is well-defined and correct.

\paragraph*{Well-Definedness and Correctness.}

From the description of the algorithm in Phase II it is not clear that the labels with certain properties the algorithm is supposed to output do actually exist.
In order to show that the algorithm is well-defined, we first need a helper lemma based on the following definitions.

\begin{definition}
	Let $e = (u, v)$ be an edge in $G'$.
	We denote the set of vertices that have ancestor $u$ or are equal to $u$ by $V_e$, the set of edges with at least one endpoint in $V_e$ by $E_e$, and the set of half-edges $(w, e')$ with $e' \in E_e$ by $H_e$.
	A labeling of the half-edges in $H_e$ (with a label from $\sout$ each) is a \emph{correct solution} for $H_e$ if it is a correct solution on the compatibility subtree $CS_e$ induced by $V_e \cup \{ u \}$ where we do not have any constraint for node $u$ (in the definition of a correct solution for a compatibility tree (see \Cref{def:comptree})).
\end{definition}

\begin{definition}
	Let $u,v$ be two nodes such that $v$ is an ancestor of $u$.
	We denote the set of nodes that are an endpoint of some edge in $\betw(u,v)$ by $V_{u,v}$, and the set of half-edges $(w, e)$ with $e \in \betw(u,v)$ by $H_{u,v}$.
	A labeling of the half-edges in $H_{u,v}$ (with a label from $\sout$ each) is a \emph{correct solution} for $H_{u,v}$ if it is a correct solution on the compatibility subtree $CS_{u,v}$ induced by $V_{u,v}$ where we do not have any constraint for nodes $u$ and $v$ (in the definition of a correct solution for a compatibility tree).
\end{definition}

\begin{lemma}\label{lem:encodecomplete}
	At the beginning of Phase I (i.e., for $i = 0$), and after any iteration $i$, the following two properties hold.
	\begin{enumerate}
		\item For any half-edge $(u,e)$ satisfying $M^e(u) \neq \undec$, the set $M^e(u)$ contains precisely the output labels $\ell$ such that there exists a labeling of the half-edges in $H_e$ that is a correct solution on $CS_e$ and labels $(u,e)$ with $\ell$.
		\item For any pointer $p = (v,w) \in \ps_i$, the set $\pairs_p$ contains precisely the pairs $(\ell, \ell')$ of output labels such that there exists a labeling of the half-edges in $H_{v,w}$ that is a correct solution on $CS_{v,w}$ and labels $(v,\first_p)$ with $\ell$ and $(\last_p,w)$ with $\ell'$.
	\end{enumerate}
\end{lemma}
\begin{proof}
	We prove the statement by induction in the first iteration $i$ in which $M^e(u)$ was set to some label set $L \neq \undec$, resp.\ in which $p$ was added to the set of pointers.
	The base of the induction is implied by the initialization of $M^e(u)$ and the pointer set at time $i = 0$.
	The induction step directly follows from the precise definition of the two steps in the update rules that create new pointers and change the values $M^e(u)$, namely step~\ref{step2} (which relies on the precise definition of the merge operation) and step~\ref{itemb}, respectively.
\end{proof}

Now, we are set to show that the algorithm for Phase II is well-defined and correct.

\begin{lemma}\label{lem:wellandcorrect}
	Assuming that a correct solution for the compatibility tree $G'$ exists, the algorithm $\atwo$ for Phase II is well-defined and correct.
\end{lemma}
\begin{proof}
	As the first step, we show that each time $\atwo$ is supposed to choose some output label that has to be contained in $\comp^e(u)$ for some half-edge $(u,e)$, we have $\comp^e(u) \neq \undec$.
	From the description of $\atwo$, it follows that when the above situation occurs, the half-edge $(u,e)$ in question is 1) incident to the root, or 2) incident to some node of the form $\pred_p$ for some pointer $p = (v,w) \in \ps_{\fin}$ and $e$ does not lie on the path from $v$ to $w$.
	If $(u, e)$ is incident to the root, i.e., if $u = \rooot$, then \Cref{lem:attheroot}, together with the design of the update rules and the fact that at the end of Phase I no active pointer remains, implies that $\comp^e(u) \neq \undec$.
	In the other case, observe that the existence of $p$ implies that at some point during Phase I, $\pred_p$ was a $2$-node with relevant in-edge on the path from $v$ to $w$ (due to the design of the update rules).
	Now, Observation~\ref{obs:undec} yields the desired inequality.
	
	Next, we show by induction that at each step of $\atwo$, labels with the required properties are available and the partial solution produced by choosing such labels is part of some correct global solution.
	Note that, due to \Cref{lem:between} and Observation~\ref{obs:summary}, it suffices to show the induction step for each set $\succc(p)$ of pointers separately as the processing of two distinct sets $\succc(p), \succc(p')$ is independent of each other (due to the facts that the half-edges considered when processing $\succc(p)$ are separated by some already selected output label from the half-edges considered when processing $\succc(p')$, and that the correctness of a solution for a compatibility tree is defined via constraints on edges and constraints on nodes).
	For simplicity, we will use the notation from the description of $\atwo$ in the following.
	
	For the base of the induction, observe that for the first step of point (a), there is a choice of output labels with the described properties due to \Cref{lem:encodecomplete} (in conjunction with the very first step of this proof) and the fact that a correct solution for $G'$ exists.
	Moreover, the obtained partial solution is part of some correct global solution due to the properties required in the first step of point (a) (and \Cref{lem:encodecomplete}).
	Observe also that the same holds for the second step of point (a): \Cref{lem:encodecomplete} together with the fact that there exists a correct global solution that respects the partial coloring computed so far ensures that labels with the required properties exist; the properties (and \Cref{lem:encodecomplete}) in turn imply that the new obtained partial solution is still part of some correct global solution.
	(A bit more concretely, the fact that the partial solution after the first step of point (a) is extendable to a correct global solution implies that there must be a label $\ell^*$ as described in the second step of point (a) since $\pairs_{p_j}$ contains all label pairs that can be completed to a correct solution inside the subtree hanging from $\last_{p_j}$ (by \Cref{lem:encodecomplete}) and the condition $(\ell^*) \in S_{u_j}$ just states that the output is correct ``at $u_j$''; the fact that the resulting output label pair (at the half-edges $(r,\first_{p_j})$ and $(\last_{p_j},u_j)$) is contained in $\pairs_{p_j}$ implies that the new partial solution can still be extended to a correct global solution, again due to the characterization of $\pairs_{p_j}$ given in \Cref{lem:encodecomplete}.)
	
	For the induction step, an analogous argumentation shows that, also for point (b), the extendability of the obtained partial solution implies the availability of labels with the stated properties, and the properties of the labels imply the extendability of the new obtained partial solution to a correct global solution.
	
	To prove the correctness of the algorithm and conclude the proof, given the above, it suffices to show that each half-edge becomes labeled at some point.
	To this end, observe that \Cref{lem:between}, Observation~\ref{obs:summary}, and the definition of $\timey(\cdot)$ imply that any pointer $p \in \ps_{\fin}$ is containedagree in at most one $\timey(i)$.
	Since there are only finitely many pointers in $\ps_{\fin}$, there must be some positive $i$ such that $\timey(i) = \emptyset$; hence $\atwo$ terminates.
	Since \Cref{lem:between} implies that all basic pointers have been processed when $\atwo$ terminates, we obtain the desired statement that each half-edge becomes labeled.
\end{proof}

\paragraph*{Bounding the Number of Iterations.}

What is left to be done is to bound the number of iterations in Phase II.

\begin{lemma}[Restating \Cref{lem:runtimephasetwo}]
	Algorithm $\atwo$ terminates after $O(\log n)$ iterations.
\end{lemma}
\begin{proof}
	Let $p = (u,v)$ be some pointer contained in some $\timey(i)$ satisfying $\pred_p \neq \bot$, and let $p' \in \succc(p)$.
	Let $j$, resp.\ $j'$, denote the first iteration in Phase I such that $p$, resp.\ $p'$, was an active pointer at the end of iteration $j$ (possibly $j = 0$ or $j' = 0$).
	Our first goal is to show that $j' < j$.
	
	To this end, observe that, due to the design of the update rules in Phase I (and the fact that the (active) outgoing pointer of a node can only grow), $p$ must be the result of the merge operation $\merge(q, q')$, where $q = (u, \pred_p)$ and $q' = (\pred_p, v)$, and this merge must have been performed in iteration $j$.
	Hence, $q$ and $q'$ must have been active pointers in iteration $j - 1$, which implies that, if $p' \in \{ q, q'\}$, then $j' < j$, as desired.
	Thus, assume that $p' \notin \{ q, q'\}$, which, by Observation~\ref{obs:summary}, implies that $p' = (w, \pred_p)$ for some leaf $w$, and $\last_{p'}$ does not lie on the path from $u$ to $v$.
	Since, at the end of iteration $j - 1$, node $\pred_p$ is a $2$-node with its relevant in-edge lying on the path from $u$ to $v$ (as $\merge(q, q')$ is performed in iteration $j$), each active pointer $q''$ ending in $\pred_p$ at the end of iteration $j - 1$ must satisfy that $\last_{q''}$ lies on the path from $u$ to $v$, and this fact cannot change in the further course of Phase I.
	Hence, $p'$ must have been active before iteration $j - 1$, and, again, we obtain $j' < j$.
	
	By Observation~\ref{obs:summary}, we can conclude that for any pointer $p'$ in any $\timey(i')$, there must be a pointer $p$ in $\timey(i'-1)$ such that the first iteration in Phase I at the end of which $p$ was active is strictly larger than the first iteration in Phase I at the end of which $p'$ was active (where we consider the starting configuration to be ``at the end of iteration $0$'').
	This implies that $\timey(i'' + 2) = \emptyset$, where $i''$ is the number of iterations in Phase I.
	Hence, Algorithm $\atwo$ terminates after $O(\log n)$ iterations, by \Cref{lem:runtimephaseone}.
\end{proof}

\subsection{Implementation in the \mpc Model} \label{sec:highimple}
In this section, we describe how to implement algorithm $\fA$ from \Cref{sec:high} in the low-space \mpc model. As the implementation of the rooting is provided in \Cref{sec:rooting}, we can focus on the main two phases.

We first consider the preprocessing of $\fA$ from \Cref{sec:trafothere}, i.e., the part of the algorithm where we bring the number of nodes down to $O(n/\log n)$, and the part where the solution computed on the compatibility tree in \Cref{sec:phasetwo} is transformed into a solution for the considered \lcl, i.e., algorithm $\atwo$.
Both parts can be implemented in a straightforward manner, due to \Cref{lem:gi}, unless we run into memory issues due to the fact that we have to execute $O(\log \log n)$ of the iterations described in \Cref{sec:trafothere}, instead of just one.
Note that a global memory overhead can only be possibly produced by the new edges that are introduced in the graphs $G_1, G_2, \dots, G_t$ since the node set only shrinks during that process and the memory required for the ``compatibility information'' of the compatibility graphs cannot be asymptotically larger than the memory required for storing the edges.
Moreover, as the total number of edges (produced during the process) that are incident to any particular vertex is at most $O(\log \log n)$, we cannot run into issues with the local memory.
Hence, it suffices to show that the total number of edges produced during the $O(\log \log n)$ iterations does not exceed $\Theta(n)$.
However, this directly follows from the fact that for each new edge that is introduced during those iterations, a node is removed.

For the remainder of the section, we consider the part of $\fA$ from \Cref{sec:phaseone}, i.e., algorithm $\aone$, that solves the compatibility tree.
We start by collecting all information that has to be stored during $\fA'$.
For simplicity, we already assign this information to the nodes of the compatibility tree $G'$.
We observe that unless a node has to store more than $n^{\delta}$ words or the total amount of information to be stored is in $\omega(n)$, the algorithm can be na\"{\i}vely implemented by standard techniques.
Soon we will see that the total amount of information to be stored is in $O(n)$, and the only issue to be taken care of is that the local memory of ``nodes'' is exceeded.
We will explain later how to resolve this issue.

For Phase I (see \Cref{sec:phaseone}), we maintain two pieces of information, namely
\begin{enumerate}
	\item a set of pointers, and
	\item the sets $\comp(u)$.
\end{enumerate}
For each pointer $p$, some additional information ($\pairs_p$, $\pred_p$, $\first_p$, $\last_p)$, and whether it is active or not) has to be stored, but since the memory required for this additional information is only a constant multiple of the memory required to store the pointer itself (in particular as there are only a constant number of output labels in $\sout$), we can ignore this information.
Moreover, as the design of the update rules in Phase I ensures, the number of pointers produced in each iteration is upper bounded by the number of nodes of the compatibility tree, which is $O(n/\log n)$, by \Cref{cor:fewernodes}.
Since, by \Cref{lem:runtimephaseone}, there are only $O(\log n)$ iterations in Phase I, the total number of pointers that have to be stored is in $O(n)$; hence, our global memory of $O(m+n)$ is not exceeded.
We already note that we will store each pointer $p = (u,v)$ at both of its endpoints; the overhead introduced by this does not change the required global memory asymptotically.

Together with a set $\comp(u)$, we also have to store the information about all the leaf-root pointers in $\ps_{\fin}$ that end in $u$ (in order to perform the steps in Phase II); however, by the design of the update rules of Phase I and the fact that degrees are bounded (and $|\sout|$ is constant), all of this information requires just a constant number of words to be stored.
We will store each set $\comp(u)$ and its associated information in node $u$; as the required amount of memory per node is constant (in words), we can ignore this information in the remainder of this section.

For the implementation of Phase II, the information stored in Phase I is still required, but will not be changed or expanded.
Note that the characterization of the pointers in $\timey(i)$, given by Observation~\ref{obs:summary}, provides a straightforward implementation: the pointers that are processed in iteration $1$ are easily identified (as they are the only leaf-root pointers in $\ps_{\fin}$), and the pointers to be processed in any later iteration are precisely those in a set of the form $\succc(p)$, where $p$ is a pointer processed in the previous iteration.
Observe also that, by \Cref{lem:nicesplit}, Observation~\ref{obs:summary} and the fact that degrees are bounded, each node is involved in the processing of only a constant number of pointers in each iteration.
Hence, each iteration of Phase II can be easily implemented in a constant number of rounds.

From the description of the update rules of Phase I, it is easy to see that, again, each iteration (now of Phase I) can be performed in constant time provided that we can perform all merge operations (i.e., step~\ref{step2}) in constant time.
From the above discussion, it follows that, in order to obtain the desired runtime of $O(\log n)$ rounds for the complete algorithm $\fA$, the only thing left to be done is to show that we can perform the merge operations in each iteration in constant time while storing the pointers in a way that does not exceed the local memory of the machines.
In the following, we explain how to achieve this.
Note that the merge operation that creates a pointer $(u,w)$ from $(u,v)$ and $(v,w)$ can be understood as $v$ forwarding (the head of) pointer $(u,v)$ to node $w$.

\paragraph*{Pointer Forwarding Tree.}

Our approach relies on a broadcast tree structure that we create for each node with a large number of incoming active pointers (see \Cref{def:aggTreeStructure} of \Cref{sec:broadcasttree}).
Each node $v$ creates a $n^\delta$-ary virtual rooted tree, where the idea is to store the incoming active pointers in the leaves of the tree.
Importantly, different nodes might be stored at different machines, but since the number of incoming active pointers is bounded by $O(n)$, the communication tree has constant depth which allows us to perform operations efficiently.

To perform the actual pointer forwarding, consider a non-virtual node $v$ and suppose that $v$ wants to forward its incoming active pointers to the non-virtual node $u$ in which the active pointer starting from $v$ ends.
Notice that the active pointers incoming to $v$ are stored in the leaves of the virtual tree $T_v$ rooted at $v$ and similarly for $u$ in the tree $T_u$ rooted at $u$.
Now, we can simply attach the tree $T_v$ to the node in $T_u$ that currently stores the pointer $(u,v)$ (for $u$).
Thereby, the pointers previously incoming to $v$ are now stored in the broadcast tree of $u$, and are therefore incoming to $u$.

This might, however, result in the depth of the broadcast tree increasing by an additive term of $1/\delta$.
To mend this, consider the following balancing process.
\begin{observation}
	Let $T_u$ be a virtual rooted tree of depth $d = O(1/\delta)$ of at most $n$ nodes.
	Then, in $O(1/\delta)$ rounds, we can reduce the depth to $2/\delta$ such that all the leaf nodes of $T_u$ are still leaf nodes.
\end{observation}
\begin{proof}
	The root initiates the following operation.
	First, using converge-cast, it learns the number $n_v$ of nodes in each (virtual) tree rooted from each of its child $v$.
	Then, the root creates a new $n^\delta$-ary virtual tree $T^*$ of depth $O(1/\delta)$.
	Proportionally to the number $n_v$, root $u$ assigns subtrees of $T^*$ to child $v$, such that all incoming pointers corresponding to $v$ fit into the subtree.
	Clearly, this is possible since $n_v \leq n$ and the new virtual node is assigned to at most $n^\delta$ children of $u$.
	This process is recursively continued until the leaves of $T^*$ are assigned to the leaves of $T_u$.
	Then, we can change the pointers from the old broadcast tree nodes to the new ones, and we have obtained our broadcast tree of depth $O(1/\delta)$.
	Notice that a na\"{\i}ve implementation results in a $1/\delta$ number of converge-casts, but the number of leaves per subtree can be pre-computed and stored.
\end{proof}

\subsection{Proof of Theorem~\ref{thm:high}}\label{sec:proofhigh}

Before executing anything, we first run the deterministic connected components algorithm from \cite{coy2021deterministic} on $G^2$ (which is another graph on the same vertex set, but in which two vertices are adjacent if their distance in $G$ is at most 2) that runs in $O(\log D)+O(\log_{m/n}\log n)=O(\log n)$ rounds with $O(m+n)$ words of global memory. The algorithm is component-stable, if, when contracting (during the algorithm of \cite{coy2021deterministic}), we aggregate the minimum ID for every component. Using this minimum ID, all nodes can compute the size of their component using the aggregation tree structure (see \Cref{def:aggTreeStructure} in \Cref{sec:broadcasttree}).

Before running the algorithm of \Cref{sec:high}, we root the input graph using the method described in \Cref{sec:rooting}, which is compatible with forests and is component-stable. 

Regarding the algorithm itself, all arguments are local, i.e., nodes in separate components do not communicate, the algorithm is component-stable. Furthermore, since every node knows $n_i$, which is the size of component $i$ it belongs to, we can substitute $n$ with $n_i$ in all global memory arguments of the section. Since $\sum n_i=n$, the global memory bound holds. As we now can apply our algorithm on each component separately, we can in the rest of the proof assume that we are given a single tree. 

The runtime bound follows from  \Cref{cor:fewernodes} and \Cref{lem:gi,lem:runtimephaseone,lem:runtimephasetwo} as these bound the runtimes of Phase~I (leaves-to-root),  Phase~II (root-to-leaves), and of each iteration of the preprocessing phase and the postprocessing phase. Additionally, each of the steps can be implemented under the memory constraints given by the low-space \mpc model, as argued in  \Cref{sec:highimple}.

For the correctness, \Cref{lem:wellandcorrect} shows that we obtain a valid solution of the compatibility tree that is procuded after the preprocessing. Then, \Cref{lem:gi} shows that the postprocessing phase transforms the latter solution correctly to a solution of the original \lcl on the actual input graph.

\section{The Broadcast Tree} \label{sec:broadcasttree}

A commonly used subroutine in the \mpc model is the broadcast (converge-cast, aggregation) tree. The \mpc broadcast tree is constant-depth $n^\epsilon$-ary tree structure. It enables broadcasting messages to all machines in constant time while respecting the $O(n^\delta)$ local memory and $O(m+n)$ global memory bounds. It is often assumed to exist without much discussion~\cite{GGC20, sirocco, Behnezhad19, broadcast}. Let us restate its formal definition for completeness. 

\begin{definition}[Aggregation Tree Structure, \cite{MPCbasictools}] \label{def:aggTreeStructure}
	Assume that an \mpc algorithm receives a collection of sets $A_1,\dots, A_k$ with elements from a totally ordered domain as input. In an aggregation tree structure for $A_1,\dots,A_k$, the elements of $A_1,\dots,A_k$ are stored in lexicographically sorted order (they are primarily sorted by the number $i \in \{1,\dots,k\}$ and within each set $A_i$ they are sorted increasingly). For each $i \in \{1,\dots,k\}$ such that the elements of $A_i$ appear on at least 2 different machines, there is a tree of constant depth containing the machines that store elements of $A_i$ as leafs and where each inner node of the tree has at most $\sqrt{S}$ children. The tree is structured such that it can be used as a search tree for the elements in $A_i$ (i.e., such that an in-order traversal of the tree visits the leaves in sorted order). Each inner node of these trees is handled by a separate additional machine. In addition, there is a constant-depth aggregation tree of degree at most $\sqrt{S}$ connecting all the machines that store elements of $A_1 ,\dots, A_k$.
\end{definition}

\clearpage
\phantomsection
\addcontentsline{toc}{section}{References}
\bibliographystyle{alphaurl}
\bibliography{mpc-landscape-arxiv}

\end{document}